\pdfoutput=1
\documentclass[11pt, a4paper, twoside,leqno]{amsart}
\usepackage[centering, totalwidth = 390pt, totalheight = 600pt]{geometry}
\usepackage{amssymb, amsmath, amsthm,mathtools,multicol}
\usepackage{microtype, stmaryrd, url, lmodern, eucal,bbold}
\usepackage[shortlabels]{enumitem}
\usepackage[latin1]{inputenc}
\usepackage{color}
\definecolor{darkgreen}{rgb}{0,0.45,0}
\usepackage[colorlinks,citecolor=darkgreen,linkcolor=darkgreen]{hyperref}
\usepackage[british]{babel}

\makeatletter
\def\@cite#1#2{[{#1\if@tempswa ,~#2\fi}]}% NEW
\makeatother
\DeclareMathAlphabet{\mathbf}{OT1}{cmr}{b}{n}
\DeclareFontFamily{U}{matha}{\hyphenchar\font45}
\DeclareFontShape{U}{matha}{m}{n}{
      <5> <6> <7> <8> <9> <10> gen * matha
      <10.95> matha10 <12> <14.4> <17.28> <20.74> <24.88> matha12
      }{}
\DeclareSymbolFont{matha}{U}{matha}{m}{n}
\DeclareMathSymbol{\gg}            {2}{matha}{"22}

\usepackage[arrow, matrix, tips, curve, graph, rotate]{xy}
\SelectTips{cm}{10}

\makeatletter
\def\matrixobject@{%
  \edef \next@{={\DirectionfromtheDirection@ }}%
  \expandafter \toks@ \next@ \plainxy@
  \let\xy@@ix@=\xyq@@toksix@
  \xyFN@ \OBJECT@}
\let\xy@entry@@norm=\entry@@norm
\def\entry@@norm@patched{%
  \let\object@=\matrixobject@
  \xy@entry@@norm }
\AtBeginDocument{\let\entry@@norm\entry@@norm@patched}
\makeatother

\newcommand{\twocong}[2][0.5]{\ar@{}[#2] \save ?(#1)*{\cong}\restore}
\newcommand{\twoeq}[2][0.5]{\ar@{}[#2] \save ?(#1)*{=}\restore}
\newcommand{\rtwocell}[3][0.5]{\ar@{}[#2] \ar@{=>}?(#1)+/l 0.2cm/;?(#1)+/r 0.2cm/^{#3}}
\newcommand{\ltwocell}[3][0.5]{\ar@{}[#2] \ar@{=>}?(#1)+/r 0.2cm/;?(#1)+/l 0.2cm/^{#3}}
\newcommand{\ltwocello}[3][0.5]{\ar@{}[#2] \ar@{=>}?(#1)+/r 0.2cm/;?(#1)+/l 0.2cm/_{#3}}
\newcommand{\dtwocell}[3][0.5]{\ar@{}[#2] \ar@{=>}?(#1)+/u  0.2cm/;?(#1)+/d 0.2cm/^{#3}}
\newcommand{\dltwocell}[3][0.5]{\ar@{}[#2] \ar@{=>}?(#1)+/ur  0.2cm/;?(#1)+/dl 0.2cm/^{#3}}
\newcommand{\drtwocell}[3][0.5]{\ar@{}[#2] \ar@{=>}?(#1)+/ul  0.2cm/;?(#1)+/dr 0.2cm/^{#3}}
\newcommand{\dthreecell}[3][0.5]{\ar@{}[#2] \ar@3{->}?(#1)+/u  0.2cm/;?(#1)+/d 0.2cm/^{#3}}
\newcommand{\utwocell}[3][0.5]{\ar@{}[#2] \ar@{=>}?(#1)+/d 0.2cm/;?(#1)+/u 0.2cm/_{#3}}
\newcommand{\dtwocelltarg}[3][0.5]{\ar@{}#2 \ar@{=>}?(#1)+/u  0.2cm/;?(#1)+/d 0.2cm/^{#3}}
\newcommand{\utwocelltarg}[3][0.5]{\ar@{}#2 \ar@{=>}?(#1)+/d  0.2cm/;?(#1)+/u 0.2cm/_{#3}}

\newdir{(}{{}*!<0em,-.14em>-\cir<.14em>{l^r}}
\newdir{ (}{{}*!/-5pt/\dir{(}}
\newdir{ >}{{}*!/-5pt/\dir{>}}

\newcommand{\cat}[1]{\mathrm{\mathcal #1}}
\newcommand{\thg}{{\mathord{\text{--}}}}

\DeclarePairedDelimiter{\abs}{|}{|}
\DeclarePairedDelimiter{\dbr}{\llbracket}{\rrbracket}
\DeclarePairedDelimiter{\spn}{\langle}{\rangle}

\newcommand{\res}[2]{\left.{#1}\right|_{#2}}

\newcommand{\defeq}{\mathrel{\mathop:}=}
\newcommand{\cd}[2][]{\vcenter{\hbox{\xymatrix#1{#2}}}}
\newcommand{\quot}{\delimiter"502F30E\mathopen{}}

\renewcommand{\phi}{\varphi}
\newcommand{\A}{{\mathcal A}}

\newcommand{\C}{{\mathcal C}}
\newcommand{\D}{{\mathcal D}}
\newcommand{\E}{{\mathcal E}}
\newcommand{\F}{{\mathcal F}}

\newcommand{\N}{{\mathcal N}}

\renewcommand{\P}{{\mathcal P}}

\let\sec=\S
\renewcommand{\S}{{\mathcal S}}
\newcommand{\T}{{\mathcal T}}

\newcommand{\X}{{\mathcal X}}

\newcommand{\xtor}[1]{\cd[@1]{{}
    \ar[r]|-{\object@{|}}^{#1} & {}}}

\makeatletter

\def\hookleftarrowfill@{\arrowfill@\leftarrow\relbar{\relbar\joinrel\rhook}}
\def\twoheadleftarrowfill@{\arrowfill@\twoheadleftarrow\relbar\relbar}
\def\leftbararrowfill@{\arrowdoublefill@{\leftarrow\mkern-5mu}\relbar\mapstochar\relbar\relbar}
\def\Leftbararrowfill@{\arrowdoublefill@{\Leftarrow\mkern-2mu}\Relbar\Mapstochar\Relbar\Relbar}
\def\leftringarrowfill@{\arrowdoublefill@{\leftarrow\mkern-3mu}\relbar{\mkern-3mu\circ\mkern-2mu}\relbar\relbar}
\def\lefttriarrowfill@{\arrowfill@{\mathrel\triangleleft\mkern0.5mu\joinrel\relbar}\relbar\relbar}
\def\Lefttriarrowfill@{\arrowfill@{\mathrel\triangleleft\mkern1mu\joinrel\Relbar}\Relbar\Relbar}

\def\hookrightarrowfill@{\arrowfill@{\lhook\joinrel\relbar}\relbar\rightarrow}
\def\twoheadrightarrowfill@{\arrowfill@\relbar\relbar\twoheadrightarrow}
\def\rightbararrowfill@{\arrowdoublefill@{\relbar\mkern-0.5mu}\relbar\mapstochar\relbar\rightarrow}
\def\Rightbararrowfill@{\arrowdoublefill@{\Relbar\mkern-2mu}\Relbar\Mapstochar\Relbar\Rightarrow}
\def\rightringarrowfill@{\arrowdoublefill@\relbar\relbar{\mkern-2mu\circ\mkern-3mu}\relbar{\mkern-3mu\rightarrow}}
\def\righttriarrowfill@{\arrowfill@\relbar\relbar{\relbar\joinrel\mkern0.5mu\mathrel\triangleright}}
\def\Righttriarrowfill@{\arrowfill@\Relbar\Relbar{\Relbar\joinrel\mkern1mu\mathrel\triangleright}}

\def\leftrightarrowfill@{\arrowfill@\leftarrow\relbar\rightarrow}
\def\mapstofill@{\arrowfill@{\mapstochar\relbar}\relbar\rightarrow}

\renewcommand*\xleftarrow[2][]{\ext@arrow 20{20}0\leftarrowfill@{#1}{#2}}
\providecommand*\xLeftarrow[2][]{\ext@arrow 60{22}0{\Leftarrowfill@}{#1}{#2}}
\providecommand*\xhookleftarrow[2][]{\ext@arrow 10{20}0\hookleftarrowfill@{#1}{#2}}
\providecommand*\xtwoheadleftarrow[2][]{\ext@arrow 60{20}0\twoheadleftarrowfill@{#1}{#2}}
\providecommand*\xleftbararrow[2][]{\ext@arrow 10{22}0\leftbararrowfill@{#1}{#2}}
\providecommand*\xLeftbararrow[2][]{\ext@arrow 50{24}0\Leftbararrowfill@{#1}{#2}}
\providecommand*\xleftringarrow[2][]{\ext@arrow 10{26}0\leftringarrowfill@{#1}{#2}}
\providecommand*\xlefttriarrow[2][]{\ext@arrow 80{24}0\lefttriarrowfill@{#1}{#2}}
\providecommand*\xLefttriarrow[2][]{\ext@arrow 80{24}0\Lefttriarrowfill@{#1}{#2}}

\renewcommand*\xrightarrow[2][]{\ext@arrow 01{20}0\rightarrowfill@{#1}{#2}}
\providecommand*\xRightarrow[2][]{\ext@arrow 04{22}0{\Rightarrowfill@}{#1}{#2}}
\providecommand*\xhookrightarrow[2][]{\ext@arrow 00{20}0\hookrightarrowfill@{#1}{#2}}
\providecommand*\xtwoheadrightarrow[2][]{\ext@arrow 03{20}0\twoheadrightarrowfill@{#1}{#2}}
\providecommand*\xrightbararrow[2][]{\ext@arrow 01{22}0\rightbararrowfill@{#1}{#2}}
\providecommand*\xRightbararrow[2][]{\ext@arrow 04{24}0\Rightbararrowfill@{#1}{#2}}
\providecommand*\xrightringarrow[2][]{\ext@arrow 01{26}0\rightringarrowfill@{#1}{#2}}
\providecommand*\xrighttriarrow[2][]{\ext@arrow 07{24}0\righttriarrowfill@{#1}{#2}}
\providecommand*\xRighttriarrow[2][]{\ext@arrow 07{24}0\Righttriarrowfill@{#1}{#2}}

\providecommand*\xmapsto[2][]{\ext@arrow 01{20}0\mapstofill@{#1}{#2}}
\providecommand*\xleftrightarrow[2][]{\ext@arrow 10{22}0\leftrightarrowfill@{#1}{#2}}
\providecommand*\xLeftrightarrow[2][]{\ext@arrow 10{27}0{\Leftrightarrowfill@}{#1}{#2}}

\makeatother

\numberwithin{equation}{section}

\theoremstyle{plain}
\newtheorem*{Thm*}{Theorem}
\newtheorem{Thm}{Theorem}
\newtheorem{Prop}[Thm]{Proposition}
\newtheorem{Cor}[Thm]{Corollary}
\newtheorem{Lemma}[Thm]{Lemma}

\theoremstyle{definition}
\newtheorem{Defn}[Thm]{Definition}
\newtheorem{Not}[Thm]{Notation}
\newtheorem{Ex}[Thm]{Example}

\newtheorem{Rk}[Thm]{Remark}

\newcommand{\alg}[1][X]{{\boldsymbol{#1}}}

\begin{document}
\leftmargini=2em \title[The costructure--cosemantics adjunction]{The
  costructure--cosemantics adjunction\\for comodels for computational
  effects} \author{Richard Garner} \address{Department of Mathematics
  \& Statistics,
  Macquarie University, NSW 2109, Australia}
\email{richard.garner@mq.edu.au}

\subjclass[2000]{Primary: }
\date{\today}

\thanks{The support of Australian Research Council grants
  FT160100393 and DP190102432 is gratefully acknowledged.}

\begin{abstract}
  It is well established that equational algebraic theories, and the
  monads they generate, can be used to encode computational effects.
  An important insight of Power and Shkaravska is that \emph{comodels}
  of an algebraic theory $\mathbb{T}$---i.e., models in the opposite
  category $\cat{Set}^\mathrm{op}$---provide a suitable environment
  for evaluating the computational effects encoded by $\mathbb{T}$. As
  already noted by Power and Shkaravska, taking comodels yields a
  functor from accessible monads to accessible comonads on
  $\cat{Set}$. In this paper, we show that this functor is part of an
  adjunction---the ``costructure--cosemantics adjunction'' of the
  title---and undertake a thorough investigation of its properties.

  We show that, on the one hand, the cosemantics functor takes its
  image in what we term the \emph{presheaf comonads} induced by small
  categories; and that, on the other, costructure takes its image in
  the \emph{presheaf monads} induced by small categories. In
  particular, the cosemantics comonad of an accessible monad will be
  induced by an explicitly-described category called its
  \emph{behaviour category} that encodes the static and dynamic
  properties of the comodels. Similarly, the costructure monad of an
  accessible comonad will be induced by a behaviour category encoding
  static and dynamic properties of the comonad coalgebras. We tie
  these results together by showing that the
  costructure--cosemantics adjunction is \emph{idempotent}, with
  fixpoints to either side given precisely by the presheaf monads and
  comonads. Along the way, we illustrate the value of our results with
  numerous examples drawn from computation and mathematics.
\end{abstract}
\maketitle
% \setcounter{tocdepth}{1}
% \tableofcontents
\section{Introduction}
\label{sec:introduction}

It is a well-known story that the category-theoretic approach to
computational effects originates with Moggi~\cite{Moggi1991Notions}.
Given a cartesian closed category $\C$, providing a denotational
semantics for a base notion of computation, this approach allows
additional language features such as input/output, interaction with
the store, or non-determinism---all falling under the general rubric
of ``effects''---to be encoded in terms of algebraic structure borne
by objects of $\C$.

In Moggi's treatment, this structure is specified via a strong monad
$\mathsf{T}$ on $\C$; a disadvantage of this approach is that, in
taking as primitive the objects $T(A)$ of \emph{computations} with
effects from $\mathsf{T}$ and values in $A$, it gives no indication of
how the effects involved are to be encoded as language features. An
important thread~\cite{Plotkin2001Adequacy,
  Plotkin2002Notions,Plotkin2003Algebraic} in John Power's work with
Gordon Plotkin has sought to rectify this, by identifying a
computational effect not with a strong monad \emph{per se}, but rather
with a set of (computationally meaningful) algebraic operations and
equations that \emph{generate} a strong monad, in the sense made
precise by Power and Kelly in~\cite{Kelly1993Adjunctions}.

The most elementary case of the above takes $\C = \cat{Set}$: then a
computational effect in Moggi's sense is simply a monad on
$\cat{Set}$, while an effect in the Plotkin--Power sense is an
equational \emph{algebraic theory}, involving a signature of (possibly
infinitary) operations and a set of equations between terms in the
signature. As the same monad may admit many different presentations,
the Plotkin--Power approach is more refined; but it is also slightly
narrower in scope, as not every monad on $\cat{Set}$ is engendered by
an algebraic theory, but only the \emph{accessible} ones---also called
\emph{monads with rank}. (A well-known inaccessible monad is the
continuation monad $V^{V^{(\thg)}}$).

The Plotkin--Power approach makes it particularly easy to define
\emph{models} in any category $\A$ with powers: they are $\A$-objects
endowed with interpretations of the given operations which satisfy the
given equations. When $\A = \cat{Set}$, such models are the same as
algebras for the associated monad; computationally, these can be
interpreted as sets of effectful computations which have been
identified up to a notion of equivalence which respects the effect
semantics.

It is another important insight of Power and
Shkaravska~\cite{Power2004From} that there is also a computational
interpretation of \emph{comodels}. A comodel in $\A$ is simply a model
in $\A^\mathrm{op}$, and again an important case is where
$\A = \cat{Set}$. The idea is that, given an algebraic theory
$\mathbb{T}$ for effects which interact with an
``environment'', a comodel of
$\mathbb{T}$ in $\cat{Set}$ provides the kind of environment with
which such programs interact. The underlying set $S$ of such a comodel
is the set of possible states of the environment; while each
generating $A$-ary operation $\sigma$ of the theory---which requests
an element of $A$ from the environment and binds it---is
co-interpreted as a function
$\dbr{\sigma} \colon S \rightarrow A \times S$ which answers the
request, and moves to a new state. While, in the first instance, we
co-interpret only the generating $\mathbb{T}$-operations, we can
extend this inductively to all $A$-ary computations $t \in T(A)$ of
the associated monad; whereupon we can see the co-interpretation
$\dbr{t} \colon S \rightarrow A \times S$ as a way of \emph{running}
(c.f.~\cite{Uustalu2015Stateful}) the computation $t \in T(A)$
starting from some state $s \in S$ to obtain a return value in $A$ and
a final state in $S$.

The computational perspective on comodels is powerful, and has
achieved significant traction in computer science; see, for
example,~\cite{Plotkin2008Tensors, Mogelberg2014Linear,
  Uustalu2015Stateful, Pattinson2015Sound, Pattinson2016Program,
  Ahman2020Runners}. However, our objective in this paper is to return
to the original~\cite{Power2004From} and settle some of the unanswered
questions posed there. Power and
Shkaravska observe that the category of comodels of any algebraic
theory is \emph{comonadic} over $\cat{Set}$, so that, if we choose to
identify algebraic theories with monads, then we have a process which
associates to each accessible monad on $\cat{Set}$ a certain comonad
on $\cat{Set}$. This leads them to ask:
\begin{quote}
  ``Does the construction of a comonad on $\cat{Set}$ from a monad
  with rank on $\cat{Set}$ yield an interesting relationship between
  monads and comonads?''
\end{quote}

We will answer this question in the affirmative, by showing that this
construction provides the right adjoint part of an adjunction
\begin{equation}\label{eq:51}
  \cd{
    {\cat{Mnd}_a(\cat{Set})^\mathrm{op}} \ar@<-4.5pt>[r]_-{\mathrm{Cosem}} \ar@{<-}@<4.5pt>[r]^-{\mathrm{Costr}} \ar@{}[r]|-{\bot} &
    {\cat{Cmd}_a(\cat{Set})}
  }
\end{equation}
between accessible monads and accessible comonads on $\cat{Set}$. The
corresponding left adjoint can be seen as taking a comonad
$\mathsf{Q}$ on $\cat{Set}$ to the largest algebraic theory for which
any $\mathsf{Q}$-coalgebra is a comodel.
Following~\cite{Lawvere1963Functorial, Dubuc1970Kan-extensions}, we
refer to the two functors in this adjunction as ``cosemantics'' and
``costructure''.

In fact, the mere existence of the adjunction~\eqref{eq:51} is not
hard to establish. Indeed, as we will see, its two directions were
already described in~\cite{Katsumata2019Interaction}, with our
costructure corresponding to their \emph{dual monad} of a comonad, and
our cosemantics being their \emph{Sweedler dual comonad} of a monad.
Our real contribution is that we do not seek merely to
construct~\eqref{eq:51}, but also to understand it thoroughly and
concretely.

In one direction, we will explicitly calculate the cosemantics
functor; this will, among other things, answer~\cite{Power2004From}'s
request that ``we should very much like to be able to characterise
those comonads, at least on $\cat{Set}$, that arise from categories of
comodels''. These comonads will be what we term
\emph{presheaf comonads}; for a small category $\mathbb{B}$, the
associated presheaf comonad $\mathsf{Q}_\mathbb{B}$ is that induced by
the adjunction to the left in:
\begin{equation*}
  \cd{
    {\cat{Set}^{\mathbb{B}}}
    \ar@<-4.5pt>[r]_-{\mathsf{res}_{J}}
    \ar@{<-}@<4.5pt>[r]^-{\mathsf{ran}_{J}} \ar@{}[r]|-{\top} &
    {{\cat{Set}^{\mathrm{ob}(\mathbb{B})}}} \ar@<-4.5pt>[r]_-{\Sigma} \ar@{<-}@<4.5pt>[r]^-{\Delta} \ar@{}[r]|-{\top} &
    {\cat{Set}} 
  } \qquad \qquad  Q_\mathbb{B}(A) = \sum_{b \in
    \mathbb{B}}\prod_{c \in \mathbb{B}} A^{\mathbb{B}(b,c)}\rlap{ ,}
\end{equation*}
with the explicit formula as to the right. These comonads are known
entities in computer science: they are precisely the interpretations
of \emph{directed containers} as introduced in~\cite{Ahman2012When},
and in~\cite{Ahman2017Taking} were termed \emph{dependently typed
  coupdate comonads}.

In fact, we do more than merely characterising the image of the
cosemantics functor: we prove for each accessible monad $\mathsf{T}$
on $\cat{Set}$ that its image under cosemantics is the presheaf
comonad of an \emph{explicitly} given category
$\mathbb{B}_\mathsf{T}$, which we term the \emph{behaviour category}
of $\mathsf{T}$. Since the category of Eilenberg--Moore
$\mathsf{Q}_{\mathbb{B}_\mathsf{T}}$-coalgebras is equivalent to the
functor category $[\mathbb{B}_\mathsf{T}, \cat{Set}]$, we may also
state this result as:

\begin{Thm*}
  Given an accessible monad $\mathsf{T}$ with behaviour category
  $\mathbb{B}_\mathsf{T}$, the category of $\mathsf{T}$-comodels 
  is equivalent to $[\mathbb{B}_\mathbb{T}, \cat{Set}]$ via an equivalence commuting with the forgetful
  functors to $\cat{Set}$.
\end{Thm*}

Our description of the behaviour category $\mathbb{B}_\mathsf{T}$ is
quite intuitive. Objects $\beta \in \mathbb{B}_\mathsf{T}$ are
elements of the final $\mathsf{T}$-comodel in $\cat{Set}$, which may
be described in many ways; we give a novel presentation as what we
term \emph{admissible behaviours} of $\mathsf{T}$. These comprise
families $(\beta_A \colon T(A) \rightarrow A)_{A \in \cat{Set}}$ of
functions, satisfying two axioms expressing that $\beta$ acts like a
state of a comodel in providing a uniform way of running
$\mathsf{T}$-computations to obtain values. As for morphisms of
$\mathbb{B}_\mathsf{T}$, these will be transitions between admissible
behaviours determined by \emph{$\mathsf{T}$-commands}, i.e., unary
operations $m \in T(1)$. We will see that maps with domain $\beta$ in
$\mathbb{B}_\mathsf{T}$ are $\mathsf{T}$-commands identified up to an
equivalence relation $\sim_\beta$ which identifies commands which act
in the same way on all states of behaviour $\beta$.

We also describe the action of the cosemantics functor on morphisms.
Thus, given a map of monads
$f \colon \mathsf{T}_1 \rightarrow \mathsf{T}_2$---which encodes an
\emph{interpretation} or \emph{compilation} of effects---we describe
the induced map of presheaf comonads
$\mathsf{Q}_{\mathbb{B}_{\mathsf{T}_2}} \rightarrow
\mathsf{Q}_{\mathbb{B}_{\mathsf{T}_1}}$. As explained
in~\cite{Ahman2017Taking}, maps of presheaf comonads do not correspond
to functors, but rather to to so-called
\emph{cofunctors}~\cite{Higgins1993Duality, Aguiar1997Internal} of the
corresponding categories, involving a mapping \emph{forwards} at the
level of objects, and mappings \emph{backwards} on morphisms. We are
able to give an explicit description of the cofunctor on behaviour
categories induced by a monad morphism, and on doing so we obtain our
second main result:

\begin{Thm*}
  The functor
  $\mathbb{B}_{(\thg)} \colon \cat{Mnd}_a(\cat{Set})^\mathrm{op}
  \rightarrow \cat{Cof}$ taking each accessible monad to its behaviour
  category, and each map of accessible monads to the induced cofunctor
  on behaviour categories yields a within-isomorphism factorisation
\begin{equation*}
  \cd[@C+0.6em@-0.3em]{
     & \cat{Cof} \ar[d]^-{\mathsf{Q}_{(\thg)}} \\
    \cat{Mnd}_a(\cat{Set})^\mathrm{op} \ar[r]^-{\mathrm{Cosem}}
    \ar@{-->}[ur]^-{\mathbb{B}_{(\thg)}} & \cat{Cmd}_a(\cat{Set})\rlap{ .}
  }
\end{equation*}
\end{Thm*}

For the other direction of the adjunction~\eqref{eq:51}, we
will in an analogous manner give an explicit
calculation of the image of the costructure functor. On objects, the monads in this image are what we term \emph{presheaf
monads}; here, for a small category $\mathbb{B}$, the presheaf
monad $\mathsf{T}_\mathbb{B}$ is that induced by the adjunction to the
left in:
\begin{equation*}
  \cd{
    {\cat{Set}^{\mathbb{B}^\mathrm{op}}}
    \ar@<-4.5pt>[r]_-{\mathsf{res}_{J}}
    \ar@{<-}@<4.5pt>[r]^-{\mathsf{lan}_{J}} \ar@{}[r]|-{\bot} &
    {{\cat{Set}^{\mathrm{ob}(\mathbb{B})}}} \ar@<-4.5pt>[r]_-{\Pi} \ar@{<-}@<4.5pt>[r]^-{\Delta} \ar@{}[r]|-{\bot} &
    {\cat{Set}}
  } \qquad \qquad 
  T_\mathbb{B}(A) = \prod_{b \in
    \mathbb{B}}\sum_{c \in \mathbb{B}} {\mathbb{B}(b,c)} \times A\rlap{ ,}
\end{equation*}
with the explicit formula as to the right (note the duality with the
comonad $\mathsf{Q}_\mathsf{B}$). Much like before, we will prove for
each accessible comonad $\mathsf{Q}$ on $\cat{Set}$ that its image
under costructure is of the form $\mathsf{T}_{\mathbb{B}_\mathsf{Q}}$
for an explicitly given ``behaviour category''
$\mathbb{B}_\mathsf{Q}$. The picture is perhaps less compelling in
this direction, but the objects of $\mathbb{B}_\mathsf{Q}$ can again
be described as ``behaviours'', by which we now mean elements of the
final $\mathsf{Q}$-coalgebra $Q(1)$; while morphisms of
$\mathbb{B}_\mathsf{Q}$ are uniform ways of transitioning between
$\mathsf{Q}$-behaviours. Like before, we also compute the
costructure functor on morphisms, and again find that each
comonad morphism induces a cofunctor between behaviour categories; so yielding our third main result:
\begin{Thm*}
    The functor
  $\mathbb{B}_{(\thg)} \colon \cat{Cmd}_a(\cat{Set})
  \rightarrow \cat{Cof}$ taking an accessible comonad to its
  behaviour category, and a map of accessible comonads
  to the induced cofunctor on behaviour categories yields a within-isomorphism~factorisation 
\begin{equation*}
  \cd[@C+0.6em@-0.3em]{
    & \cat{Cof} \ar[d]^-{\mathsf{Q}_{(\thg)}} \\
    \cat{Cmd}_a(\cat{Set}) \ar[r]^-{\mathrm{Costr}}
    \ar@{-->}[ur]^-{\mathbb{B}_{(\thg)}} & \cat{Mnd}_a(\cat{Set})^\mathrm{op}\rlap{ .}
  }
\end{equation*}
\end{Thm*}

It remains only to understand how costructure and cosemantics interact
with each other. The crucial observation is that~\eqref{eq:51} is an
example of a so-called \emph{idempotent} (or \emph{Galois})
adjunction. Here, an adjunction $F \dashv G \colon \D \rightarrow \C$
is \emph{idempotent} if any application of $F$ yields a \emph{fixpoint
  to the left}, i.e., an object of $\D$ at which the adjunction counit
is invertible, while any application of $G$ yields a fixpoint to the
right, i.e., an object of $\C$ at which the adjunction unit is
invertible. In these terms, our final main result can be stated as:
\begin{Thm*}
  The costructure--cosemantics adjunction~\eqref{eq:51} is idempotent.
  Its fixpoints to the left and the right are the presheaf monads and
  presheaf comonads.
\end{Thm*}

Let us note that the results of this paper are only the first step in
a larger investigation. On the one hand, to deal with recursion, we
will require a comprehensive understanding of \emph{enriched} versions
of the costructure--cosemantics adjunction. On the other hand, even in
the unenriched world of equational algebraic theories, we may be
interested in understanding costructure and cosemantics for comodels in other categories than $\cat{Set}$: for
example, \emph{topological} comodels, which encode information not
only about behaviours of states, but also about finitistic,
computable observations of such behaviour.  We hope to
pursue these avenues in future work.

We conclude this introduction with a brief overview of the contents of
the paper. We begin in Section~\ref{sec:algebr-theor-effects} with
background material on algebraic theories, their models and comodels,
and the relation to monads on $\cat{Set}$, along with relevant
examples relating to computational effects. In
Section~\ref{sec:presh-monads-comon}, we prepare the ground for our
main results by investigating the classes of presheaf monads and
comonads. Then in Section~\ref{sec:monad-comon-adjunct}, we give the
construction of the costructure--cosemantics adjunction~\eqref{eq:51},
and explain how its two directions encapsulate constructions
of~\cite{Katsumata2019Interaction}.

In Section~\ref{sec:char-image-cosem}, we calcuate the values of the
cosemantics functor. We begin with a general category-theoretic
argument that shows that its must take values in presheaf comonads; we
then give a concrete calculation of the presheaf comonad associated to
a given accessible monad, in other words, to the calculation of the
behaviour category of the given monad. We also describe the cofunctors
between behaviour categories induced by monad morphisms.

In Section~\ref{sec:calc-costr-funct}, we turn to the costructure
functor, showing by a direct calculation that it sends each accessible
comonad to the presheaf monad of an appropriate behaviour category. As
before, we also describe the cofunctor on behaviour categories induced
by each map of accessible comonads. Then in
Section~\ref{sec:idemp-monad-comon}, we tie these results together by
exhibiting the idempotency of the costructure--cosemantics adjunction,
and characterising its fixpoints as the presheaf monads and comonads.
Finally, Sections~\ref{sec:cosem-exampl-appl}
and~\ref{sec:costr-exampl} are devoted to examples of behaviour
categories and cofunctors calculated using our main results.

\section{Algebraic theories and their (co)models}
\label{sec:algebr-theor-effects}

\subsection{Algebraic theories}
\label{sec:algebraic-theories}

In this background section, we recall the definition of (possibly infinitary)
algebraic theory; the notions of model and comodel in any suitable
category; and the relation to monads on $\cat{Set}$. We also recall
the applications of these notions in the study of computational
effects.

\begin{Defn}[Algebraic theory]
  \label{def:1}
  A \emph{signature} comprises a set $\Sigma$ of \emph{function
    symbols}, and for each $\sigma \in \Sigma$ a set $\abs \sigma$,
  its \emph{arity}.   Given a signature $\Sigma$ and a set $A$, we define the set
  $\Sigma(A)$ of \emph{$\Sigma$-terms with variables in
    $A$} by the inductive clauses
  \begin{equation*}
    a \in A \implies a \in \Sigma(A)
    \quad \text{and} \quad
    \sigma \in \Sigma\text{, }t \in \Sigma(A)^{\abs \sigma} \implies
    \sigma(t)
    \in \Sigma(A)\rlap{ .}
  \end{equation*}
  An \emph{equation} over a signature $\Sigma$ is a triple $(A,t,u)$
  with $A$ a set and $t,u \in \Sigma(A)$.  An \emph{algebraic theory} $\mathbb{T}$ is a
  signature $\Sigma$ and a set $\E$ of equations over it.
\end{Defn}

\begin{Defn}[$\mathbb{T}$-terms]
  \label{def:8}
  Given a signature $\Sigma$ and terms $t \in \Sigma(A)$ and
  $u \in \Sigma(B)^A$, we define the \emph{substitution}
  $t(u) \in \Sigma(B)$ recursively by:
  \begin{equation*}
    \quad a \in A \implies a(u) = u_a \quad \text{and} \quad 
    \sigma \in
    \Sigma\text{, } t \in \Sigma(A)^{\abs \sigma} \implies
    (\sigma(t))(u) = \sigma(\lambda i.\,t_i(u))\rlap{ .}
  \end{equation*}
  Given an algebraic theory $\mathbb{T} = (\Sigma, \E)$ and a set $B$,
  we define \emph{$\mathbb{T}$-equivalence} to be the smallest
  equivalence relation $\equiv_\mathbb{T}$ on $\Sigma(B)$ such that:
  \begin{enumerate}[(i)]
  \item If $(A, t, u) \in \E$ and $v \in \Sigma(B)^A$, then
    $t(v) \equiv_\mathbb{T} u(v)$;
  \item If $\sigma \in \Sigma$ and $t_i \equiv_\mathbb{T} u_i$ for all
    $i \in \abs \sigma$, then $\sigma(t) \equiv_\mathbb{T} \sigma(u)$.
  \end{enumerate}
  The set $T(A)$ of \emph{$\mathbb{T}$-terms with variables in $A$} is
  the quotient $\Sigma(A) \quot \equiv_\mathbb{T}$.
\end{Defn}

When an algebraic theory $\mathbb{T}$ is thought of as specifying a
computational effect, we think of $T(A)$ as giving the set of
computations with effects from $\mathbb{T}$ and returning a value in
$A$. The following standard examples illustrate this.

\begin{Ex}[Input]
  \label{ex:6} Given a set $V$, the theory of \emph{$V$-valued input}
  comprises a single $V$-ary function symbol $\mathsf{read}$,
  satisfying no equations, whose action we think of as:
  \begin{equation*}
    (t \colon V \rightarrow X) \mapsto 
    \text{\textsf{let read$()$ be $v$.\,$t(v)$}}\rlap{ .}
  \end{equation*}
  For this theory, terms $t \in T(A)$ are computations which can
  request $V$-values from an external source and use them to determine
  a return value in $A$. For example, when $V = \mathbb{N}$, the
  program which requests two input values and returns their sum is
  encoded by
  \begin{equation}\label{eq:15}
    \text{\textsf{let read$()$ be $n$.\,let read$()$ be $m$.\,$n+m$}}
    \ \in \ T(\mathbb{N})\rlap{ .}
  \end{equation}
\end{Ex}

For an algebraic theory \emph{qua} computational effect, it is
idiomatic that its function symbols are read in continuation-passing
style; so the domain of a function symbol $X^V \rightarrow X$ is a
scope in which an element of $V$ is available to determine a
continuation, and applying the operation binds this element to yield a
continuation \emph{simpliciter}.

\begin{Ex}[Output]
  \label{ex:9}
  Given a set $V$, the theory of \textit{$V$-valued output} comprises
  an $V$-indexed family of unary function symbols
  $(\mathsf{write}_v : v \in V)$ subject to no equations. In the
  continuation-passing style, we denote the action of
  $\mathsf{write}_v$ by
  \begin{equation*}
    t  \mapsto 
    \textsf{let write$(v)$ be $x$.\,$t$} \qquad \text{or, more simply,} \qquad
    t  \mapsto 
    \text{\textsf{write}$(v)$\textsf{;} }t\rlap{ .}
  \end{equation*}
\end{Ex}

\begin{Ex}[Read-only state]
  \label{ex:7}
  Given a set $V$, the theory of \emph{$V$-valued read-only state} has
  a single $V$-ary operation $\mathsf{get}$, satisfying the equations
  \begin{equation}\label{eq:23}
    \mathsf{get}(\lambda v.\, x) \equiv x \qquad \text{and} \qquad
    \mathsf{get}(\lambda v.\, \mathsf{get}(\lambda w.\,
    x_{vw})) \equiv \mathsf{get}(\lambda v.\, x_{vv})\rlap{ .}
  \end{equation}
  These equations express that reading from read-only state should not
  change that state, and that repeatedly reading the state should
  always yield the same answer. In universal algebra, these are the
  so-called ``rectangular band'' identities; while in another
  nomenclature they express that $\mathsf{get}$ is \emph{copyable} and
  \emph{discardable}; see~\cite{Thielecke1997Categorical}.
\end{Ex}
  
\begin{Ex}[State,~\cite{Plotkin2002Notions}]
  \label{ex:10}
  Given a set $V$, the theory of \emph{$V$-valued state} comprises an
  $V$-ary operation $\mathsf{get}$ and a $V$-indexed family of unary
  operations $\mathsf{put}_v$, subject to the following equations:
  \begin{equation*}
    \mathsf{get}(\lambda v.\, \mathsf{put}_v (x)) \equiv x \quad
    \mathsf{put}_u(\mathsf{put}_{v}(x)) \equiv \mathsf{put}_{v}(x)
    \quad \mathsf{put}_u(\mathsf{get}(\lambda v.\, x_v)) \equiv
    \mathsf{put}_u(x_u)\rlap{ .}
  \end{equation*}
  Read in continuation-passing style, these axioms capture the
  semantics of reading and updating a store containing an element of
  $V$.
\end{Ex}

We now describe the appropriate morphisms between algebraic theories.

\begin{Defn}[Category of algebraic theories]
  \label{def:24}
  Let $\mathsf{T}_1 = (\Sigma_1, \E_1)$ and
  $\mathsf{T}_2 = (\Sigma_2, \E_2)$ be algebraic theories. An
  \emph{interpretation
    $f \colon \mathbb{T}_1 \rightarrow \mathbb{T}_2$} is given by
  specifying, for each $\sigma \in \Sigma_1$, a term
  $\sigma^f \in \Sigma_2(\abs \sigma)$ such that, on defining for each
  $t \in \Sigma_1(A)$ the term $t^f \in \Sigma_2(A)$ by the recursive
  clauses
  \begin{equation}\label{eq:2}
    a \in A \Rightarrow a^f = a \quad \text{and} \quad
    \sigma \in \Sigma\text{, }t \in \Sigma(A)^{\abs \sigma} \Rightarrow
    (\sigma(t))^f = \sigma^f(\lambda a.\, (t_a)^f)\rlap{ ,}
  \end{equation}
  we have that $t^f \equiv_{\mathbb{T}_2} u^f$ for all
  $(A,t,u) \in \E_1$. With the obvious composition, we obtain a
  category $\cat{Alg\T h}$ of algebraic theories and interpretations.
\end{Defn}

An interpretation $\mathbb{T}_1 \rightarrow \mathbb{T}_2$ between
theories can be understood as a way of translating computations with
effects from $\mathbb{T}_1$ into ones with effects from
$\mathbb{T}_2$.

\begin{Ex}
  \label{ex:4}
  Let $h \colon V \rightarrow W$ be a function between sets, let
  $\mathbb{T}_1$ be the theory of $V$-valued output, and let
  $\mathbb{T}_2$ be the theory of $W$-valued state. We have an
  interpretation $f \colon \mathbb{T}_1 \rightarrow \mathbb{T}_2$
  defined by $\mathsf{write}^f = \mathsf{put}_{h(v)}$.
\end{Ex}

\begin{Ex}
  \label{ex:5}
  Let $h \colon W \rightarrow V$ be a function between sets, let
  $\mathbb{T}_1$ be the theory of $V$-valued read-only state, and let
  $\mathbb{T}_2$ be the theory of $W$-valued state. We have an
  interpretation $f \colon \mathbb{T}_1 \rightarrow \mathbb{T}_2$
  defined by $\mathsf{get}^f = \mathsf{get}(\lambda w.\, h(w))$.
\end{Ex}

\subsection{Models and comodels}
\label{sec:models-comodels}

We now describe the notion of \emph{model} of an algebraic theory in a
suitable category. Recall that a category $\C$ has \emph{powers} if,
for every $X \in \C$ and set $A$, the $A$-fold self-product
$(\pi_a \colon X^A \rightarrow X)_{a \in A}$ exists in $\C$.

\begin{Defn}[$\Sigma$-structure]
  \label{def:17}
  Let $\Sigma$ be a signature. A \emph{$\Sigma$-structure
    $\alg = (X,\dbr{\thg}_{\alg})$} in a category $\C$ with powers
  comprises an \emph{underlying object} $X \in \C$ and
  \emph{operations}
  $\dbr {\sigma}_{\alg} \colon X^{\abs \sigma} \rightarrow X$ for each
  $\sigma \in \Sigma$. Given a $\Sigma$-structure
  $\alg \in \C^\Sigma$, we define for each $t \in \Sigma(A)$ the
  \emph{derived operation} $\dbr{t}_{\alg} \colon X^A \rightarrow X$
  by the following recursive clauses:
  \begin{equation}\label{eq:24}
    \dbr{a}_{\alg} = \pi_a \qquad \text{and} \qquad
    \smash{\dbr{\sigma(t)}_{\alg} = X^A
      \xrightarrow{(\dbr{t_i}_{\alg})_{i \in \abs \sigma}} X^{\abs
        \sigma} \xrightarrow{\dbr{\sigma}_{\alg}} X}\rlap{ .}
  \end{equation}
\end{Defn}

\begin{Defn}[$\mathbb{T}$-model]
  \label{def:18}
  Let $\mathbb{T} = (\Sigma, \E)$ be an algebraic theory. A
  \emph{$\mathbb{T}$-model} in a category with powers $\C$ is a
  $\Sigma$-structure $\alg$ such that
  $\dbr{t}_{\alg} = \dbr{u}_{\alg} \colon X^A \rightarrow X$ for all
  $(A,t,u) \in \E$. We write $\C^\mathbb{T}$ for the category whose
  objects are $\mathbb{T}$-models in $\C$, and whose maps
  $\alg \rightarrow \alg[Y]$ are $\C$-maps $f \colon X \rightarrow Y$
  such that
  $\dbr{\sigma}_{\alg[Y]} \circ f^{\abs{\sigma}} = f \circ
  \dbr{\sigma}_{\alg}$ for all $\sigma \in \Sigma$. We write
  $U^\mathbb{T} \colon \C^\mathbb{T} \rightarrow \C$ for the obvious
  forgetful functor.
\end{Defn}

Dual to the notion of model is the notion of \emph{comodel}. Recall
that a category $\C$ has \emph{copowers} if every $A$-fold
self-coproduct $(\nu_a \colon X \rightarrow A \cdot X)_{a \in A}$
exists in $\C$.

\begin{Defn}[Comodel]
  \label{def:3}
  Let $\mathbb{T}$ be an algebraic theory. A
  \emph{$\mathbb{T}$-comodel} in a category $\C$ with copowers is a
  model of $\mathbb{T}$ in $\C^\mathrm{op}$; it thus comprises an
  object $X \in \C$ and ``co-operations''
  $\dbr{\sigma}_{\alg} \colon X \rightarrow \abs \sigma \cdot X$,
  subject to the equations of $\mathbb{T}$. We write ${}^\mathbb{T}\C$
  for the category of $\mathbb{T}$-comodels in $\C$ and
  ${}^\mathbb{T}U \colon {}^\mathbb{T}\C \rightarrow \C$ for the
  forgetful functor.
\end{Defn}

If we say simply ``model'' or ``comodel'', we will by default mean
model or comodel in $\cat{Set}$. As explained in~\cite{Power2004From,
  Plotkin2008Tensors}, set-based comodels provide deterministic
environments suitable for evaluating computations with effects from
$\mathbb{T}$.

\begin{Ex}
  \label{ex:38} A comodel $\alg[S]$ of the theory of $V$-valued input
  is a state machine which responds to requests for $V$-characters; it
  comprises a set of states $S$ and a map
  $\dbr{\mathsf{read}}_{\alg[S]} = (g,n) \colon S \rightarrow V \times
  S$ assigning to each $s \in S$ a character $g(s) \in V$ to be read
  and a new state $n(s) \in S$.
\end{Ex}

\begin{Ex}
  \label{ex:37}
  A comodel ${\alg[S]}$ of the theory of $V$-valued output is a state
  machine which changes its state in response to $V$-characters: it
  comprises a set of states $S$ and maps
  $\dbr{\mathsf{write}_v}_{{\alg[S]}} \colon S \rightarrow S$ for each
  $v \in V$, or equally, a single map
  $p \colon V \times S \rightarrow S$, providing for each character
  $v \in V$ and state $s \in S$ a new state $p(v,s) \in S$.
\end{Ex}
\begin{Ex}
  \label{ex:1}
    A comodel $\alg[S]$ of the theory of $V$-valued read-only state comprises a set $S$
    together with a function
    $\dbr{\mathsf{get}}_{\alg[S]} = (g,n) \colon S \rightarrow V \times S$ satisfying
    \begin{equation}\label{eq:22}
      n(s) = s \qquad \text{and} \qquad \bigl(g(s), g(n(s)), n(n(s))\bigr) = \bigl(g(s),g(s),n(s)\bigr)\rlap{ .}
    \end{equation}
    The first axiom allows us to ignore $n$, and moreover implies the
    second axiom; whence a comodel amounts to nothing more than a set
    $S$ and a function $g \colon S \rightarrow V$. 
\end{Ex}

\begin{Ex}
  \label{ex:36}
  A comodel of the theory of $V$-valued state involves a set of states
  $S$ together with maps $(g,n) \colon S \rightarrow V \times S$ and
  $p \colon S \times V \rightarrow S$ rendering commutative the
  diagrams
  \begin{equation*}
    \cd[@C-1em]{
      & V \times S \ar@{<-}[d]^-{(g,n)} & 
      V \times S \ar@{<-}[d]_-{V \times p} \ar@{<-}[rr]^-{(g,n)} & & V \ar@{<-}[d]^-{p} &
      S \ar@{<-}[d]_-{p} \ar@{<-}[r]^-{p} & V \times S\rlap{ .}
      \ar@{<-}[d]_-{V \times p} &
      \\
      S \ar@{<-}[ur]^-{p} \ar@{=}[r]^-{} & V &
      V \times (V \times S) \ar@{<-}[r]^-{\cong} & (V \times V) \times S
      \ar@{<-}[r]^-{\Delta \times S} & V \times S &
      V \times S \ar@{<-}[r]^-{V \times \pi_1} & V \times (V \times S)
    }
  \end{equation*}
  These are equally the conditions that
  \begin{equation*}
    p(g(s),n(s)) = s\text{,} \ \ n(p(v,s)) = p(v,s)\text{,} \ \ 
    g(p(v,s)) = v \ \ \text{and} \ \ p(v',p(v,s)) = p(v',s)\rlap{ ,}
  \end{equation*}
  the first two of which imply $n = \mathrm{id}_S$. Thus, to give the
  comodel is to give the set $S$ together with maps
  $g \colon S \rightarrow V$ and $p \colon V \times S \rightarrow S$
  satisfying the axioms
  \begin{equation*}
    p(g(s),s) = s \qquad g(p(v,s)) = v \quad \text{and} \quad p(v',p(v,s)) = p(v',s)\rlap{ ;}
  \end{equation*}
  as noted in~\cite{Uustalu2015Stateful}, this is precisely a (very
  well-behaved) \emph{lens} in the sense
  of~\cite{Foster2007Combinators}.
\end{Ex}

A $\mathbb{T}$-comodel allows us to evaluate $\mathbb{T}$-computations
by way of the derived operations of Definition~\ref{def:17}. Indeed,
given a comodel $\alg[S] \in {}^\mathbb{T} \cat{Set}$ and a term
$t \in T(A)$, the co-operation
$\dbr{t} \colon S \rightarrow A \times S$ is defined by the recursive
clauses
\begin{equation}\label{eq:19}
\begin{aligned}
  a \in A &\implies \dbr{a} (s) = (a,s)  \qquad \ \ \,
  \\ \text{and }\sigma 
  \in \Sigma, t \in T(A)^{\abs \sigma} &\implies
  \dbr{\sigma(t)}(s)= \dbr{t_i}(s') \text{ where }
  \dbr{\sigma}(s) = (i,s')\rlap{ .}
\end{aligned}
\end{equation}

(Here, and henceforth, we drop the subscript $\alg[S]$ where confusion
seems unlikely.) The semantics of this is clear: given a
$\mathbb{T}$-computation $\sigma(t) \in T(A)$ and starting state $s$,
we respond to the request for a $\abs \sigma$-element posed by the
outermost operation symbol of $\sigma(t)$ by evaluating
$\dbr{\sigma}(s)$ to obtain $i \in \abs \sigma$ along with a new state
$s'$; substituting this $i$ into $t$ yields the simpler computation
$t_i$ which we now run from state $s'$. When we hit a value $a \in A$
we return it along with the final state reached.

\begin{Ex}
  \label{ex:39}
  Consider the theory of $\mathbb{N}$-valued input, and the term
  from~\eqref{eq:15}, which we may equally write as
  $t = \mathsf{read}(\lambda n.\, \mathsf{read}(\lambda m.\, n+m)) \in
  T(\mathbb{N})$. If $\alg[S]$ is the comodel with $S = \{s,s',s''\}$
  and $\dbr{\mathsf{read}} \colon S \rightarrow \mathbb{N} \times S$
  given as to the left in:
  \begin{align*}
    s & \mapsto (7,s') & s &\mapsto (18,s'') \\
    s' &\mapsto (11,s'') & s' & \mapsto (24,s'') \\
    s'' &\mapsto (13,s'') &
    s'' & \mapsto (26,s'')\rlap{ ,}
  \end{align*}
  then $\dbr{t} \colon S \rightarrow \mathbb{N} \times S$ is given as
  to the right. For example, for $\dbr{t}(s)$ we calculate that
  $\dbr{\mathsf{read}(\lambda n.\, \mathsf{read}(\lambda m.\, n+m))}(s)
  = \dbr{\mathsf{read}(\lambda m.\, 7+m)}(s') = \dbr{7+11}(s'') =
  (18,s'')$.
\end{Ex}

We conclude our discussion of models and comodels by describing the
functoriality of the assignment $\mathbb{T} \mapsto \C^\mathbb{T}$.

\begin{Defn}[Semantics and cosemantics]
  \label{def:47}
  For a category $\C$ with powers, the \emph{semantics functor}
  $\cat{Sem}_\C \colon \cat{Alg\T h}^\mathrm{op} \rightarrow \cat{CAT}
  / \C$ is given by
  $\mathbb{T} \mapsto (U^{\mathbb{T}} \colon \C^\mathbb{T} \rightarrow
  \C)$ on objects; while on maps, an interpretation
  $f \colon \mathbb{T}_1 \rightarrow \mathbb{T}_2$ is taken to the
  functor
  $f^\ast \colon \C^{\mathbb{T}_2} \rightarrow \C^{\mathbb{T}_1}$ over
  $\C$ acting via
  $(X, \dbr{\thg}_{\alg}) \mapsto (X, \dbr{(\thg)^f}_{\alg})$.

  Dually, for any category $\C$ with copowers, we define the
  \emph{cosemantics functor}
  $\mathrm{Cosem}_\C \colon \cat{Alg\T h}^\mathrm{op} \rightarrow
  \cat{CAT} / \C$ by
  $\mathbb{T} \mapsto ({}^\mathbb{T} U \colon {}^\mathbb{T} \C
  \rightarrow \C)$ on objects, and on morphisms in the same manner as
  above; more formally, we have
  $\mathrm{Cosem}_\C =
  \mathrm{Sem}_{\C^\mathrm{op}}(\thg)^\mathrm{op}$.
\end{Defn}

The functoriality of (co)semantics implies that theories $\mathbb{T}$ and
$\mathbb{T}'$ which are isomorphic in $\cat{Alg\T h}$ have the same (co)models in any
category with (co)powers $\C$; we call such theories \emph{Morita equivalent}.

\begin{Ex}
  \label{ex:8}
  Let $h \colon V \rightarrow W$ be a function, and let
  $f \colon \mathbb{T}_1 \rightarrow \mathbb{T}_2$ be the
  interpretation of $V$-valued output into $W$-valued state of
  Example~\ref{ex:4}. For each comodel
  $\alg[S] = (S, g \colon S \rightarrow W, p \colon S \times W
  \rightarrow S)$ of $W$-valued state, the associated comodel
  $f^\ast\alg[S]$ of $V$-valued output is
  $(S, p \circ (1 \times h) \colon S \times V \rightarrow S)$.
\end{Ex}

\begin{Ex}
  \label{ex:12}
  Let $h \colon W \rightarrow V$ be a function, and let
  $f \colon \mathbb{T}_1 \rightarrow \mathbb{T}_2$ be the
  interpretation of $V$-valued read-only state into $W$-valued state
  of Example~\ref{ex:5}. For each comodel
  $\alg[S] = (S, g \colon S \rightarrow W, p \colon S \times W
  \rightarrow S)$ of $W$-valued state, the associated comodel
  $f^\ast\alg[S]$ of $V$-valued read-only state is
  $(S, hg \colon S \rightarrow V)$.
\end{Ex}

\subsection{The associated monad}
\label{sec:associated-monad}

Finally in this section, we recall how an algebraic theory gives rise
to a monad on $\cat{Set}$, and the manner in which this interacts with
semantics. We specify our monads as Kleisli triples in the style
of~\cite[Exercise~1.3.12]{Manes1976Algebraic}.

\begin{Defn}[Associated monad]
  \label{def:15}
  The \emph{associated monad} $\mathsf{T}$ of an algebraic theory
  $\mathbb{T}$ % is the monad induced by the adjunction
  % $F^\mathbb{T} \dashv U^\mathbb{T} \colon \cat{Set}^\mathbb{T}
  % \rightarrow \cat{Set}$. Thus, $\mathsf{T}$
  has action on objects $A \mapsto T(A)$; unit maps
  $\eta_A \colon A \rightarrow T(A)$ given by inclusion of variables;
  and Kleisli extension $u^\dagger \colon T(A) \rightarrow T(B)$ of
  $u \colon A \rightarrow T(B)$ given by $t \mapsto t(u)$. The
  assignment $\mathbb{T} \mapsto \mathsf{T}$ is the action on objects
  of the \emph{associated monad functor}
  $\mathrm{Ass} \colon \cat{Alg\T h} \rightarrow
  \cat{Mnd}(\cat{Set})$, which on morphisms takes an interpretation
  $f \colon \mathbb{T}_1 \rightarrow \mathbb{T}_2$ to the monad
  morphism $\mathsf{T}_1 \rightarrow \mathsf{T}_2$ whose components
  $T_1(A) \rightarrow T_2(A)$ are the assignments $t \mapsto t^f$
  defined as in~\eqref{eq:2}.
\end{Defn}

\begin{Prop}
  \label{prop:1}
  The associated monad functor
  $\cat{Alg\T h} \rightarrow \cat{Mnd}(\cat{Set})$ is full and
  faithful, and a monad $\mathsf{T}$ is in its essential image just
  when it is accessible.\qed
\end{Prop}

Here, a monad on $\cat{Set}$ is \emph{accessible} if its underlying
endofunctor is accessible, in the sense of being a small colimit of
representable functors. There are well-known monads on $\cat{Set}$
which are not accessible, for example the power-set monad $\mathbb{P}$
and the continuation monad $V^{V^{(\thg)}}$; nonetheless, we may treat
any monad $\mathsf{T}$ on $\cat{Set}$ ``as if it were induced by an
algebraic theory'' by adopting the following conventions: if
$a \in A$, then we may write $a \in T(A)$ in place of
$\eta_A(a) \in T(A)$, and if $t \in T(A)$ and $u \in T(B)^A$, then we
may write $t(u)$ in place of $u^\dagger(t)$.

We now discuss how the model and comodel semantics of an algebraic
theory can be expressed in terms of the associated monad.

\begin{Defn}[$\mathsf{T}$-models and comodels]
  \label{def:20}
  Let $\mathsf{T}$ be a monad on $\cat{Set}$ and let $\C$ be a
  category with powers. A \emph{$\mathsf{T}$-model $\alg$ in $\C$} is
  an object $X \in \C$ together with operations
  $\dbr{t}_{\alg} \colon X^A \rightarrow X$ for every set $A$ and
  $t \in T(A)$, subject to the axioms
  \begin{equation}\label{eq:8}
    \dbr{a}_{\alg} = \pi_a \qquad \text{and} \qquad
    \smash{\dbr{t(u)}_{\alg} = X^B
      \xrightarrow{(\dbr{u_a}_{\alg})_{a \in A}} X^{A} \xrightarrow{\dbr{t}_{\alg}} X}\rlap{ .}
  \end{equation}
  for all $a \in A$ and all $t \in T(A)$ and $u \in T(B)^A$. We write
  $\C^\mathsf{T}$ for the category of $\mathsf{T}$-models in $\C$,
  whose maps $ \alg[X] \rightarrow \alg[Y]$ are $\C$-maps
  $f \colon X \rightarrow Y$ with
  $\dbr{t}_{\alg[Y]} \circ f^A = f \circ \dbr{t}_{\alg}$ for all sets
  $A$ and all $t \in T(A)$; we write
  $U^\mathsf{T} \colon \C^\mathsf{T} \rightarrow \C$ for the forgetful
  functor.

  If $\C$ is a category with copowers then a $\mathsf{T}$-comodel in
  $\C$ is a $\mathsf{T}$-model in $\C^\mathrm{op}$, involving
  co-operations $\dbr{t}_{\alg} \colon X \rightarrow A \cdot X$ subject
  to the dual axioms
  \begin{equation}\label{eq:21}
    \dbr{a}_{\alg} = \nu_a \qquad \text{and} \qquad
    \smash{\dbr{t(u)}_{\alg} = X
      \xrightarrow{\dbr{t}_{\alg}} A \cdot X
      \xrightarrow{\spn{\dbr{u_a}_{\alg}}_{a \in A}} B \cdot X}\rlap{ .}
  \end{equation}
  We write ${}^\mathsf{T} U \colon {}^\mathsf{T} \C \rightarrow \C$ for
  the forgetful functor from the category of $\mathsf{T}$-comodels.
\end{Defn}

\begin{Defn}[Semantics and cosemantics]
  \label{def:4}
  For any category $\C$ with powers, the \emph{semantics functor}
  $\mathrm{Sem}_\C \colon \cat{Mnd}(\cat{Set})^\mathrm{op} \rightarrow
  \cat{CAT} / \C$ is given by
  $\mathsf{T} \mapsto (U^{\mathsf{T}} \colon \C^\mathsf{T} \rightarrow
  \C)$ on objects; while a monad morphism
  $f \colon \mathsf{T}_1 \rightarrow \mathsf{T}_2$ is taken to the
  functor
  $f^\ast \colon \C^{\mathsf{T}_2} \rightarrow \C^{\mathsf{T}_1}$ over
  $\C$ acting via
  $(X, \dbr{\thg}_{\alg}) \mapsto (X, \dbr{f(\thg)}_{\alg})$. For a
  category $\C$ with copowers, we define the \emph{cosemantics
    functor} by
  $\mathrm{Cosem}_\C \defeq
  \mathrm{Sem}_{\C^\mathrm{op}}(\thg)^\mathrm{op} \colon
  \cat{Mnd}(\cat{Set})^\mathrm{op} \rightarrow \cat{CAT} /
  \C$.% is the composite
\end{Defn}

The following result, which again is entirely standard, tells us that
we lose no semantic information in passing from an algebraic theory to
the associated monad.

\begin{Prop}
  \label{prop:37}
  For any category $\C$ with powers (respectively, copowers), the
  triangle to the left (respectively, right) below commutes to within
  natural isomorphism:
  \begin{equation*}
    \cd[@C-2.4em@!C]{
      {\cat{Alg\T h}^\mathrm{op}} \ar[rr]^-{\mathrm{Ass}^\mathrm{op}}
      \ar[dr]_-{\mathrm{Sem}_\C} & &
      {\cat{Mnd}(\cat{Set})^\mathrm{op}} \ar[dl]^-{\mathrm{Sem}_\C} \\ &
      {\cat{CAT} / \C}
    } \qquad 
    \cd[@C-2.4em@!C]{
      {\cat{Alg\T h}^\mathrm{op}} \ar[rr]^-{\mathrm{Ass}^\mathrm{op}}
      \ar[dr]_-{\mathrm{Cosem}_\C} & &
      {\cat{Mnd}(\cat{Set})^\mathrm{op}\rlap{ .}} \ar[dl]^-{\mathrm{Cosem}_\C} \\ &
      {\cat{CAT} / \C}
    }
  \end{equation*}
\end{Prop}

In light of this result, we will henceforth prefer to deal with
accessible monads, though noting as we go along any 
simplifications afforded by having available a presentation
via an algebraic theory.

\section{Presheaf monads and comonads}
\label{sec:presh-monads-comon}
\subsection{Presheaf monads and comonads}
\label{sec:presheaf-comonads}
In this section, we describe and study the presheaf monads
and comonads which will be crucial to our main results. This
is largely revision from the literature, though
Propositions~\ref{prop:22} and~\ref{prop:13} are novel.

\begin{Defn}[Presheaf monad and comonad]
  \label{def:10}
  Let $\mathbb{B}$ be a small category. The \emph{presheaf monad}
  $\mathsf{T}_\mathbb{B} \in \cat{Mnd}_a(\cat{Set})$ and the
  \emph{presheaf comonad}
  $\mathsf{Q}_\mathbb{B} \in \cat{Cmd}_a(\cat{Set})$ are the monad and
  comonad induced by the respective adjunctions:
  \begin{equation}\label{eq:31}
    \cd{
      {\cat{Set}^{\mathbb{B}^\mathrm{op}}}
      \ar@<-4.5pt>[r]_-{\mathsf{res}_{J^\mathrm{op}}}
      \ar@{<-}@<4.5pt>[r]^-{\mathsf{lan}_{J^\mathrm{op}}} \ar@{}[r]|-{\bot} &
      {{\cat{Set}^{\mathrm{ob}(\mathbb{B})}}} \ar@<-4.5pt>[r]_-{\Pi} \ar@{<-}@<4.5pt>[r]^-{\Delta} \ar@{}[r]|-{\bot} &
      {\cat{Set}} & &
      {\cat{Set}^{\mathbb{B}}}
      \ar@<-4.5pt>[r]_-{\mathsf{res}_{J}}
      \ar@{<-}@<4.5pt>[r]^-{\mathsf{ran}_{J}} \ar@{}[r]|-{\top} &
      {{\cat{Set}^{\mathrm{ob}(\mathbb{B})}}} \ar@<-4.5pt>[r]_-{\Sigma} \ar@{<-}@<4.5pt>[r]^-{\Delta} \ar@{}[r]|-{\top} &
      {\cat{Set}} \rlap{ ,}
    }
  \end{equation}
  where $J \colon \mathrm{ob}(\mathbb{B}) \rightarrow \mathbb{B}$ is
  the inclusion of objects, and where $\mathrm{res}$,
  $\mathrm{lan}$ and $\mathrm{ran}$ denote restriction, left Kan
  extension and right Kan extension. If we write $\mathbb{B}_b$ for the
  set of all $\mathbb{B}$-maps with domain $b$, then the underlying
  endofunctors 
  are given by:
  \begin{equation*}
    T_{\mathbb{B}}(A) = \textstyle\prod_{b \in \mathbb{B}} \mathbb{B}_b \times A
    \qquad \text{and} \qquad Q_{\mathbb{B}}(A) = \sum_{b \in \mathbb{B}} A^{\mathbb{B}_b}\rlap{ ;}
  \end{equation*}
  the unit and multiplication for $\mathsf{T}_\mathbb{B}$ are given by:
  \begin{align*}
    \eta_A \colon A &\rightarrow \textstyle\prod_{b}
    (\mathbb{B}_b \times A) & \mu_A \colon \textstyle\prod_{b}
    (\mathbb{B}_b \times \prod_{b'}
    (\mathbb{B}_{b'} \times A))
    &\rightarrow \textstyle\prod_{b}
    (\mathbb{B}_b \times A)\\
    a & \mapsto \lambda b.\, (1_b, a) & \lambda b.\, (f_b, \lambda b'.\,
    (g_{bb'}, a_{bb'})) &\mapsto
    \lambda b.\, (g_{b,\mathrm{cod}(f_b)} \circ f_b, a_{b,\mathrm{cod}(f_b)})\text{ .}
  \end{align*}
  while the counit and comultiplication for $\mathsf{Q}_\mathbb{B}$ are
  given by:
  \begin{equation}\label{eq:30}
    \begin{aligned}
      \varepsilon_A \colon \textstyle\sum_{b} A^{\mathbb{B}_b} & \rightarrow A & \quad
      \delta_A \colon \textstyle\sum_{b} A^{\mathbb{B}_b} &
      \rightarrow \textstyle\sum_{b} \bigl(\sum_{b'}
      A^{\mathbb{B}_{b'}}\bigr)^{\mathbb{B}_b}\\
      (b,\varphi) & \mapsto \varphi(1_b) & (b, \varphi) & \mapsto \bigl(b,
      \lambda f.\, (\mathrm{cod}(f), \lambda g.\, \varphi(gf))\bigr)\rlap{ .}
    \end{aligned}
  \end{equation}
  We call a general $\mathsf{T} \in \cat{Mnd}_a(\cat{Set})$ a
  \emph{presheaf monad} if it is isomorphic to some
  $\mathsf{T}_\mathbb{B}$, and correspondingly on the comonad side.
\end{Defn}
Presheaf monads and comonads have been considered in computer
science; in~\cite{Ahman2017Taking} they are termed
``dependently typed update monads'' and ``dependently typed coupdate
comonads'' respectively, but both have a longer history, as we now
recall.

To the comonad side, we note that the underlying endofunctor of a
presheaf comonad is \emph{polynomial}, i.e., a coproduct of
representable functors. Such endofunctors are exactly those which
arise as the interpretations of set-based
\emph{containers}~\cite{Abbott2005Containers:}, and
in~\cite{Ahman2012When}, this was enhanced to a characterisation of
polynomial comonads as the interpretations of so-called \emph{directed
  containers}. Now, as observed in \cite{Ahman2017Taking}, directed
containers correspond bijectively to small categories, and so we
conclude that the presheaf comonads on $\cat{Set}$ are precisely the
polynomial comonads. For self-containedness, we include a short direct
proof of this fact.

\begin{Prop}
  \label{prop:6}
  For a comonad $\mathsf{Q}$ on $\cat{Set}$, the following  conditions
  are equivalent:
  \begin{enumerate}[(i)]
  \item $\mathsf{Q}$ is a presheaf comonad;
  \item The underlying endofunctor $Q$ is a coproduct of representables;
  \item The underlying endofunctor $Q$ preserves connected limits.
  \end{enumerate}
\end{Prop}
\begin{proof}
  Clearly (i) $\Rightarrow$ (ii), and (ii) $\Leftrightarrow$ (iii) is
  standard category theory due to Diers~\cite{Diers1978Spectres}; so
  it remains to show (ii) $\Rightarrow$ (i). Suppose, then, that
  $Q = \sum_{b \in B} (\thg)^{E_b}$ is a coproduct of representables.
  By the Yoneda lemma, giving $\varepsilon \colon Q \Rightarrow 1$ is
  equivalent to giving the elements
  $1_b \defeq \varepsilon_{E_b}(b, \lambda f.\, f) \in E_b$ for each
  $b \in B$. Similarly, giving $\delta \colon Q \Rightarrow QQ$ is
  equivalent to giving for each $b \in B$ an element of $QQ(E_b)$,
  i.e., elements $\alpha(b) \in B$ and
  $\lambda f.\, (c(f), \rho_f) \colon E_{\alpha(b)} \rightarrow
  \sum_{b'} E_b^{E_{b'}}$. Now the three comonad axioms correspond
  under the Yoneda lemma to the following assertions:
  \begin{itemize}
  \item The axiom $\varepsilon Q \circ \delta = 1_Q$ asserts that
    $\alpha(b) = b$ and $\rho_f(1_{c(f)}) = f$;
  \item The axiom $Q \varepsilon \circ \delta = 1_Q$ asserts that
    $c(1_b) = b$ and $\rho_{1_b} = \mathrm{id}_{E_b}$;
  \item The axiom $Q\delta \circ \delta = \delta Q \circ \delta$
    asserts that $c(g) = c(\rho_f(g))$ and
    $\rho_f \circ \rho_g = \rho_{\rho_f(g)}$.
  \end{itemize}
  But these are precisely the data and axioms of a small category
  $\mathbb{B}$ with object-set $B$, with $\mathbb{B}_b = E_b$, with
  identities $1_b$, with codomain map $c$, and with precomposition by
  $f$ given by $\rho_f$; and on defining $\mathbb{B}$ in this way, we
  clearly have $\mathsf{Q} = \mathsf{Q}_\mathbb{B}$.
\end{proof}

To the monad side, presheaf monads seem to have been first considered
in~\cite[Example~8.7]{Johnstone1990Collapsed}, in terms
of a presentation as an algebraic theory. We now recall this
presentation, though framing it in terms of the applications
of~\cite{Ahman2014Update}. 
\begin{Not}
  \label{not:1}
  Let $\Sigma$ be a signature, and $q \in \Sigma(A)$ and $t,u \in
  \Sigma(B)$ terms where without loss of generality  $A$
  is disjoint from $B$. For any $i \in A$, we write $t
  \equiv_{q,i} u$ (read as ``$t$ and $u$ are equal in the $i$th place
  of $q$'') as an abbreviation for the equation
  \begin{equation*}
    q\Bigl(\lambda a.\,
    \begin{cases}
      t(\vec b) & \text{if $a = i$}\\
      a & \text{if $a \neq i$}
    \end{cases}
\Bigr) \equiv q\Bigl(\lambda a.\,
    \begin{cases}
      u(\vec b) & \text{if $a = i$}\\
      a & \text{if $a \neq i$}
    \end{cases}\Bigr)
\ \qquad \text{in $\Sigma(A \cup B \setminus \{i\})$}\rlap{ .}
  \end{equation*}
\end{Not}

\begin{Ex}[Dependently typed update]
  \label{def:27}
  Let $\mathbb{B}$ be a small category, whose objects we view as
  \emph{values}, and whose arrows
  $b \rightarrow b'$ we view as \emph{updates} from $b$ to $b'$. 
  The theory of \emph{$\mathbb{B}$-valued dependently typed update} is
  generated by an $\mathrm{ob}(\mathbb{B})$-ary operation $\mathsf{get}$ satifying the axioms
  of read-only state, together with a unary operation
  $\mathsf{upd}_f$ for each morphism $f \colon b \rightarrow b'$
  in $\mathbb{B}$, subject to the
  equations 
  \begin{align}
  \mathsf{upd}_{f}(x) &\equiv_{\mathsf{get},c} x & \text{for $f
    \colon b \rightarrow b'$ and $c \neq b$ in $\mathbb{B}$;}\label{eq:3}\\
    \mathsf{upd}_{f}(\mathsf{get}(\lambda a.\, x_a))
    & \equiv_{\mathsf{get},b} \mathsf{upd}_f(x_{b'}) & \text{for $f \colon
    b \rightarrow b'$ in $\mathbb{B}$;}\label{eq:5}\\
  \mathsf{upd}_{1_b}(x) &\equiv x &\text{for $b \in \mathrm{ob}(\mathbb{B})$;}\label{eq:6}\\
    \mathsf{upd}_{f}(\mathsf{upd}_{g}(x)) &\equiv_{\mathsf{get},b}
    \mathsf{upd}_{gf}(x) & \text{for $f \colon b \rightarrow b'$, $g
      \colon b' \rightarrow b''$ in $\mathbb{B}$.}\label{eq:7}
  \end{align}
  The intended semantics is that $\mathsf{get}$ reads a value
  associated to the current state; while $\mathsf{upd}_f$, for
  $f \colon b\rightarrow b'$ in $\mathbb{B}$, attempts to update the
  value $b$ to $b'$ via $f$. If the value of the current state is not
  $b$, then the update fails (the first axiom above); while if the
  value \emph{is} $b$, then we move to a new state with associated
  value $b'$ (the second axiom). The remaining axioms assert that
  updates compose as expected.
\end{Ex}

We now justify our nomenclature by showing that the theory of
$\mathbb{B}$-valued dependently typed update generates the presheaf
monad $\mathsf{T}_\mathbb{B}$, which as we have explained, is equally
an example of a dependently typed update monad as
in~\cite{Ahman2014Update}.

\begin{Prop}
  \label{prop:21}
  For any small category $\mathbb{B}$, the  theory of
  $\mathbb{B}$-valued dependently typed
  update generates 
  the presheaf monad $\mathsf{T}_\mathbb{B}$.
\end{Prop}
\begin{proof}
  For each set $A$, we make $T_\mathbb{B}(A) = \prod_b (\mathbb{B}_b
  \times A)$ a model of the theory $\mathbb{T}$ of dependently-typed update
  by taking $\dbr{\mathsf{get}}\bigl(\lambda b. (\lambda c.\, g_{bc}, \lambda c.\,
    a_{bc})\bigr) = \lambda b.\, (g_{bb}, a_{bb})$ and % taking, for $f \colon b \rightarrow b'$ in $\mathbb{B}$,
  \begin{equation*}
    \dbr{\mathsf{upd}_f}(\lambda c.\, (g_c, a_c)) = \lambda c.\,
    \begin{cases}
      (g_c, a_c) & \text{if $c \neq b$;}\\
      (g_{b'}f, a_{b'}) & \text{if $c = b$} 
    \end{cases} \qquad \text{for $f \colon
        b \rightarrow b'$ in $\mathbb{B}$.}
  \end{equation*}
  It is easy to verify that an equation $t \equiv_{\mathsf{get}, b} u$
  holds in $T_\mathbb{B}(A)$ precisely when the interpretations
  $\dbr{t}$ and $\dbr{u}$ have the same postcomposition with the
  projection
  $\pi_b \colon \prod_b (\mathbb{B}_b \times A) \rightarrow
  \mathbb{B}_b \times A$; whence the axioms for dependently typed
  update are satisfied. Thus we may extend
  $\eta_A \colon A \rightarrow \prod_b (\mathbb{B}_b \times A)$
  uniquely to a homomorphism
  $p_A \colon T(A) \rightarrow \prod_b (\mathbb{B}_b \times A)$; we
  also have a function
  $i_A \colon \prod_b (\mathbb{B}_b \times A) \rightarrow T(A)$ given
  by
  $\lambda c.\,(g_c, a_c) \mapsto \mathsf{get}(\lambda c.\,
  \mathsf{upd}_{g_c}(a_c))$ which we claim is also a model
  homomorphism. It commutes with $\mathsf{get}$ since
  $\mathsf{get}(\lambda b.\, \mathsf{get}(\lambda c.\,
  \mathsf{upd}_{g_{bc}}(a_{bc}))) = \mathsf{get}(\lambda c.\,
  \mathsf{upd}_{g_{cc}}(a_{cc}))$. To see it commutes with
  $\mathsf{upd}_f$ for $f \colon b \rightarrow b'$, we observe that
  \begin{align*}
  \mathsf{upd}_f(\mathsf{get}(\lambda c.\,
    \mathsf{upd}_{g_c}(a_{c})) &\equiv_{\mathsf{get},b}
    \mathsf{upd}_f(\mathsf{upd}_{g_{b'}}(a_{b'}))
    \equiv_{\mathsf{get},b}
    \mathsf{upd}_{g_b' f}(a_{b'})\\
    \text{and }   \mathsf{upd}_f(\mathsf{get}(\lambda c.\,
    \mathsf{upd}_{g_c}(a_{c})) &\equiv_{\mathsf{get},c}  \mathsf{get}(\lambda c.\,
    \mathsf{upd}_{g_c}(a_{c})) \equiv_{\mathsf{get},c}  \mathsf{upd}_{g_c}(a_{c}) \text{ for $c \neq b$,}
  \end{align*}
  from which it follows that
  \begin{equation*}
    \mathsf{upd}_f(\mathsf{get}(\lambda c.\,
    \mathsf{upd}_{g_c}(a_{c})) = \mathsf{get}\Bigl(\lambda c.\,     \begin{cases}
      \mathsf{upd}_{g_c}(a_c) & \text{if $c \neq b$;}\\
      \mathsf{upd}_{g_{b'}f}(a_{b'}) & \text{if $c = b$}
    \end{cases} \Bigr)
  \end{equation*}
  as desired.
  Since $\mathsf{get}(\lambda b.\, \mathsf{upd}_{1_b}(a))
  \equiv \mathsf{upd}_{1_b}(a) \equiv a$, we have
  $i_A(\eta_A(a)) = a$, from which it follows that $i_Ap_A = 1_{T(A)}$. On
  the other hand, $p_Ai_A = 1$ by a short calculation, 
  and so $p_A$ and $i_A$ are mutually inverse. In this way, we obtain a natural isomorphism $T \cong T_\mathbb{B}$,
  which by construction is compatible with the units of the monads
  $\mathsf{T}$ and $\mathsf{T}_\mathbb{B}$. Compatibility
  with the multiplications follows since:
  \begin{align*}
    \mathsf{get}(\lambda b.\,\mathsf{upd}_{f_b}(\mathsf{get}(\lambda
    c.\, \mathsf{upd}_{g_{bc}}(a_{bc})))) &\equiv_\mathbb{T}
    \mathsf{get}(\lambda b.\,\mathsf{upd}_{f_b}(
    \mathsf{upd}_{g_{b, \mathsf{c}(f_b)}}(a_{b, \mathsf{c}(f_b)}))) \\&\equiv_\mathbb{T}
    \mathsf{get}(\lambda b.\,\mathsf{upd}_{g_{b,\mathsf{c}(f_b)} \circ f_b}(a_{b,
      \mathsf{c}(f_b)}))\rlap{ .}\qedhere
  \end{align*}
\end{proof}

\subsection{Morphisms of presheaf monads and comonads}
\label{sec:morph-presh-monads}
We now examine morphisms between presheaf monads and
comonads. Beginning again on the comonad side, we observe, as
in~\cite{Ahman2017Taking}, that comonad morphisms between presheaf comonads
do not correspond to functors, but rather to \emph{cofunctors}:

\begin{Defn}[Cofunctor]
  \label{def:44}\cite{Higgins1993Duality, Aguiar1997Internal}
  A \emph{cofunctor} $F \colon \mathbb{B} \rightsquigarrow \mathbb{C}$
  between small categories comprises an action on objects
  $F \colon \mathrm{ob}(\mathbb{B}) \rightarrow
  \mathrm{ob}(\mathbb{C})$ together with actions on morphisms
  $F_{b} \colon \mathbb{C}_{Fb} \rightarrow \mathbb{B}_b$
  for each $b \in \mathbb{B}$, subject to the axioms:
  \begin{enumerate}[(i)]
  \item $F(\mathrm{cod}(F_b(f))) = \mathrm{cod}(f)$ for all $f \in
    \mathbb{C}_{Fb}$;
  \item $F_b(1_{F b}) = 1_b$ for all $b \in \mathbb{B}$;
  \item $F_b(gf) = F_{\mathrm{cod}(F_b f)}(g) \circ F_b(f)$ for all $f
    \in \mathbb{C}_{Fb}$ and $g \in
    \mathbb{C}_{\mathrm{cod}(f)}$.
  \end{enumerate}
  We write $\cat{Cof}$ for the category of small categories and
  cofunctors.
\end{Defn}
\begin{Prop}
  \label{prop:16}
  Taking presheaf comonads is the action on objects of a fully faithful functor
  $\mathsf{Q}_{(\thg)} \colon \cat{Cof} \rightarrow
  \cat{Cmd}_a(\cat{Set})$, which on morphisms 
  sends a cofunctor $F \colon \mathbb{B} \rightsquigarrow \mathbb{C}$ 
  to the comonad morphism
  $\mathsf{Q}_F \colon \mathsf{Q}_\mathbb{B} \rightarrow \mathsf{Q}_\mathbb{C}$ with components
  \begin{equation}\label{eq:28}
  \begin{aligned}
    % \textstyle\prod_{c \in \mathbb{C}} (\mathbb{C}_c \times A) & \rightarrow
    % \textstyle\prod_{b \in \mathbb{B}} (\mathbb{B}_b \times A) & 
    \textstyle\sum_{b \in \mathbb{B}} A^{\mathbb{B}_b} & \rightarrow
    \textstyle\sum_{c \in \mathbb{C}} A^{\mathbb{C}_c} \\
    % \lambda c.\, (f_c, a_c) & \mapsto \lambda b.\, (F_b(f_{Fb}),
    % a_{Fb}) &
    (b, \varphi) & \mapsto (Fb, \varphi \circ F_b)\rlap{ .}
  \end{aligned}
  \end{equation}
  % are the respective components of a monad morphism
  % $\mathsf{T}_F \colon \mathsf{T}_\mathbb{C} \rightarrow \mathsf{T}_\mathbb{B}$ and a
  % comonad morphism
  % $\mathsf{Q}_F \colon \mathsf{Q}_\mathbb{B} \rightarrow \mathsf{Q}_\mathbb{C}$.% ; and these
  % assignments give rise to functors
  % \begin{equation*}
  %   \mathsf{T}_{(\thg)} \colon \cat{Cof} \rightarrow
  %   \cat{Mnd}_a(\cat{Set})^\mathrm{op} \qquad 
  %   \mathsf{Q}_{(\thg)} \colon \cat{Cof} \rightarrow
  %   \cat{Cmd}_a(\cat{Set})\rlap{ .} 
  % \end{equation*}
\end{Prop}
This result is again due to~\cite{Ahman2017Taking},
but we sketch a proof for self-containedness.
\begin{proof}
  Let $\mathbb{B}$ and $\mathbb{C}$ be small categories. As
  $Q_\mathbb{B} = \sum_{b \in B} (\thg)^{\mathbb{B}_b}$, we see once
  again by the
  Yoneda lemma that giving
  a natural transformation 
  $\alpha \colon Q_\mathbb{B} \Rightarrow Q_\mathbb{C}$ is equivalent to giving elements
  $\alpha_{\mathbb{B}_b}(b, 1_b) \in Q_\mathbb{C}(\mathbb{B}_b)$; and
  if we write these elements as pairs
  $(Fb \in \mathbb{C}, F_b \colon \mathbb{C}_{Fb} \rightarrow
  \mathbb{B}_b)$, then $\alpha$ itself must have components given as
  in~\eqref{eq:28}. Similar arguments to the proof of
  Proposition~\ref{prop:6} now
  show that $\alpha$ commutes with the comonad counits and
  comultiplication precisely when the assignments
  $b \mapsto (Fb, F_b)$ satisfy the axioms (i)--(iii) for a cofunctor.
  % So $\mathsf{Q}_{(\thg)}$ is indeed fully faithful. It is a
  % straightforward calculation from the axioms for a cofunctor to show
  % that the displayed functions commute with the counit and
  % comultiplication maps of Definition~\ref{def:10}.
\end{proof}

On the monad side, it turns out that monad morphisms between presheaf
monads are also cofunctors between the corresponding categories. This
is not quite as straightforward to see, and for the moment we record
only the weaker statement that cofunctors induce morphisms of presheaf
monads; the full claim will be proved in Proposition~\ref{prop:26} below.

\begin{Prop}
  \label{prop:22}
  Taking presheaf monads is the action on objects of a functor
  $\mathsf{T}_{(\thg)} \colon \cat{Cof} \rightarrow
  \cat{Mnd}_a(\cat{Set})^\mathrm{op}$, which on morphisms 
  sends a cofunctor $F \colon \mathbb{B} \rightsquigarrow \mathbb{C}$ 
  to the monad morphism
  $\mathsf{T}_F \colon \mathsf{T}_\mathbb{C} \rightarrow \mathsf{T}_\mathbb{B}$ with components
  \begin{equation}\label{eq:32}
  \begin{aligned}
    \textstyle\prod_{c \in \mathbb{C}} (\mathbb{C}_c \times A) & \rightarrow
    \textstyle\prod_{b \in \mathbb{B}} (\mathbb{B}_b \times A) \\
    \lambda c.\, (f_c, a_c) & \mapsto \lambda b.\, (F_b(f_{Fb}),
    a_{Fb})\rlap{ .}
  \end{aligned}
  \end{equation}
\end{Prop}
\begin{proof}
  It is easy to check that the components~\eqref{eq:32} are compatible
  with the units and multiplications of the presheaf monads
  $\mathsf{T}_\mathbb{B}$ and $\mathsf{T}_\mathbb{C}$.
\end{proof}

\subsection{Semantics}
\label{sec:coalg-diagr-comon}

We now consider the semantics associated to presheaf comonads and
monads. Starting again on the comonad side, it turns out that
the adjunction~\eqref{eq:31} inducing the presheaf comonad
$\mathsf{Q}_\mathbb{B}$ is comonadic, but not \emph{strictly} so;
thus, $\mathsf{Q}_\mathbb{B}$-coalgebras are not
exactly presheaves $\mathbb{B} \rightarrow \cat{Set}$, but only
something equivalent:

\begin{Defn}[Left $\mathbb{B}$-set]
  \label{def:39}
  Let $\mathbb{B}$ be a small category. A \emph{left $\mathbb{B}$-set}
  is a set $X$ endowed with a projection map
  $p \colon X \rightarrow \mathrm{ob}(\mathbb{B})$ and an action
  $\mathord\ast \colon \textstyle\sum_{x \in X} \mathbb{B}_{p(x)}
  \rightarrow X$, notated as $(x,f) \mapsto f \ast x$, satisfying the
  typing axiom $p(f \ast x) = \mathrm{cod}(f)$ and the functoriality
  axioms $\mathrm{id} \ast x = x$ and
  $g \ast (f \ast x) = (g \circ f) \ast x$. We write
  $\mathbb{B}\text-\cat{Set}$ for the category of left
  $\mathbb{B}$-sets, whose maps are functions commuting with the
  projections and actions. We write
  $U^\mathbb{B} \colon \mathbb{B}\text-\cat{Set} \rightarrow
  \cat{Set}$ for the forgetful functor $(X,p,\ast) \mapsto X$.

  Given a cofunctor $F \colon \mathbb{B} \rightsquigarrow \mathbb{C}$
  between small categories, we define the functor
  $\Sigma_F \colon \mathbb{B}\text-\cat{Set} \rightarrow
  \mathbb{C}\text-\cat{Set}$ over $\cat{Set}$ to act as follows, where
  we write $f \ast^F x \defeq F_{p(x)}(f) \ast x$:
  \begin{equation*}
    (X \xrightarrow{p} \mathrm{ob}(\mathbb{B}),\,
    \Sigma_{x \in X} \mathbb{B}_{p(x)} \xrightarrow{\smash{\ast}} X) \ \mapsto\ (X
    \xrightarrow{\smash{Fp}} \mathrm{ob}(\mathbb{C}), \Sigma_{x \in X} \mathbb{C}_{F(p(x))} \xrightarrow{\smash{\ast^F}} X)\rlap{ .}
  \end{equation*}
\end{Defn}

In~\cite{Ahman2014Coalgebraic}, 
what we call a left $\mathbb{B}$-set was termed a \emph{coalgebraic
  update lens}; the following result, which is immediate from the
definitions, was also observed there.

\begin{Prop}
  \label{prop:19}
  For any small category $\mathbb{B}$, the category of
  Eilenberg--Moore $\mathsf{Q}_\mathbb{B}$-coalgebras is isomorphic to
  $\mathbb{B}\text-\cat{Set}$ via an isomorphism commuting with the
  functors to $\cat{Set}$. These isomorphisms are natural with respect
  to cofunctors $F \colon \mathbb{B} \rightsquigarrow \mathbb{C}$.% , in
%   the sense of rendering commutative each square
% \begin{equation*}
%     \cd{
%       {\cat{Coalg}(\mathsf{Q}_{\mathbb{B}}) } \ar[r]^-{\cong}
%       \ar[d]_{(\mathsf{Q}_F)_!} &
%       {\mathbb{B}\text-\cat{Set}} \ar[d]^{\Sigma_F} \\
%       {\cat{Coalg}(\mathsf{Q}_{\mathbb{C}})} \ar[r]_-{\cong} &
%       {\mathbb{C}\text-\cat{Set}}\rlap{ .}
%     }
% \end{equation*}
\end{Prop}

Turning now to the presheaf monad $\mathsf{T}_\mathbb{B}$, it follows
from the results of~\cite{Johnstone1990Collapsed} that the category of
$\mathsf{T}_\mathbb{B}$-models is equivalent to the category of
presheaves $X \colon \mathbb{B}^\mathrm{op} \rightarrow \cat{Set}$
which \emph{either} have each $X(b)$ empty, \emph{or} each $X(b)$
non-empty. However, we will be less interested in characterising the
$\mathsf{T}_\mathbb{B}$-models in $\cat{Set}$ than the
$\mathsf{T}_\mathbb{B}$-\emph{comodels}. We may exploit the fact that
$\mathsf{T}_\mathbb{B}$ is generated by the theory of
$\mathbb{B}$-valued dependently typed update to obtain such a
characterisation.

\begin{Prop}
  \label{prop:13}
  For any small category $\mathbb{B}$, the category of comodels of the
  theory $\mathbb{T}_\mathbb{B}$ of $\mathbb{B}$-valued
  dependently-typed update is isomorphic to
  $\mathbb{B}\text-\cat{Set}$ via a functor
  \begin{equation}\label{eq:10}
    \cd{
      {\mathbb{B}\text-\cat{Set}} \ar[rr]^-{} \ar[dr]_-{U^\mathbb{B}} & &
      {{}^{\mathbb{T}_\mathbb{B}} \cat{Set}} \ar[dl]^-{{}^{\mathbb{T}_\mathbb{B}} U} \\ &
      {\cat{Set}}
    }
  \end{equation}
  which sends a left $\mathbb{B}$-set $(X,p,\mathord \ast)$ to the
  $\mathbb{T}_\mathbb{B}$-comodel $\alg = (X, \dbr{\thg}_{\alg})$ with
  \begin{equation*}
    \dbr{\mathsf{get}}_{\alg}(x) = (p(x), x) \qquad
    \dbr{\mathsf{upd}_f}_{\alg}(x) =
    \begin{cases}
      x & \text{ if $p(x) \neq \mathrm{dom}(f)$;}\\
      f \ast x & \text{ if $p(x) = \mathrm{dom}(f)$.}
    \end{cases}
  \end{equation*}
\end{Prop}

\begin{proof}
  A comodel of $\mathbb{B}$-valued dependently typed state is firstly,
  a comodel of $\mathrm{ob}(\mathbb{B})$-valued read-only state, i.e.,
  a set $X$ endowed with a function
  $p \colon X \rightarrow \mathrm{ob}(\mathbb{B})$. On top of this, we
  have functions $\dbr{\mathsf{upd}_f} \colon X \rightarrow X$ for
  each $f \colon b \rightarrow b'$ in $\mathbb{B}$ which satisfy the
  equations~\eqref{eq:3}--\eqref{eq:7}. The first forces
  $\dbr{\mathsf{upd}_f}$ to act trivially on the fibre $p^{-1}(c)$ for
  all $c \neq b$, while the second forces it to map $p^{-1}(b)$ into
  $p^{-1}(b')$. So to give the $\dbr{\mathsf{upd}_f}$'s
  satisfying~\eqref{eq:3} and~\eqref{eq:5} is equally to give
  functions $f \ast (\thg) \colon p^{-1}(b) \rightarrow p^{-1}(b')$
  for each $f \colon b \rightarrow b'$ in $\mathbb{B}$. Now the last
  two axioms~\eqref{eq:6} and~\eqref{eq:7} impose the functoriality
  constraints $1_b \ast x = x$ and
  $g \ast (f \ast x) = (g \circ f) \ast x$, so that, in sum, a
  comodel of $\mathbb{B}$-valued dependently typed update can be
  identified with a left $\mathbb{B}$-set, via
  the identification given in the statement of the result.
\end{proof}

\section{The costructure--cosemantics adjunction}
\label{sec:monad-comon-adjunct}

In this section, we construct the adjunction
which is the main object of study of this paper. We begin by
explaining how taking comodels
yields a \emph{cosemantics} functor from accessible monads to
accessible comonads on $\cat{Set}$. We then show that this functor has
a left adjoint, as displayed below, which we term the
\emph{costructure} functor; and finally, we explain how this relates
to the material of~\cite{Katsumata2019Interaction}.
\begin{equation}\label{eq:4}
  \cd{
    {\cat{Mnd}_a(\cat{Set})^\mathrm{op}} \ar@<-4.5pt>[r]_-{\mathrm{Cosem}} \ar@{<-}@<4.5pt>[r]^-{\mathrm{Costr}} \ar@{}[r]|-{\bot} &
    {\cat{Cmd}_a(\cat{Set})}
  }
\end{equation}

\subsection{The cosemantics comonad of an accessible monad}
\label{sec:assoc-comon-an}

Our first task is to show that the cosemantics functor of
Definition~\ref{def:4} yields the right adjoint functor
in~\eqref{eq:4}. We begin with the basic facts about Eilenberg--Moore
semantics for comonads.

\begin{Defn}[Eilenberg--Moore semantics]
  \label{def:9}
  Let $\cat{Cmd}(\C)$ be the category of comonads in $\C$. The
  \emph{Eilenberg--Moore semantics functor}
  $\mathrm{EM} \colon \cat{Cmd}(\C) \rightarrow \cat{CAT} / \C$ sends
  a comonad $\mathsf{Q} = (Q, \varepsilon, \delta)$ to the forgetful
  functor $U^\mathsf{Q} \colon \cat{Coalg}(\mathsf{Q}) \rightarrow \C$
  from its category of Eilenberg--Moore coalgebras, and sends
  $f \colon \mathsf{Q} \rightarrow \mathsf{P}$ to the functor
  $\cat{Coalg}(\mathsf{Q}) \rightarrow \cat{Coalg}(\mathsf{P})$ over
  $\C$ acting by
  $(X, x \colon X \rightarrow QX) \mapsto (X, f_X \circ x \colon X
  \rightarrow PX)$.
\end{Defn}

\begin{Lemma}
  \label{lem:3}
  For any category $\C$, the semantics functor
  $\mathrm{EM} \colon \cat{Cmd}(\C) \rightarrow \cat{CAT} / \C$ is
  full and faithful, and its essential image comprises the strictly
  comonadic functors.
\end{Lemma}

\begin{proof}
  The first part is~\cite[Theorem~6.3]{Barr1985Toposes}; the second is
  easy from the definitions.
\end{proof}

Here, a functor $V \colon \D \rightarrow \C$ is strictly comonadic if
it has a right adjoint $G$, and the canonical comparison functor from
$\D$ to the category of coalgebras for the comonad $VG$ is an
isomorphism. Concrete conditions for a functor to be strictly
comonadic are given by the Beck comonadicity
theorem~\cite[Theorem~3.14]{Barr1985Toposes}.

\begin{Prop}
  \label{prop:2}
  For an accessible $\cat{Set}$-monad $\mathsf{T}$, the forgetful
  functor
  ${{}^\mathsf{T} U \colon {}^\mathsf{T} \cat{Set} \rightarrow
    \cat{Set}}$ from the category of comodels is strictly comonadic
  for an accessible comonad.
\end{Prop}

In the finitary case, this is~\cite[Theorem~2.2]{Power2004From}; this
more general form can be proven as a routine application of the theory
of locally presentable categories. We omit this, as
Theorem~\ref{thm:2} provides an independent elementary argument.

\begin{Cor}
  \label{cor:2}
  The cosemantics functor
  $\cat{Mnd}_a(\cat{Set})^\mathrm{op} \rightarrow \cat{CAT} /
  \cat{Set}$ factors as
  \begin{equation*}
    \cat{Mnd}_a(\cat{Set})^\mathrm{op} \xrightarrow{\mathrm{Cosem}}
    \cat{Cmd}_a(\cat{Set}) \xrightarrow{\mathrm{EM}} \cat{CAT} / \cat{Set}\rlap{ .}
  \end{equation*}
\end{Cor}

\subsection{The costructure monad of an accessible comonad}
\label{sec:costructure-monad-an}

We now show that the cosemantics functor
$\cat{Mnd}_a(\cat{Set})^\mathrm{op} \rightarrow
\cat{Cmd}_a(\cat{Set})$ has a left adjoint. This will arise from the
``structure--semantics adjointness''
of~\cite{Lawvere1963Functorial,Dubuc1970Kan-extensions}, which we now
recall.

\begin{Defn}[Endomorphism monad]
  \label{def:14}
  Let $\C$ be a category with powers which is not necessarily locally
  small. We say that $X \in \C$ is \emph{tractable} if, for any set
  $A$, the collection of maps $X^A \rightarrow X$ form a set. For such
  an $X$, the \emph{endomorphism monad} $\mathsf{End}_\C(X)$ on
  $\cat{Set}$ has action on objects $A \mapsto \C(X^A,X)$; unit
  functions $A \rightarrow \C(X^A,X)$ given by $a \mapsto \pi_a$; and
  Kleisli extension
  $u^\dagger \colon \C(X^A, X) \rightarrow \C(X^B, X)$ of
  $u \colon A \rightarrow \C(X^B, X)$ given by
  $t \mapsto t \circ (u_a)_{a \in A}$.
\end{Defn}

Note that endomorphism monads need \emph{not} be accessible; for
example, the endomorphism monad of $V \in \cat{Set}$ is the
non-accessible continuation monad $V^{V^{(\thg)}}$.

\begin{Lemma}
  \label{lem:4}
  Let $\C$ be a category with powers, not necessarily locally small,
  and let $X \in \C$ be tractable. There is a bijection, natural in
  $\mathsf{T}$, between monad morphisms
  $\mathsf{T} \rightarrow \mathsf{End}_\C(X)$ and $\mathsf{T}$-model
  structures on $X$.
\end{Lemma}

\begin{proof}
  To give a monad map $\mathsf{T} \rightarrow
  \mathsf{End}_\C(X)$ is to give functions $T(A) \rightarrow \C(X^A,X)$
  for each set $A$, compatibly with units and Kleisli
  extensions. If we write the action of these functions as $t \mapsto
  \dbr{t}_{\alg}$, then these compatibilities are precisely the
  conditions~\eqref{eq:24} to make the $\dbr{t}_{\alg}$'s into a
  $\mathsf{T}$-model structure on $X$.
\end{proof}

In the following result, we call a functor $V \colon \A
\rightarrow \C$ \emph{tractable} if it is tractable as an object of
the (not necessarily locally small) functor category $[\A, \C]$.
\begin{Prop}[Structure/semantics]
  \label{prop:10}
  Let $\C$ be a category with powers. The semantics
  functor~$\mathrm{Sem}_\C \colon \cat{Mnd}(\cat{Set})^\mathrm{op} \rightarrow
  \cat{CAT} / \C$ of Definition~\ref{def:4} has a partial left adjoint
  at each tractable $V \colon \A \rightarrow \C$, given by the
  endomorphism monad $\mathsf{End}_{[\A,\C]}(V)$.
\end{Prop}
\begin{proof}
  Let $V \colon \A \rightarrow \C$ be tractable, so that
  $\mathsf{End}_{[\A, \C]}(V)$ exists. By Lemma~\ref{lem:4}, there
  is a bijection, natural in $\mathsf{T}$,
  between monad morphisms
  $\mathsf{T} \rightarrow \mathsf{End}_{[\A, \C]}(V)$
  and $\mathsf{T}$-model structures on $V$ in $[\A, \C]$. Now since powers
  in $[\A, \C]$ are computed componentwise, $\mathsf{T}$-model
  structures on $V$ correspond, naturally in $\mathsf{T}$, with liftings
  \begin{equation*}
    \cd{
      & \C^\mathsf{T} \ar[d]^-{U^\mathsf{T}}\\
      \A \ar[r]^-{V} \ar@{-->}[ur] & \C
    }
  \end{equation*}
  of $V$ through $U^\mathsf{T}$, i.e., with maps $V \rightarrow
  \mathrm{Sem}_\C(\mathsf{T})$ in $\cat{CAT} / \C$.
\end{proof}

Because we are interested in comodels rather than models, we will
apply this result in its dual form: thus, we speak of the
\emph{cotractability} of a functor $V \colon \A \rightarrow \C$, meaning that each collection
$[\A, \C](V, A \cdot V)$ is a set, and the \emph{coendomorphism} monad
$\mathsf{Coend}_{[\A, \C]}(V)$ with action on objects $A \mapsto [\A, \C](V, A \cdot
V)$.% , providing the value at $V$ of a partial left adjoint to
% $\mathrm{Cosem}_\C \colon \cat{Mnd}(\cat{Set})^\mathrm{op} \rightarrow
% \cat{CAT} / \C$.

\begin{Lemma}
  \label{lem:2}
  Let $\mathsf{Q}$ be an accessible comonad on $\cat{Set}$. The
  forgetful functor from the category of Eilenberg--Moore coalgebras
  $U^\mathsf{Q} \colon \cat{Coalg}(\mathsf{Q}) \rightarrow \cat{Set}$
  is cotractable, and the coendomorphism monad
  $\mathsf{Coend}_{[\cat{Coalg}(\mathsf{Q}), \cat{Set}]}(U^\mathsf{Q})$ is accessible.
\end{Lemma}
\begin{proof}
  For cotractability, we show that for each set $A$, the collection of
  natural transformations
  $U^\mathsf{Q} \Rightarrow A \cdot U^\mathsf{Q}$ form a
  set. If we write $G^\mathsf{Q}$ for the right adjoint of
  $U^\mathsf{Q}$, then transposing under the adjunction
  $(\thg) \circ G^\mathsf{Q} \dashv (\thg) \circ U^\mathsf{Q}$ yields
  \begin{equation}\label{eq:1}
    [\mathsf{Q}\text-\cat{Coalg}, \cat{Set}](U^\mathsf{Q}, A \cdot
    U^\mathsf{Q}) \cong [\cat{Set}, \cat{Set}](Q, A \cdot \mathrm{id})\rlap{ ,}
  \end{equation}
  whose right-hand side is a set since $Q$ is accessible; whence also
  the left-hand side.

  So $\mathsf{Coend}(U^\mathsf{Q})$ exists; to show accessibility,
  note that a natural transformation
  $Q \Rightarrow A \cdot \mathrm{id}$ is equally a pair of natural
  transformations $Q \Rightarrow \Delta A$ and 
  $Q \Rightarrow \mathrm{id}$; and since $\cat{Set}$ has a terminal
  object, to give $Q \Rightarrow \Delta A$ is equally to give a
  function $Q1 \rightarrow A$. We conclude that
  $\mathrm{Coend}(U^\mathsf{Q}) \cong (\thg)^{Q1} \times [\cat{Set}, \cat{Set}](Q,\mathrm{id})\rlap{ ,}$
  which is a small coproduct of representable functors, and hence accessible.
\end{proof}
\begin{Prop}
  \label{prop:7}
  The functor
  $\mathrm{Cosem} \colon \cat{Mnd}_a(\cat{Set})^\mathrm{op}
  \rightarrow \cat{Cmd}_a(\cat{Set})$ of Corollary~\ref{cor:2} admits
  a left adjoint
  $\mathrm{Costr} \colon \cat{Cmd}_a(\cat{Set}) \rightarrow
  \cat{Mnd}_a(\cat{Set})^\mathrm{op}$, whose value at the accessible
  comonad $\mathsf{Q}$ is given by the coendomorphism monad
  $\mathsf{Coend}(U^\mathsf{Q})$.
\end{Prop}

\begin{proof}
  The preceding result shows that $\mathsf{Coend}(U^\mathsf{Q})$
  exists and is accessible for each accessible comonad $\mathsf{Q}$.
  Now by Lemma~\ref{lem:3}, Proposition~\ref{prop:2} and
  Proposition~\ref{prop:10}, we have natural isomorphisms
  \begin{equation*}
    \cat{Cmd}_a(\cat{Set})(\mathsf{Q}, \mathrm{Cosem}(\mathsf{T}))
    \cong
    \cat{CAT} / \cat{Set}(U^\mathsf{Q}, {}^\mathsf{T} U) \cong
    \cat{Mnd}_a(\cat{Set})(\mathsf{T},
    \mathsf{Coend}(U^\mathsf{Q}))\text{ .} \qedhere
  \end{equation*}
\end{proof}
\begin{Rk}
  \label{rk:1}
  For future use, we record the concrete form of the adjointness
  isomorphisms of the costructure--cosemantics adjunction. Given a
  comonad morphism
  $\alpha \colon \mathsf{Q} \rightarrow \mathrm{Cosem}(\mathsf{T})$,
  corresponding by Lemma~\ref{lem:3} and Proposition~\ref{prop:2} to a
  functor $H$ as in:
  \begin{equation*}
    \cd[@!C@C-2em@-0.5em]{
      {\cat{Coalg}(\mathsf{Q})} \ar[rr]^{H} \ar[dr]_-{U^\mathsf{Q}} & &
      {{}^\mathsf{T}\cat{Set}}\rlap{ ,} \ar[dl]^-{{}^\mathsf{T}U} \\ &
      {\cat{Set}} }
  \end{equation*}
  the adjoint transpose
  $\bar \alpha \colon \mathsf{T} \rightarrow
  \mathsf{Coend}(U^\mathsf{Q})$ of $\alpha$ sends $t \in T(A)$ to
  $\bar \alpha(t) \colon U^\mathsf{Q} \Rightarrow A \cdot U^\mathsf{Q}$
  with components
  $\bar \alpha(t)_{(X,x)} = \dbr{t}_{H(X, x)} \colon X \rightarrow A
  \times X$.
\end{Rk}

\subsection{Relation to duals and Sweedler duals}
\label{sec:relat-duals-sweedl}

This completes our construction of the costructure--cosemantics
adjunction~\eqref{eq:4}; and in the rest of this section, we explain
its relation to the notions of~\cite{Katsumata2019Interaction}. The
main objects of study in \emph{loc.\ cit.}~are the
\emph{interaction laws} between a monad $\mathsf{T}$ and a comonad
$\mathsf{Q}$ on a category with products; these are natural families
of maps $TX \times QY \rightarrow X \times Y$ which are compatible
with the monad and comonad structures. In Section~3.4
of~\cite{Katsumata2019Interaction}, the authors show that such
monad--comonad interaction laws can also be expressed in terms of:
\begin{itemize}
\item Monad morphisms $\mathsf{T} \rightarrow \mathsf{Q}^\circ$, where
  $\mathsf{Q}^\circ$ is the \emph{dual monad} of $\mathsf{Q}$;
\item Comonad morphisms $\mathsf{Q} \rightarrow \mathsf{T}^\bullet$,
  where $\mathsf{T}^\bullet$ is the \emph{Sweedler dual comonad} of
  $\mathsf{T}$.
\end{itemize}
It may or may not be the case that the dual monad of a comonad, or the
Sweedler dual comonad of a monad, exist; however, they do always exist
when we are dealing with accessible monads and comonads on
$\cat{Set}$, and the definitions are as follows:

\begin{Defn}[Dual monad]
  \label{def:7}
  The \emph{dual} of an accessible comonad $\mathsf{Q}$ on $\cat{Set}$
  is the accessible monad $\mathsf{Q}^\circ$ with
  $Q^\circ(A) = [\cat{Set}, \cat{Set}](Q, A \cdot \mathrm{id})$, with
  unit map $\eta_A \colon A \rightarrow Q^\circ A$ given by
  \begin{equation*}
    a \qquad \mapsto\qquad Q \xrightarrow{\varepsilon} \mathrm{id}
    \xrightarrow{\nu_a} A \cdot \mathrm{id}
  \end{equation*}
  and with the Kleisli extension
  $u^\dagger \colon [\cat{Set}, \cat{Set}](Q, A \cdot \mathrm{id})
  \rightarrow [\cat{Set}, \cat{Set}](Q, B \cdot \mathrm{id})$ of
  $u \colon A \rightarrow [\cat{Set}, \cat{Set}](Q, B \cdot
  \mathrm{id})$ given by
  \begin{equation*}
    Q \xrightarrow{\tau} A \cdot \mathrm{id} \qquad \mapsto \qquad Q
    \xrightarrow{\delta} QQ \xrightarrow{\tau Q} A \cdot
    Q \xrightarrow{\spn{u_a}_{a \in A}} B \cdot \mathrm{id}\rlap{ .}
  \end{equation*}
  The assignment $\mathsf{Q} \mapsto \mathsf{Q}^\circ$ is the action on
  objects of the \emph{dual monad} functor
  $\cat{Cmd}_a(\cat{Set}) \rightarrow
  \cat{Mnd}_a(\cat{Set})^\mathrm{op}$, whose action on morphisms takes
  a comonad map $f \colon \mathsf{Q} \rightarrow \mathsf{P}$ to the
  monad map $\mathsf{P}^\circ \rightarrow \mathsf{Q}^\circ$ with
  components $\alpha \mapsto \alpha f$.
\end{Defn}

\begin{Defn}[Sweedler dual comonad]
  \label{def:2}
  The \emph{Sweedler dual} of an accessible monad $\mathsf{T}$ on
  $\cat{Set}$ is the accessible comonad $\mathsf{T}^\bullet$ providing
  the value at $\mathsf{T}$ of a right adjoint to the dual monad
  functor.
\end{Defn}

We now show that, in fact, these constructions relating accessible
monads and comonads are precisely the two directions of our
adjunction~\eqref{eq:4}.

\begin{Prop}
  \label{prop:20}
  For each $\mathsf{Q} \in \cat{Cmd}_a(\cat{Set})$, there is a monad
  isomorphism $\mathsf{Q}^\circ \cong \mathrm{Costr}(\mathsf{Q})$
  taking $\alpha \colon Q \Rightarrow A \cdot \mathrm{id}$ in
  $Q^\circ A$ to
  $\tilde \alpha \colon U^\mathsf{Q} \Rightarrow A \cdot U^\mathsf{Q}$
  in $\mathsf{Coend}(U^\mathsf{Q})(A)$ with components
  \begin{equation*}
    \tilde \alpha_{(X,x)} = X \xrightarrow{x} QX \xrightarrow{\alpha_X}
    A \times X\rlap{ .}
  \end{equation*}
  It follows that
  $\mathrm{Costr} \cong (\thg)^\circ \colon \cat{Cmd}_a(\cat{Set})
  \rightarrow \cat{Mnd}_a(\cat{Set})^\mathrm{op}$ and, consequently,
  that
  $\mathrm{Cosem} \cong (\thg)^\bullet \colon
  \cat{Mnd}_a(\cat{Set})^\mathrm{op} \rightarrow
  \cat{Cmd}_a(\cat{Set})$.
\end{Prop}

\begin{proof}
  Consider the category $\X$ whose objects are endofunctors of
  $\cat{Set}$, and whose morphisms $F \rightarrow F'$ are natural
  transformations
  $FU^\mathsf{Q} \Rightarrow F'U^\mathsf{Q} \colon
  \cat{Coalg}(\mathsf{Q}) \rightarrow \cat{Set}$. It is easy to see
  that $\mathsf{Coend}(U^\mathsf{Q})$ is equally the coendomorphism
  monad of the object $\mathrm{id}_{\cat{Set}} \in \X$. By transposing
  under the adjunction
  $(\thg) \circ G^\mathsf{Q} \dashv (\thg)\circ U^\mathsf{Q}$, we see
  that $\A$ is isomorphic to the co-Kleisli category $\X'$ of the
  comonad $(\thg) \circ \mathsf{Q}$ on $[\cat{Set}, \cat{Set}]$, and
  the coendomorphism monad of $\mathrm{id}_{\cat{Set}}$ in $\X'$ is
  easily seen to be $\mathsf{Q}^\circ$. Thus
  $\mathsf{Q}^\circ \cong \mathsf{Coend}(U^\mathsf{Q})$, and tracing
  through the correspondences shows this isomorphism to be given as in
  the statement of the result.
\end{proof}

\section{Calculating the cosemantics functor}
\label{sec:char-image-cosem}

\subsection{Cosemantics is valued in presheaf comonads}
\label{sec:cosem-valu-diagr}
In this section, we give an explicit calculation of the values of the
cosemantics functor from monads to comonads. As a first step towards
this, we observe that:

\begin{Prop}
  \label{prop:40}
  The cosemantics functor
  $\mathrm{Cosem} \colon \cat{Mnd}_a(\cat{Set})^\mathrm{op} \rightarrow
  \cat{Cmd}_a(\cat{Set})$ sends every monad to a presheaf comonad;
  whence it admits a factorisation to within
  isomorphism through
  $\mathsf{Q}_{(\thg)} \colon \cat{Cof} \rightarrow
  \cat{Cmd}_a(\cat{Set})$.
\end{Prop}
\begin{proof}
  Note that the second clause follows from the first and
  Proposition~\ref{prop:16}. To prove the first, let
  $\mathsf{T} \in \cat{Mnd}_a(\cat{Set})$. To show that
  $\mathrm{Cosem}(\mathsf{T})$ is a presheaf comonad, it suffices by
  Proposition~\ref{prop:6} to prove that its underlying endofunctor
  preserves connected limits. Since this endofunctor is engendered by
  ${}^\mathsf{T} U \colon {}^\mathsf{T} \cat{Set} \rightarrow
  \cat{Set}$ and its (limit-preserving) right adjoint, it suffices to
  show that ${}^\mathsf{T} U$ preserves connected limits. In fact, it
  creates them: for indeed, since a $\mathsf{T}$-comodel $\alg[S]$ in
  $\cat{Set}$ involves co-operations
  $\dbr{t} \colon S \rightarrow A \times S$ subject to suitable
  equations, and since each functor $A \times (\thg)$ preserves
  connected limits, it follows easily that, for any connected diagram
  of $\mathsf{T}$-comodels, the limit of the diagram of underlying
  sets bears a unique comodel structure making it the limit in the
  category of comodels.
\end{proof}

What we would like to do is to give an explicit description of the
factorisation of this proposition. Thus, at the level of
objects, we will describe for each accessible monad $\mathsf{T}$ on
$\cat{Set}$ a small category $\mathbb{B}_\mathsf{T}$ such that
$\mathrm{Cosem}(\mathsf{T}) \cong \mathsf{Q}_{\mathbb{B}_\mathsf{T}}$;
or equally, in light of Proposition~\ref{prop:19}, such that we have
an isomorphism in $\cat{CAT} / \cat{Set}$:
\begin{equation}\label{eq:16}
  \cd[@!C@C-1em@-0.5em]{
    {}^\mathsf{T} \cat{Set} \ar[dr]_-{{}^\mathbb{T} U}
    \ar[rr]^-{\cong} & & \mathbb{B}_\mathsf{T}\text-\cat{Set}\rlap{ .}
    \ar[dl]^-{U^{\mathbb{B}_\mathsf{T}}} \\ & \cat{Set}
  }
\end{equation}
We term this category $\mathbb{B}_\mathsf{T}$ the \emph{behaviour
  category} of $\mathsf{T}$.
\subsection{Behaviours and the final comodel}
\label{sec:oper-equiv}
The first step is to describe the object-set of the behaviour category
$\mathbb{B}_\mathsf{T}$ associated to an accessible monad
$\mathsf{T}$. By considering~\eqref{eq:16}, we see that this object-set can be found as the
image under $U^{\mathbb{B}_\mathsf{T}}$ of the final object of
$\mathbb{B}_\mathsf{T}\text-\cat{Set}$: and so equally as the underlying set of
the \emph{final comodel} of $\mathsf{T}$. While there are many
possible constructions of the final comodel---see, for
example,~\cite[Theorem~2.2]{Power2004From}
or~\cite[Lemma~4.6]{Pattinson2016Program}---we would like to give a
new one which fully exploits the fact that the structures we are
working with are comodels.

As is well known, when looking at coalgebraic structures,
characterising the final object is bound up with answering the
question of when states have the same observable behaviour. For
example, if $(g,n) \colon S \rightarrow V \times S$ and
$(g',n') \colon S' \rightarrow V \times S'$ are comodels of the theory
of $V$-valued input, we may say that states $s \in S$ and $s' \in S'$
are \emph{behaviourally equivalent} if they yield the same stream of
values:
\begin{equation*}
  (g(s), g(n(s)), g(n(n(s))), \dots) = 
  (g'(s'), g'(n'(s')), g'(n'(n'(s'))), \dots)\rlap{ .}
\end{equation*}
We may restate this property in more structural ways. Indeed, states
$s \in S$ and $s' \in S'$ are behaviourally equivalent just when any
of the following conditions holds:
\begin{itemize}
\item They are related by some \emph{bisimulation}, i.e., a relation
  $R \subseteq S \times S'$ whose projections $S \leftarrow R
  \rightarrow S'$ can be lifted to a span of comodels $\alg[S]
  \leftarrow \alg[R] \rightarrow \alg[S']$.
\item They become equal in some comodel $\alg[S]''$; i.e., there are
  comodel homomorphisms $q \colon \alg[S] \rightarrow \alg[S]'' \leftarrow
  \alg[S]' \colon q'$ such that $q(s) = q'(s')$;
\item They become equal in the final comodel
  $\alg[V^\mathbb{N}]$.
\end{itemize}
The correspondence between these conditions holds in much greater
generality; see, for example~\cite{Rutten1993On-the-foundations}.
However, for the comodels of an accessible monad
$\mathsf{T}$, there is a further, yet more intuitive, formulation: $s$
and $s'$ are behaviourally equivalent if, in running any
$\mathsf{T}$-computation $t \in T(A)$, we obtain the same $A$-value
by running $t$ with $\alg[S]$ from initial state $s$, as by running
$t$ with $\alg[S]'$ from initial state
$s'$. More formally, we have the following definition, which appears
to be novel---though it is closely related
to~\cite{Pattinson2015Sound}'s notion of \emph{comodel bisimulation}.

\begin{Defn}[Behaviour of a state]
  \label{def:29}
  Let $\mathsf{T}$ be an accessible monad and $\alg[S]$ a
  $\mathsf{T}$-comodel. The \emph{behaviour} $\beta_s$ of a state $s \in
  S$ is the family of functions
  \begin{equation*}
    (\beta_s)_A \colon T(A) \rightarrow A \qquad \qquad 
    t \mapsto \pi_1(\dbr{t}_{\alg[S]}(s))\rlap{ .}
  \end{equation*}
  Given $\mathsf{T}$-comodels $\alg[S]$ and $\alg[S]'$, we say that
  states $s \in S$ and $s' \in S'$ are \emph{operationally equivalent}
  (written $s \sim_o s'$) if $\beta_s = \beta_{s'}$.
\end{Defn}

We now show that operational equivalence has the same force as the
other notions of behavioural equivalence listed above.

\begin{Prop}
  \label{prop:24}
  Let $\mathsf{T}$ be an accessible monad and let $\alg[S], \alg[S]'$ be
  $\mathsf{T}$-comodels in $\cat{Set}$. For any states $s \in S$ and
  $s' \in S'$, the following conditions are equivalent:
  \begin{enumerate}[(i)]
  \item $s$ and $s'$ are operationally equivalent;
  \item $s \mathrel R s'$ for some bisimulation $R \subseteq S \times
    S'$ between $\alg[S]$ and $\alg[S]'$;
  \item $q(s) = q'(s')$ for some cospan of homomorphisms $q \colon
    \alg[S] \rightarrow \alg[S]'' \leftarrow \alg[S]' \colon q'$;
  \item $f(s) = f'(s')$ for $f \colon \alg[S] \rightarrow {\alg[B]}_\mathsf{T}
    \leftarrow \alg[S]' \colon f'$ the
    unique maps to the final comodel.
  \end{enumerate}
\end{Prop}

\begin{proof}
  For (i) $\Rightarrow$ (ii), we show that operational equivalence
  $\sim_o$ is a bisimulation between $\alg[S]$ and $\alg[S]'$; this
  means showing that, if
  $u_1 \sim_o u_2$ and $t \in T(A)$, then the
  co-operations $\dbr{t}_{\alg[S]}(s_1) = (a_1,s_1')$ and
  $\dbr{t}_{\alg[S]'}(s_2) = (a_2,s_2')$ satisfy
  $a_1 = a_2 \in A$ and $s_1' \sim_o s_2' \in S$. We have
  $a_1 = a_2$ since $s_1 \sim_o s_2$. To show
  $s_1' \sim_o s_2'$, consider any term $u \in T(B)$, and observe by~\eqref{eq:19} that
  \begin{equation*}
    \dbr{t(\lambda a.\, u)}(s_1) = \dbr{u}(s_1') \ \
    \text{and} \ \ 
    \dbr{t(\lambda a.\, u)}(s_2) = \dbr{u}(s_2')\rlap{ .}
  \end{equation*}
  Since $s_1 \sim_o s_2$, the left-hand sides above have the same
  first component; whence the same is true for the right-hand sides,
  so that $s_1' \sim_o s_2'$ as desired.

  The next two implications are standard. For (ii) $\Rightarrow$
  (iii), we take $\alg[S] \rightarrow \alg[S]'' \leftarrow \alg[S]'$
  to be the pushout of
  $\alg[S] \leftarrow \alg[R] \rightarrow \alg[S]'$; and for (iii)
  $\Rightarrow$ (iv), we postcompose
  $\alg[S] \rightarrow \alg[S]'' \leftarrow \alg[S]'$ with the unique
  comodel map $\alg[S]'' \rightarrow {\alg[B]}_\mathsf{T}$. Finally, for (iv)
  $\Rightarrow$ (i), note by the definition of comodel
  homomorphism that if $h \colon \alg[S] \rightarrow \alg[S]'$ then $\beta_s =
  \beta_{h(s)}$ for all $s \in S$. So if $f(s) = f'(s')$ as in (iv) then $\beta_{s} =
  \beta_{f(s)} = \beta_{f'(s')} = \beta_{s'}$ and so $s \sim_o s'$ as desired.
\end{proof}

From this result, we see that a final $\mathsf{T}$-comodel
can have at most one element of a given behaviour $\beta$. In fact, in
the spirit of~\cite[Theorem~4]{Kupke2009Characterising}, we may
characterise the final comodel as having exactly one element of each
behaviour $\beta$ which is \emph{admissible}, in the sense of being
the behaviour of some element of some comodel. It turns out that this
requirement can be captured purely algebraically.

\begin{Not}
  Let $\mathsf{T}$ be an accessible monad. Given terms $t \in T(A)$
  and $u \in T(B)$, we write $t \gg u$ for the term
  $t(\lambda a.\, u) \in T(B)$. Noting that $\gg$ is an associative
  operation, we may write $t \gg u \gg v$ for $(t \gg u) \gg v = t \gg
  (u \gg v)$, and so on.
\end{Not}
The intuition is that, if $t$ and $u$ are programs returning values in
$A$ and $B$, then $t \gg u$ is the program which first performs $t$,
then discards the return value and continues as $u$. The notation we
use is borrowed from Haskell, where \verb$t >> u$ is used with exactly
this sense.
\begin{Defn}[Admissible behaviour]
  \label{def:35}
  By an
  \emph{admissible behaviour} for $\mathsf{T}$, we mean a family of
  functions $\beta_A \colon TA \rightarrow A$, as $A$ ranges over sets, such that
  \begin{equation}\label{eq:52}
    a \in A  \!\implies\! \beta_A(a) = a\text{,} \quad 
    t \in TB, u \in (TA)^B \!\implies\! \beta_A(t(u)) = \beta_A(t \gg 
    u_{\beta_B(t)})\rlap{ .}
  \end{equation}
  We may drop the subscript in ``$\beta_A$''
  where this does not lead to ambiguity.
\end{Defn}
These conditions are intuitively reasonable; for example, the second
condition says that, if the result of running $t \in T(B)$ is
$b \in B$, then the result of running $t(u) \in T(A)$ coincides with
that of running $t$, discarding the return value, and then running
$u_b$.

\begin{Rk}
  \label{rk:2}
  If we have a presentation of the accessible monad $\mathsf{T}$ by an
  algebraic theory $\mathbb{T}$, then we can use~\eqref{eq:52} and
  induction on the structure of $\mathbb{T}$-terms to show that an
  admissible behaviour $\beta$ is determined by the values
  $\beta(\sigma_1 \gg \dots \gg \sigma_n) \in \abs{\sigma_n}$ for each
  non-empty list $\sigma_1, \dots, \sigma_n$ of generating operations
  in $\Sigma$. This is practically useful in computing the admissible
  behaviours of a theory.
\end{Rk}

\begin{Prop}
  \label{prop:33}
  The final comodel ${\alg[B]}_\mathsf{T}$ of an accessible monad $\mathsf{T}$ is
  the set of admissible behaviours with 
  $\dbr{t}_{\alg[B]_\mathsf{T}}(\beta) = (\beta(t), \partial_t \beta)$, where
  $\partial_t \beta(u) = \beta(t \gg u)$.
\end{Prop}
\begin{proof}
  Since every accessible monad can be presented by some algebraic
  theory, it follows from Remark~\ref{rk:2} that the admissible
  behaviours of $\mathsf{T}$ do indeed form a set. We must also show
  $\dbr{t}$ is well-defined, i.e., that if $\beta$ is admissible, then
  so is $\partial_t \beta$. For this, we calculate
  \begin{equation*}
    (\partial_t \beta)(u(v)) \!=\! \beta(t \gg u(v)) \!=\!
    \beta((t \gg u)(v))
    \!=\! \beta(t \gg
    u \gg v_{\beta(t \gg u)})\! =\! (\partial_t
    \beta)(u \gg v_{\partial_t\beta(u)})\rlap{ .}
  \end{equation*}

  We now show ${\alg[B]}_\mathsf{T}$ is a comodel, i.e., that the
  conditions of~\eqref{eq:21} hold. For the first condition
  $\dbr{a}_{{\alg[B]}_\mathsf{T}} = \nu_a$, we must show that
  $(\beta(a), \partial_a \beta) = (a, \beta)$ for all
  $\beta \in B_\mathsf{T}$: but $\beta(a) = a$ by~\eqref{eq:52}, while
  that $\partial_a \beta = \beta$ is clear since $a \gg (\thg)$ is the
  identity operator. For the second condition in~\eqref{eq:21}, we must
  show for any $\beta \in B_\mathsf{T}$ that
  $\dbr{t(u)}(\beta) = (\beta(t(u)), \partial_{t(u)} \beta)$ is equal
  to
  \begin{equation*}
    \spn{\dbr{u_a}}_{a \in A}(\dbr{t}(\beta)) =
    \spn{\dbr{u_a}}_{a \in A}(\beta(t), \partial_t \beta) = 
    \dbr{u_{\beta(t)}}(\partial_t \beta) = (\partial_t \beta(u_{\beta(t)}),
    \partial_{u_{\beta(t)}}(\partial_t \beta))\rlap{ .}
  \end{equation*}
  But in the first component
  $\beta(t(u)) = \beta( t \gg u_{\beta(t)}) = \partial_t
  \beta(u_{\beta(t)})$; while in the second,
  \begin{align*}
    \partial_{t(u)} \beta(v) &= \beta(t(u) \gg v)
    =  \beta(t(\lambda a.\, u_a \gg v)) \\ &= \beta(t \gg u_{\beta(t)}
    \gg v)= \partial_t \beta(u_{\beta(t)} \gg v) =
    (\partial_{u_{\beta(t)}}(\partial_t \beta))(v)
  \end{align*}
  as desired. So ${\alg[B]}_\mathsf{T}$ is a comodel.

  We now show that, for any comodel $\alg[S]$, there is a homomorphism
  $\beta_{(\thg)} \colon \alg[S] \rightarrow {\alg[B]}_\mathsf{T}$ given
  by $s \mapsto \beta_s$. For this to be well-defined, each $\beta_s$
  must be an admissible behaviour; but for any $t \in T(A)$, we have
  $\dbr{t}_{\alg[S]}(s) = (\beta_s(t), \dbr{t \gg
    \mathrm{id}}_{\alg[S]}(s))$, and so
  $\dbr{t(u)}_{\alg[S]}(s) = \dbr{u_{\beta_s(t)}}_{\alg[S]}(\dbr{t \gg
    \mathrm{id}}_{\alg[S]}(s)) = \dbr{ t \gg
    u_{\beta_s(t)}}_{\alg[S]}(s)$; now taking first components yields
  $\beta_s(t(u)) = \beta_s(t \gg u_{\beta_s(t)})$ as required. To show
  that $\beta_{(\thg)}$ is a homomorphism
  $\alg[S] \rightarrow {\alg[B]}_\mathsf{T}$, we calculate that
  \begin{align*}
    (1
    \times \beta_{(\thg)})\dbr{t}_{\alg[S]}(s) &= (1 \times
    \beta_{(\thg)})(\beta_s(t), \dbr{t \gg \mathrm{id}}_{\alg[S]}(s)) =
    (\beta_s(t), \beta_{\dbr{t \gg \mathrm{id}}_{\alg[S]}(s)}) \\ &= (\beta_s(t), \partial_t
    (\beta_s)) = \dbr{t}_{{\alg[B]}_\mathsf{T}}(\beta_s)\rlap{ .}
  \end{align*}
  It remains to show that $\beta_{(\thg)}$ is the \emph{unique}
  homomorphism $\alg[S] \rightarrow {\alg[B]}_\mathsf{T}$. But since
  $\dbr{t}_{{\alg[B]}_\mathsf{T}}(\beta) = (\beta(t), \partial_t
  \beta)$, the behaviour of any $\beta \in {\alg[B]}_\mathsf{T}$ is
  $\beta$ itself; and since as in Proposition~\ref{prop:24},
  homomorphisms preserve behaviour, any homomorphism
  $\alg[S] \rightarrow {\alg[B]}_\mathsf{T}$ must necessarily send $s$
  to $\beta_s$.
\end{proof}

\begin{Ex}
  \label{ex:33}While the final comodels of the algebraic theories
  considered so far are well-known, we illustrate the construction via
  admissible behaviours, exploiting Remark~\ref{rk:2} to compute them
  in each case.
  \begin{itemize}[itemsep=0.25\baselineskip]
  \item For $V$-valued input, an admissible behaviour $\beta$ is
    determined by the values
    $W_n \defeq \beta((\mathsf{read} \gg {})^n \mathsf{read}) \in V$
    for each $n \in \mathbb{N}$. But since the theory has no
    equations, \emph{any} such choice of values $W \in V^\mathbb{N}$
    yields an admissible behaviour. Thus the final comodel is
    $V^\mathbb{N}$, and we can read off from Proposition~\ref{prop:33}
    that
    $\dbr{\mathsf{read}} \colon V^\mathbb{N} \rightarrow V \times
    V^\mathbb{N}$ is given by $W \mapsto (W_0, \partial W)$, where
    $(\partial W)_i = W_{i+1}$.
  \item For $V$-valued output, an admissible behaviour $\beta$ is
    determined by the trivial choices
    $\beta(\mathsf{put}_{v_1} \gg \dots \gg \mathsf{put}_{v_n}) \in
    1$; whence there is a unique admissible behaviour, and the final
    comodel is the one-element set with the trivial co-operations.
  \item For $V$-valued read-only state, since
    $\mathsf{get} \gg (\thg)$ is the identity operator, an admissible
    behaviour is uniquely determined by the value
    $\beta(\mathsf{get}) \in V$. Any such choice yields an admissible
    behaviour, and so the final comodel is $V$ with the co-operation
    $\dbr{\mathsf{get}} = \Delta \colon V \rightarrow V \times V$.
  \item For $V$-valued state, $\mathsf{get} \gg (\thg)$ is again the
    identity operator, and so an admissible behaviour $\beta$ is
    determined by the values
    $\beta(\mathsf{put}_{v_1} \gg \dots \gg \mathsf{put}_{v_n} \gg
    \mathsf{get}) \in V$ for $v_1, \dots, v_n \in V$. When $n > 0$,
    the $\mathsf{put}$ axioms force
    $\beta(\mathsf{put}_{v_1} \gg \dots \gg \mathsf{put}_{v_n} \gg
    \mathsf{get}) = v_n$ and so $\beta$ is uniquely determined by
    $\beta(\mathsf{get}) \in V$. Thus again the final comodel is $V$,
    with the same $\dbr{\mathsf{get}}$ as before, and with
    $\dbr{\mathsf{put}} = \pi_1 \colon V \times V \rightarrow V$.
  \end{itemize}
\end{Ex}

\subsection{The behaviour category of an accessible monad}
\label{sec:behav-categ-an-1}
We observed in the proof of
Proposition~\ref{prop:33} that if $\mathsf{T}$ is an accessible monad
and $\alg[S]$ is a $\mathsf{T}$-comodel in $\cat{Set}$, then
$\dbr{t}_{\alg[S]}(s) = (\beta_s(t), \dbr{t \gg \mathrm{id}}_{\alg[S]}(s))$
for all $s \in \alg[S]$   and $t \in T(A)$. Thus $\alg[S]$ is completely
determined by the two functions
\begin{equation}\label{eq:58}
  \begin{aligned}
    S &\rightarrow B_\mathsf{T} & & & & & S \times T(1) & \rightarrow S \\
    s & \mapsto \beta_s &  & & & & (s, m) & \mapsto \dbr{m}_{\alg[S]}(s)\rlap{ .}
  \end{aligned}
\end{equation}
giving the behaviour of each state, together with what we might call
the \emph{dynamics} of the comodel: the right action of unary
operations on states. However, these two structures are not
independent. One obvious restriction is that the right action by
$m \in T(1)$ must send elements of behaviour $\beta$ to elements of
behaviour $\partial_m \beta$. However, due to~\eqref{eq:52} there is a
further constraint; the following definition is intended to capture
this.

\begin{Defn}[$\beta$-equivalence]
  \label{def:36}
  Let $\mathsf{T}$ be an accessible monad, and let $\beta$ be an admissible
  $\mathsf{T}$-behaviour. We say that unary operations
  $m,n \in T(1)$ are \emph{atomically $\beta$-equivalent} if there
  exists $v \in T(A)$ and $m', n' \in T(1)^A$ such that
  \begin{equation}\label{eq:46}
    m = v(m') \quad \text{and}\quad n
    = v(n') \quad \text{and} \quad m'_{\beta(v)}
    = n'_{\beta(v)}\rlap{ .}
  \end{equation}
  We write $\sim_\beta$ for the smallest equivalence relation on
  $T(1)$ which identifies atomically $\beta$-equivalent terms.
  Alternatively, $\sim_\beta$ is the smallest equivalence relation
  such that $t(m) \sim_\beta (t \gg m_{\beta(t)})$ for
  all $t \in T(A)$ and $m \in T(1)^A$.
\end{Defn}

\begin{Rk}
  \label{rk:4}
  If we have a presentation of the accessible monad $\mathsf{T}$ by an
  algebraic theory $\mathbb{T} = (\Sigma, \E)$, then we may simplify the task of
  computing the equivalence relation $\sim_\beta$ by observing that,
  by induction on the structure of $\mathbb{T}$-terms, each
  $m \in T(1)$ is $\sim_\beta$-equivalent to 
  $\sigma_1 \gg \dots \gg \sigma_n \gg \mathrm{id}$ for some
  $\sigma_1, \dots, \sigma_n \in \Sigma$.
\end{Rk}

The motivation for this definition is that $\sim_\beta$ will identify
two unary operations if and only if they act in the same way on any
state of behaviour $\beta$. The ``only if'' direction is
part (iii) of the following lemma; the ``if'' will be proved in
Corollary~\ref{cor:1}.
\begin{Lemma}
  \label{lem:13}
  Let $\beta$ be an admissible behaviour of 
  the accessible monad $\mathsf{T}$, and let $m,n,p \in T(1)$.
  \begin{enumerate}[(i)]
  \item If $m \sim_\beta n$ then $m.p \sim_\beta n.p$;
  \item If $m \sim_{\partial_p \beta} n$ then $p.m
    \sim_{\beta} p.n$;
  \item If $s \in \alg[S]$ is a state of behaviour
    $\beta$ and $m \sim_\beta n$,  then $\dbr{m}_{\alg[S]}(s) = \dbr{n}_{\alg[S]}(s)$.
  \end{enumerate}
\end{Lemma}
Here, and subsequently, if $m \in T(1)$ and $t \in T(A)$, we may write
$m.t$ for the substitution $m(t)$.
\begin{proof}
  For (i), if $m$ and $n$ are atomically
  $\beta$-equivalent via $v,m',n'$, then $m.p$ and $n.p$ are so via $v$, $(m'_a.p : a \in A)$ and $(n'_a.p : a
  \in A)$; whence $m \sim_\beta n$ implies $m(p) \sim_\beta
  n(p)$. (ii) is similar, observing that if $m,n$ are
  atomically $\partial_p \beta$-equivalent via $v,m',n'$, then $p.m)$
  and $p.n$ are atomically $\beta$-equivalent via $p.v, m',n'$.
  Finally, for (iii), if $m,n$ are atomically $\beta$-equivalent, then
  writing $s' = \dbr{v \gg \mathrm{id}}(s)$ we have 
  \begin{equation*}
    \dbr{m}(s) = \dbr{v(m')}(s) =
    \dbr{m'_{\beta(v)}}(s') = 
    \dbr{n'_{\beta(v)}}(s') = \dbr{v(n')}(s) =
    \dbr{n}(s)\text{ .}\qedhere \end{equation*}
\end{proof}

With this in place, we can now give the main definition of this section:

\begin{Defn}[Behaviour category of a monad]
  \label{def:38}
  Let $\mathsf{T}$ be an accessible monad.
  The \emph{behaviour category} $\mathbb{B}_\mathsf{T}$ of $\mathsf{T}$ has admissible
  behaviours as objects, and hom-sets 
  \begin{equation*}
    \mathbb{B}_\mathsf{T}(\beta, \beta') = \{ m \in T(1) \mid \beta' =
    \partial_m \beta \} \quot \mathord{\sim_\beta}\rlap{ .}
  \end{equation*}
  Identities are given by the neutral element of $T(1)$, and the
  composite of $m \colon \beta \rightarrow \beta'$ and
  $n \colon \beta' \rightarrow \beta''$ is
  $m.n \colon \beta \rightarrow \beta''$; note this is well-defined by
  Lemma~\ref{lem:13}(i--ii).
\end{Defn}

\begin{Thm}
  \label{thm:2}
  Given an accessible monad $\mathsf{T}$ with behaviour category
  $\mathbb{B}_\mathsf{T}$, the category of $\mathsf{T}$-comodels 
  is isomorphic to the category of left
  $\mathbb{B}_\mathsf{T}$-sets via an isomorphism commuting with the forgetful
  functors to $\cat{Set}$:
  \begin{equation}\label{eq:18}
    \cd[@!C@C-1em@-0.5em]{
      {}^\mathsf{T} \cat{Set} \ar[dr]_-{{}^\mathsf{T} U}
      \ar[rr]^-{\cong} & & \mathbb{B}_\mathsf{T}\text-\cat{Set}\rlap{ .} \ar[dl]^-{U^{\mathbb{B}_\mathsf{T}}} \\ & \cat{Set}
    }
  \end{equation}
  For a comodel $\alg[S]$, the corresponding left
  $\mathbb{B}_\mathsf{T}$-set $X_{\alg[S]}$ has underlying set $S$,
  projection to $B_\mathsf{T}$ given by the behaviour map $\beta_{(\thg)} \colon S \rightarrow B_\mathsf{T}$,
  and action given by
  \begin{equation}\label{eq:55}
    (s, \beta_s \xrightarrow{m} \partial_m \beta_s) \mapsto
    \dbr{m}_{\alg[S]}(s)\rlap{ ;}
    % X_{\alg[S]}(\beta) = \{\,s \in S \mid \beta_s = \beta\,\} \qquad
    % \text{and} \qquad X_{\alg[S]}(m \colon \beta \rightarrow \beta')
    % \colon s \mapsto \dbr{m}_{\alg[S]}(s)\rlap{ ;}
  \end{equation}
  for a left $\mathbb{B}_\mathsf{T}$-set $(X, p, \ast)$, the
  corresponding comodel $\alg[S]_X$
  has underlying set $X$ and
  \begin{equation}\label{eq:56}
    % S_X = \textstyle\sum_{\beta \in B} X(\beta) \qquad \text{and} \qquad
    \dbr{t}_{\alg[S]_X} \colon x \mapsto (\beta(t),
    t^\flat \ast x) \qquad \text{for all $x \in p^{-1}\beta$}
  \end{equation}
  where we write $t^\flat$ for $t \gg \mathrm{id}$.
\end{Thm}
\begin{proof}
  We concentrate on giving the isomorphism of categories at the level
  of objects; indeed, since it is to commute with the faithful
  functors to $\cat{Set}$, to obtain the isomorphism on arrows we need
  only check that a function lifts along ${}^\mathsf{T} U$ just when
  it lifts along $U^{\mathbb{B}_\mathsf{T}}$, and we leave this to the reader.
  
  Now, on objects, one direction is easy: the action map~\eqref{eq:55} of the presheaf
  associated to a comodel is well-defined by Lemma~\ref{lem:13}(iii),
  and is clearly functorial. % Furthermore, since comodel homomorphisms
  % preserve behaviours, and commute with the right action by unary
  % operations, each map of comodels $f \colon \alg[S] \rightarrow
  % \alg[T]$ induces a presheaf map $X_{\alg[S]} \rightarrow
  % X_{\alg[T]}$ which acts as $f$ does on elements. So we have a
  % functor $X_{(\thg)} \colon {}^\mathbb{T} \cat{Set} \rightarrow
  % \cat{Set}^{\mathbb{B}}$.
  In the converse direction, we must prove that the $\alg[S]_X$
  associated to a left $\mathbb{B}_\mathsf{T}$-set $X$ satisfies the comodel
  axioms in~\eqref{eq:21}. For the first axiom, we have that
  $\dbr{a}(s) = (\beta(a), a^\flat \ast s) = (a, s) =
  \nu_a(s)$ for all $s \in p^{-1}\beta$, since $\beta(a) = a$ and
  $a^\flat = \mathrm{id} \in T(1)$. For the second, given $s \in
  p^{-1}\beta$ we must show 
  $\dbr{t(u)}(s) = (\beta(t(u)), t(u)^\flat \ast
  s)$ is equal to
\begin{equation*}
    \spn{\dbr{u_a}}_{a \in A}(\dbr{t}(s)) =
    \spn{\dbr{u_a}}(\beta(t), t^\flat \ast s) = 
    \dbr{u_{\beta(t)}}(t^\flat \ast s) = (\partial_t \beta(u_{\beta(t)}),
    u_{\beta(t)}^\flat \ast (t^\flat \ast s))\rlap{ .}
  \end{equation*}
  But in the first component $\beta(t(u)) = \beta(
  t \gg u_{\beta(t)}) = \partial_t \beta(u_{\beta(t)})$; while in the
  second, since $u_{\beta(t)}^\flat \ast (t^\flat \ast s) =
  (u_{\beta(t)}^\flat \circ t^\flat) \ast s = 
  t^\flat(u_{\beta(t)}^\flat) \ast s$, %  is exactly as in 
  % 
  % For the inductive step, suppose $t =
  % \sigma(u)$ where each $u_i$ satisfies the inductive hypothesis. We
  % then have
  % \begin{align*}
  %   \dbr{t}_{\alg[S]_X}(s) &=
  %   \dbr{u_{\beta(\sigma)}}_{\alg[S]_X}(\tilde \sigma \ast s) = (\partial_\sigma
  %   \beta(u_{\beta(\sigma)}),
  %   \tilde{u}_{\beta(\sigma)} \ast (\tilde{\sigma} \ast s))
  %   \\ &= (\beta(\tilde \sigma.u_{\beta(\sigma)}),
  %   (\tilde{u}_{\beta(\sigma)} \circ \tilde{\sigma}) \ast s)
  %   = (\beta(t),
  %   (\tilde \sigma.\tilde{u}_{\beta(\sigma)}) \ast s)\rlap{ ,}
  % \end{align*}
  % and so it remains to show that $(\tilde
  % \sigma.\tilde{u}_{\beta(\sigma)}) \ast s = \tilde t \ast s$. It
  it suffices  to show that we have  $
  t^\flat({u}_{\beta(t)}^\flat) = t(u)^\flat \colon \beta \rightarrow
  \partial_t \beta$ in $\mathbb{B}_\mathsf{T}$; but 
  \begin{equation*}
    t^\flat( u^\flat_{\beta(t)}) = t(\lambda a.\,
    u^\flat_{\beta(t)}) \quad \text{and} \quad 
    t(u)^\flat = t(\lambda a.\,
    u^\flat_a)
  \end{equation*}
  and these two terms are clearly atomically $\beta$-equivalent, and
  so equal as maps in $\mathbb{B}_\mathsf{T}$. This shows that  $\alg[S]_X$ is a
  $\mathsf{T}$-comodel. % It is once again straightforward to show that
  % each presheaf map $\alpha \colon X \rightarrow Y$ induces a comodel
  % homomorphism acting as $\alpha$ does on elements, and so we have a
  % functor
  % $\alg[S]_{(\thg)} \colon \cat{Set}^\mathbb{B} \rightarrow
  % {}^\mathbb{T} \cat{Set}$.

  It remains to show that these two assignments are mutually inverse.
  For any comodel $\alg[S]$, the comodel $\alg[S]_{X_{\alg[S]}}$
  clearly has the same underlying set, but also the same
  comodel structure, since 
  \begin{equation*}
    \dbr{t}_{\alg[S]_{X_{\alg[S]}}}(s) = (\beta_s(t),
    t^\flat \ast s) = (\beta_s(t),
    \dbr{t \gg \mathrm{id}}_{\alg[S]}(s)) = \dbr{t}_{\alg[S]}(s)\rlap{ .}
  \end{equation*}
  On the other hand, for any $\mathbb{B}_\mathsf{T}$-set $X$, the
  $\mathbb{B}_\mathsf{T}$-set $X_{\alg[S]_X}$ has the same underlying set, but
  also the same projection to $B_\mathsf{T}$, since~\eqref{eq:56} exhibits each
  $x \in \alg[S]_X$ as having behaviour $p(x)$; and the same
  action map, since
  $m \ast_{X_{\alg[S]_X}} s = \dbr{m}_{\alg[S]_X}(s) = m \ast_X s$.
\end{proof}

\begin{Cor}
\label{cor:3}
Let $\mathsf{T}$ be an accessible monad. For each $\mathsf{T}$-admissible behaviour
$\beta$, there exists a comodel $\alg[\beta]$ freely generated by a
state of behaviour $\beta$; it has underlying set 
$T(1) \quot \sim_\beta$, co-operations 
$\dbr{t}_{\alg[\beta]}(m) = (\beta(m.t), m.t^\flat)$, and 
generating state
$\mathrm{id} \in \alg[\beta]$. Morphisms $\beta \rightarrow \beta'$ in the behaviour
category of $\mathsf{T}$ are in functorial bijection with comodel homomorphisms
$\alg[\beta'] \rightarrow \alg[\beta]$.
\end{Cor}
\begin{proof}
  We obtain $\alg[\beta]$ as the image of the representable functor
  $\mathbb{B}(\beta, \thg)$ under the equivalence of
  Theorem~\ref{thm:2}. The final
  clause follows from the Yoneda lemma.
\end{proof}

We now tie up a loose end by proving the converse to
Lemma~\ref{lem:13}(iii).

\begin{Cor}
  \label{cor:1}
  Let $\mathsf{T}$ be an algebraic theory, $\beta$ an admissible
  $\mathsf{T}$-behaviour, and $m,n \in T(1)$. We have $m \sim_\beta n$
  if, and only if, for every $\mathsf{T}$-comodel
  $\alg[S]$ and state $s \in \alg[S]$ of behaviour $\beta$, we
  have $\dbr{m}_{\alg[S]}(s) = \dbr{n}_{\alg[S]}(s)$.
\end{Cor}
\begin{proof}
  The ``only if'' direction is Lemma~\ref{lem:13}(iii). For the ``if''
  direction, consider the $\mathsf{T}$-comodel $\alg[\beta]$
  classifying states of behaviour $\beta$ and the state $\mathrm{id}
  \in \alg[\beta]$. Then we have $m =
  \dbr{m}_{\alg[\beta]}(\mathrm{id}) =
  \dbr{n}_{\alg[\beta]}(\mathrm{id}) = n$ in $T(1) \quot \sim_\beta$
  and so $m \sim_\beta n$ as desired.
\end{proof}

\begin{Rk}
  \label{rk:3}
  In computing the comodel $\alg[\beta]$ classifying an
  admissible behaviour $\beta$, the main problem is to determine
  suitable equivalence-class representatives for elements of the
  underlying set $T(1) \quot \sim_\beta$. Suppose we have a subset
  $\{\mathrm{id}\} \subseteq S \subseteq T(1)$ which we believe
  constitutes a set of such representatives. By
  Corollary~\ref{cor:3}, a necessary condition for this belief to be
  correct is that $S$ underlie a
  comodel $\alg[S]$ with
  \begin{equation}\label{eq:37}
    \dbr{t}_{\alg[S]}(s) = (\beta(s.t), s') \qquad \text{for some $s' \in S$ with
      $s' \sim_\beta s.t^\flat$ ;}
  \end{equation}
  of course, this $\alg[S]$ will then be the desired classifier for
  states of behaviour $\beta$.

  In fact, the preceding result tells us that this necessary condition
  is also \emph{sufficient}. Indeed, if $S$ bears a comodel structure
  satisfying~\eqref{eq:37}, then for each
  $m \in T(1)$ we have $\dbr{m}_{\alg[S]}(\mathrm{id}) \in S$ in
  the $\sim_\beta$-equivalence class of $m$; moreover, since
  $\mathrm{id} \in \alg[S]$ clearly has behaviour $\beta$, if
  $s \neq s' \in S$, then
  $\dbr{s}_{\alg[S]}(\mathrm{id}) = s \neq s' =
  \dbr{s'}_{\alg[S]}(\mathrm{id})$, whence $s \nsim_\beta s'$ by
  Corollary~\ref{cor:1}.
\end{Rk}

\begin{Ex}
  \label{ex:35}
  For each of our running examples of algebraic theories, we compute
  the comodels classifying each admissible behaviour, and so the
  behaviour category. For these examples, it is simple enough to
  find the classifying comodels directly without exploiting
  Remark~\ref{rk:3}.
  \begin{itemize}[itemsep=0.25\baselineskip]
  \item For $V$-valued input, the object-set of the behaviour category
    is $V^\mathbb{N}$, and for each behaviour $W \in V^\mathbb{N}$,
    the comodel $\alg[W]$ classifying this behaviour may be taken to
    have underlying set $\mathbb{N}$ with co-operation
    $\dbr{\mathsf{read}}(n) = (W_n, n+1)$; the universal state of
    behaviour $W$ is then $0 \in \alg[W]$. Morphisms
    $W \rightarrow W'$ in the behaviour category correspond to states
    of $\alg[W]$ of behaviour $W'$, and these can be identified with
    natural numbers $i$ such that $W'_n = W_{n+i}$ for all $k \in \mathbb{N}$.
  \item For $V$-valued output, the behaviour category has a single
    object $\ast$, and the comodel $\alg[V^\ast]$ classifying this
    unique behaviour has as underlying set the free monoid $V^\ast$,
    with co-operations $\dbr{\mathsf{write}_v}(W) = Wv$; the universal
    state is the empty word $\varepsilon \in \alg[V^\ast]$. Endomorphisms
    of $\ast$ in the behaviour category correspond to states of
    $\alg[V^\ast]$, so that the behaviour category is precisely the the
    one-object category corresponding to the monoid $V^\ast$.
  \item For $V$-valued read-only state, the behaviour category has
    object-set $V$, and the comodel classifying $v \in V$ is the
    one-element comodel $\alg[v]$ with
    $\dbr{\mathsf{get}}(\ast) = (v, \ast)$. Clearly, there are no
    non-identity homomorphisms between such comodels, so that the
    behaviour category is the \emph{discrete} category on the set $V$.
  \item For $V$-valued state, the behaviour category again has
    object-set $V$, while the comodel $\alg[v]$ classifying \emph{any}
    $v \in V$ is the final comodel $\alg[V]$, with universal state
    $v$. Maps $v \rightarrow v'$ in the behaviour category thus
    correspond to comodel homomorphisms $\alg[V] \rightarrow \alg[V]$,
    and since $\alg[V]$ is final, the only such is the identity. Thus
    the behaviour category is the \emph{codiscrete} category on
    $V$, with a unique arrow between every two
    objects.
  \end{itemize}
\end{Ex}

\subsection{Functoriality}
\label{sec:functoriality-1}
By Theorem~\ref{thm:2}, the assignment
$\mathsf{T} \mapsto \mathbb{B}_\mathsf{T}$ is the action on objects
of the desired factorisation of the cosemantics functor
$\cat{Mnd}_a(\cat{Set})^\mathrm{op} \rightarrow
\cat{Cmd}_a(\cat{Set})$ through
$\mathsf{Q}_{(\thg)} \colon \cat{Cof} \rightarrow
\cat{Cmd}_a(\cat{Set})$. We conclude this section by describing the
corresponding action on morphisms.

\begin{Prop}
  \label{prop:4}
  Let $\mathsf{T}$ and $\mathsf{R}$ be accessible monads on
  $\cat{Set}$. For each monad morphism $f \colon \mathsf{T}
  \rightarrow \mathsf{R}$, there is a
  cofunctor $\mathbb{B}_f \colon \mathbb{B}_\mathsf{R} \rightsquigarrow
  \mathbb{B}_\mathsf{T}$ which acts on objects by $\beta \mapsto
  f^\ast \beta$; and which, on morphisms, given $\beta \in
  \mathbb{B}_\mathsf{R}$, acts by sending $m
  \colon f^\ast \beta \rightarrow \partial_m (f^\ast \beta)$ in
  $\mathbb{B}_\mathsf{T}$ to $(\mathbb{B}_f)_{\beta}(m) \defeq f(m) \colon
  \beta \rightarrow \partial_{f(m)}\beta$ in $\mathbb{B}_{\mathsf{R}}$. 
%   as follows:
%   \begin{itemize}[itemsep=0.25\baselineskip]
%   \item On objects, we send $\beta \in \mathbb{B}_\mathsf{R}$ to
%     $f^\ast\beta \in \mathbb{B}_\mathsf{T}$.
% \item On maps, given $\beta \in \mathbb{B}_\mathsf{R}$ and $m
%   \colon f^\ast \beta \rightarrow \partial_m (f^\ast \beta)$ in
%   $\mathbb{B}_\mathsf{T}$, we define ${f^\ast}_{\!\beta}(m) = f(m) \colon
%   \beta \rightarrow \partial_{f(m)}\beta$. 
%   \end{itemize}
\end{Prop}
\begin{proof}
  $f^\ast \colon \mathbb{B}_\mathsf{R} \rightsquigarrow
  \mathbb{B}_\mathsf{T}$ is well-defined on objects by the lemma
  below. For well-definedness on maps, we must show that
  $m \sim_{f^\ast \beta} n$ in $T(1)$ implies $f(m) \sim_{\beta} f(n)$
  in $R(1)$. Clearly it suffices to do so when $m$ and $n$ are
  atomically $f^\ast\beta$-equivalent via terms $v \in T(A)$ and
  $m',n' \in T(1)^A$ satisfying
  \begin{equation*}
    m = v(m') \qquad \text{and} \qquad n = v(n') \qquad \text{and}
    \qquad m'_{f^\ast\beta(v)} = n'_{f^\ast\beta(v)}\rlap{ .}
  \end{equation*}
  But since $f^\ast \beta(v) = \beta(f(v))$, it follows that
  $f(v) \in R(A)$ and $f(m'_{(\thg)}), f(n'_{(\thg)}) \in R(1)^A$
  witness $f(m)$ and $f(n)$ as atomically $\beta$-equivalent, as
  required. We must also check the three cofunctor axioms. The first
  holds since
  $f^\ast(\partial_{f(m)}\beta)(t) = (\partial_{f(m)}\beta)(f(t)) =
  \beta(f(m).f(t)) = \beta(f(m.t)) = f^\ast\beta(m.t) =
  \partial_m(f^\ast \beta)(t)$. The other two are immediate since $f$
  preserves substitution.
\end{proof}

In the following lemma, recall from Definition~\ref{def:4} that each
monad morphism $f \colon \mathsf{T} \rightarrow \mathsf{R}$ induces a
functor on comodels
$f^\ast \colon {}^\mathsf{R} \cat{Set} \rightarrow {}^\mathsf{T}
\cat{Set}$.

\begin{Lemma}
  \label{lem:5}
  Let $f \colon \mathsf{T} \rightarrow \mathsf{R}$ in
  $\cat{Mnd}_a(\cat{Set})$, and let $\alg[S]$ be a
  $\mathsf{R}$-comodel. If $s \in \alg[S]$ has $\mathsf{R}$-behaviour
  $\beta$, then $s \in f^\ast \alg[S]$ has $\mathsf{T}$-behaviour
  $f^\ast \beta$ where $(f^\ast \beta)(t) = \beta(f(t))$.
\end{Lemma}

\begin{proof}
  For any $t \in T(A)$, we have $\pi_1(\dbr{t}_{f^\ast \alg[S]}(s)) = \pi_1(\dbr{f(t)}_{\alg[S]}(s)) = \beta(f(t))$.
\end{proof}

We now prove that the cofunctor $\mathbb{B}_f$ of
Proposition~\ref{prop:4} does indeed describe the action on morphisms
of the cosemantics functor.

\begin{Thm}
  \label{thm:1}
  The functor
  $\mathbb{B}_{(\thg)} \colon \cat{Mnd}_a(\cat{Set})^\mathrm{op}
  \rightarrow \cat{Cof}$ taking each accessible monad $\mathsf{T}$ to its
  behaviour category $\mathbb{B}_\mathsf{T}$, and each map of accessible monads
  $f \colon \mathsf{T} \rightarrow \mathsf{R}$ to the cofunctor
  $\mathbb{B}_f \colon \mathbb{B}_\mathsf{R} \rightsquigarrow
  \mathbb{B}_\mathsf{T}$ of Proposition~\ref{prop:4}, yields a within-isomorphism
  factorisation 
\begin{equation*}
  \cd[@C+0.2em]{
    \twocong[0.66]{dr}{} & \cat{Cof} \ar[d]^-{\mathsf{Q}_{(\thg)}} \\
    \cat{Mnd}_a(\cat{Set})^\mathrm{op} \ar[r]_-{\mathrm{Cosem}}
    \ar@{-->}[ur]^-{\mathbb{B}_{(\thg)}} & \cat{Cmd}_a(\cat{Set})\rlap{ .}
  }
\end{equation*}
\end{Thm}
\begin{proof}
  It suffices to show that, for any map of accessible monads $f \colon \mathsf{T} \rightarrow
  \mathsf{R}$, the associated cofunctor 
  $\mathbb{B}_f \colon \mathbb{B}_\mathbb{R} \rightsquigarrow
  \mathbb{B}_{\mathsf{T}}$ renders
  commutative the square
  \begin{equation*}
    \cd{
      {{}^\mathsf{R}\cat{Set}} \ar[r]^-{\cong} \ar[d]_{f^\ast} &
      {\mathbb{B}_\mathsf{R}\text-\cat{Set}} \ar[d]^{\Sigma_{\mathbb{B}_f}} \\
      {{}^\mathsf{T}\cat{Set}} \ar[r]^-{\cong} &
      {\mathbb{B}_\mathsf{T}\text-\cat{Set}}
    }
  \end{equation*}
  whose horizontal edges are the isomorphisms of Theorem~\ref{thm:2}.
  But indeed, given an $\mathsf{R}$-comodel $\alg[S]$, its image
  around the lower composite is by Lemma~\ref{lem:5} the $\mathbb{B}_\mathsf{T}$-set with
  underlying set $S$ and projection and action maps
  \begin{align*}
    S & \rightarrow B_{\mathsf{T}} & \textstyle\sum_{s \in S} T(1)
    \quot\sim_{f^\ast(\beta_s)} & \rightarrow S \\
    s & \mapsto f^\ast(\beta_s) & (s, m) & \mapsto \dbr{f(m)}_{\alg[S]}(s)\rlap{ .}
  \end{align*}
  On the other hand, the upper composite first sends $\alg[S]$ to the
  $\mathbb{B}_\mathsf{R}$-set with underlying set $S$ and projection
  and action maps
  \begin{align*}
    S & \rightarrow B_{\mathsf{R}} & \textstyle\sum_{s \in S} R(1)
    \quot\sim_{\beta_s} & \rightarrow S \\
    s & \mapsto \beta_s & (s, n) & \mapsto \dbr{n}_{\alg[S]}(s)\rlap{ ;}
  \end{align*}
  and then applies $\Sigma_{\mathbb{B}_f}$, which, by the definition
  of $\Sigma_{(\thg)}$ and of $\mathbb{B}_f$, yields the same
  $\mathbb{B}_{\mathsf{T}}$-set as above.
\end{proof}

\begin{Ex}
  \label{ex:3}
  Let $h \colon V \rightarrow W$ be a function, and let
  $f \colon \mathbb{T}_1 \rightarrow \mathbb{T}_2$ be the associated
  interpretation of $V$-valued output into $W$-valued state of
  Example~\ref{ex:4}. The induced cofunctor $f^\ast \colon
  \mathbb{B}_{\mathbb{T}_2} \rightarrow \mathbb{B}_{\mathbb{T}_1}$ has
  as domain the codiscrete category on $W$, and as codomain, the
  monoid $V^\ast$ seen as a one-object category. On objects, $f^\ast$
  acts in the unique possible way; while on morphisms, given $w \in
  \mathbb{B}_{\mathbb{T}_2}$ and a map $f^\ast(w) \rightarrow \ast$ in
  $\mathbb{B}_{\mathbb{T}_1}$---corresponding to an
  element $v_1 \dots v_n \in V^\ast$---we have $f^\ast_w(v_1\dots
  v_n)$ in $\mathbb{B}_{\mathbb{T}_2}$ given by the unique map $w \rightarrow v_n$.
\end{Ex}

\begin{Ex}
  \label{ex:11}
  Let $h \colon W \rightarrow V$ be a function between sets, and let
  $f \colon \mathbb{T}_1 \rightarrow \mathbb{T}_2$ be the associated
  interpretation of $V$-valued read-only state into $W$-valued state of
  Example~\ref{ex:5}. The induced cofunctor $f^\ast \colon
  \mathbb{B}_{\mathbb{T}_2} \rightarrow \mathbb{B}_{\mathbb{T}_1}$ has
  as domain the codiscrete category on $W$, and as codomain the
  discrete category on $V$. On objects, $f^\ast$
  acts by $w \mapsto h(v)$, while on morphisms it acts in the unique possible way.
\end{Ex}

\section{Calculating the costructure functor}
\label{sec:calc-costr-funct}

\subsection{The behaviour category of an accessible comonad}
\label{sec:behav-categ-an-2}
In this section, we give an explicit calculation of the
costructure functor from comonads to monads. Much as for the
cosemantics functor, we will see that cosemantics takes its values in
presheaf monads, and will explicitly associate to each accessible
comonad $\mathsf{Q}$ a small category $\mathbb{B}_\mathsf{Q}$, the
\emph{behaviour category}, such
that $\mathrm{Costr}(\mathsf{Q})$ is the presheaf monad of
$\mathbb{B}_\mathsf{Q}$. We begin with some preliminary observations.

\begin{Not}
  Let $\mathsf{Q}$ be an accessible comonad on $\cat{Set}$. For
  each $x \in Q1$, we write $\iota_x \colon Q_x \rightarrow Q$ for the
  inclusion of the subfunctor with
  \begin{equation*}
    Q_x(A) = \{a \in QA : (Q!)(a) = x \text{ in } Q1\}\rlap{ .}
  \end{equation*}
  We also write $\varepsilon_x \colon Q_x \rightarrow 1$ and
  $\delta_x \colon Q_x \rightarrow Q_x Q$ for the natural
  transformations such that
  $\varepsilon_x = \varepsilon \circ \iota_x$ and
  $\iota_x Q \circ \delta_x = \delta \circ \iota_x$; to see that
  $\delta \circ \iota_x$ does indeed factor through $\iota_x Q$, we
  note that, for any $a \in Q_xA$, the element $\delta(a) \in QQA$
  satisfies
  $(Q!)(\delta(a)) = (Q!)(Q\varepsilon(\delta(a))) = (Q!)(a) = x$ so
  that $\delta(a) \in Q_xQA$ as desired.
\end{Not}
\begin{Lemma}
  \label{lem:8}
  Let $\mathsf{Q}$ be an accessible comonad on $\cat{Set}$.
  \begin{enumerate}[(i),itemsep=0.25\baselineskip
    ]
  \item The inclusions $\iota_x \colon Q_x \rightarrow Q$ exhibit $Q$ as
  the coproduct $\sum_{x \in Q1}Q_x$;
  \item Any natural transformation $f \colon Q_x \rightarrow \sum_i F_i$ into a
    coproduct factors through exactly one coproduct injection $\nu_i
    \colon F_i \rightarrow \sum_i F_i$.
  \end{enumerate}
\end{Lemma}
\begin{proof}
  (i) holds as each $QA$ is the coproduct of the $Q_x A$'s, and
  coproducts in $[\cat{Set}, \cat{Set}]$ are componentwise.
  For (ii), the component
  $f_1 \colon \{x\} \rightarrow \sum_{i \in I} F_i 1$ of $f$ clearly factors
  through just one $\nu_i$; now naturality of $f$ with respect to the unique maps $! \colon A
  \rightarrow 1$
  shows that each $f_A$ factors through just the same $\nu_i$.
\end{proof}

\begin{Defn}[Behaviour category of a comonad]
\label{def:22}
  Let $\mathsf{Q}$ be an accessible comonad on $\cat{Set}$. The
  \emph{behaviour category} of $\mathsf{Q}$ is the small category
  $\mathbb{B}_\mathsf{Q}$ 
  in which:
  \begin{itemize}
  \item Objects are elements of $Q1$;
  \item Maps with domain $x \in Q1$ are natural transformations
    $\tau \colon Q_x \rightarrow \mathrm{id}$, and the codomain of
    such a $\tau$ is
    determined by the following factorisation, whose (unique) existence is
    asserted by
 Lemma~\ref{lem:8}:
    \begin{equation}\label{eq:34}
      \cd[@-0.3em]{
        Q_x \ar[d]_-{\delta_x} \ar@{-->}[r]^-{\tau^\sharp} & Q_{\mathrm{cod}(\tau)}
        \ar@{>->}[d]^-{\iota_{\mathrm{cod}(\tau)}} \\
        Q_x Q \ar[r]^-{\tau Q} & Q\rlap{ ;}
      }
    \end{equation}
  \item The identity on $x \in Q1$ is $\varepsilon_x \colon Q_x
    \rightarrow \mathrm{id}$;
  \item Binary composition is given as follows, 
where $\tau^\sharp$ is the factorisation in~\eqref{eq:34}:
    \begin{equation*}
      \smash{(Q_{\mathrm{cod}(\tau)} \xrightarrow{\upsilon} \mathrm{id}) \circ
        (Q_x \xrightarrow{\tau} \mathrm{id}) = (Q_x
    \xrightarrow{\tau^\sharp} Q_{\mathrm{cod}(\tau)}
    \xrightarrow{\upsilon} \mathrm{id})}\rlap{ .}
\end{equation*}
\end{itemize}
The axioms expressing unitality and associativity of
composition follow from the familiar and easily-checked identities
$\varepsilon_{\mathrm{cod}(\tau)} \circ \tau^\sharp = \tau$,
$\varepsilon_x^\sharp = 1_{Q_x}$ and $(\upsilon \circ
\tau^\sharp)^\sharp = \upsilon^\sharp \circ \tau^\sharp$.
\end{Defn}
We now show that the costructure monad associated to an accessible
comonad $\mathsf{Q}$ is isomorphic to the presheaf monad
$\mathsf{T}_{\mathbb{B}_\mathsf{Q}}$. In light of
Proposition~\ref{prop:20}, it suffices to construct an isomorphism
between $\mathsf{T}_{\mathbb{B}_\mathsf{Q}}$ and the \emph{dual monad}
$\mathsf{Q}^\circ$ of Definition~\ref{def:7}.
\begin{Prop}
  \label{prop:23}
  Let $\mathsf{Q}$ be an accessible comonad.
  For any $\tau \colon Q \rightarrow A \cdot \mathrm{id}$ and any
  $x \in Q1$, there is a unique $a_x \in A$ and 
  $\tau_x \colon Q_x \rightarrow \mathrm{id}$ for which we have a factorisation
  \begin{equation}\label{eq:29}
    \cd[@-0.3em]{
      Q_x \ar@{ >->}[d]_-{\iota_x} \ar@{-->}[r]^-{\tau_x} &
      \mathrm{id} \ar@{ >->}[d]_-{\nu_{a_x}} \\
      Q \ar[r]^-{\tau} & A \cdot \mathrm{id}\rlap{ .}
    }
  \end{equation}
  In this way, we obtain a monad isomorphism
  $\theta \colon \mathsf{Q}^\circ \cong \mathsf{T}_{\mathbb{B}_{\mathsf{Q}}}$ with components
  \begin{equation}\label{eq:25}
    \begin{aligned}
      \theta_A \colon [\cat{Set}, \cat{Set}](Q, A \cdot \mathrm{id}) &\rightarrow
      \textstyle\prod_{x
        \in \mathbb{B}_\mathsf{Q}}(A \times (\mathbb{B}_\mathsf{Q})_x)\\
      \tau & \mapsto \lambda x.\, (a_x, \tau_x)\rlap{ .}
    \end{aligned}
  \end{equation}
\end{Prop}
\begin{proof}
  The existence and uniqueness of the factorisation~\eqref{eq:29}
  is a consequence of Lemma~\ref{lem:8}(ii); now that the
  induced maps~\eqref{eq:25} constitute a natural isomorphism follows
  from Lemma~\ref{lem:8}(i).
  It remains to show that these maps are the components of a
  monad isomorphism $\mathsf{Q}^\circ \cong
  \mathsf{T}_{\mathbb{B}_\mathsf{Q}}$.
  
  For compatibility with units, we have on the one hand that
  $\eta^{\mathsf{T}_{\smash{\mathbb{B}_{\mathsf{Q}}}}}_A(a) = \lambda x.\, (a,
  \varepsilon_x)$. On the other hand, 
  $\eta^{\mathsf{Q}^\circ}_A(a) = \nu_a \circ \varepsilon \colon Q
  \rightarrow \mathrm{id} \rightarrow A \cdot \mathrm{id}$, whose
  factorisation as
  in~\eqref{eq:29} is clearly $\eta^{\mathsf{Q}^\circ}_A(a) \circ \iota_x = \nu_a \circ \varepsilon_x$,
  so that
  $\theta_A(\eta^{\mathsf{Q}^\circ}_A(a)) = \lambda x.\, (a,
  \varepsilon_{x})$ as desired.

  We now
  show compatibility with multiplication. To this end, consider an
  element $\sigma \colon Q \rightarrow Q^\circ A \cdot \mathrm{id}$ of
  $Q^\circ Q^\circ A$. For each $x \in Q1$, we have an element
  $\tau_x \in Q^\circ A$ and a natural transformation
  $\sigma_x \colon Q_x \rightarrow 1$ rendering commutative the square
  to the left in
  \begin{equation}\label{eq:26}
    \cd[@-0.3em]{
      Q_x \ar[r]^-{\sigma_x} \ar@{ >->}[d]_-{\iota_x} &
      \mathrm{id} \ar@{ >->}[d]^-{\nu_{\tau_x}} &  &
      Q_{y} \ar[r]^-{\tau_{xy}} \ar@{ >->}[d]_-{\iota_{}} &
      \mathrm{id} \ar@{ >->}[d]^-{\nu_{a_{xy}}}\\
      Q \ar[r]^-{\sigma} & Q^\circ A \cdot \mathrm{id} & &
      Q \ar[r]^-{\tau_x} & A \cdot \mathrm{id}\rlap{ .}
    }
  \end{equation}
  Considering now $\tau_x \in Q^\circ A$, we have for
  each $y \in Q1$ an element
  $a_{xy} \in A$ and 
  $\tau_{xy} \colon Q_{y} \rightarrow 1$ rendering
  commutative the square above right. With this notation, the
  composite $\mu_A^{\mathsf{T}_{\smash{\mathbb{B}_\mathsf{Q}}}} \circ
  (\theta\theta)_A$ acts on $\sigma
  \in Q^\circ Q^\circ A$ via
  \begin{equation*}
       \smash{\sigma \ \xmapsto{\ \ (\theta\theta)_A\ \ }\  \lambda x.\, (\sigma_x, \lambda y.\,
       (\tau_{xy}, a_{xy})) \ \xmapsto{\ \
         \mu_A^{\mathsf{T}_{\smash{\mathbb{B}_\mathsf{Q}}}}\ \ }\  \lambda x.\, (\tau_{x,
         \mathrm{cod}(\sigma_x)} \circ \sigma_x^\sharp, a_{x, \mathrm{cod}(\sigma_x)})\rlap{ .}}
  \end{equation*}
  We must show that this is equal to the image of $\sigma$ under the composite $\theta_A \circ
  \mu^{\mathsf{Q}^\circ}_A$. From the description
  of $\mathsf{Q}^\circ$ in Definition~\ref{def:7}, we can read off
  that $\mu^{\mathsf{Q^\circ}}_A(\sigma) \in Q^\circ A$ is the lower composite in
  the diagram
  % \begin{equation}\label{eq:27}
  %   \mu^{\mathsf{Q^\circ}}_A(\sigma) = Q \xrightarrow{\delta} QQ \xrightarrow{\sigma Q} Q^\circ A \cdot Q
  %   \xrightarrow{\mathsf{ev}} A \cdot \mathrm{id} \qquad \text{ in
  %     $Q^\circ A$}
  % \end{equation}
  \begin{equation*}
    \cd[@-0.4em]{
      & Q_x \ar[r]^-{\sigma_x^\sharp} \ar[d]|-{\delta_x}
      \ar[dl]_-{\iota_x} &
      Q_{\mathrm{cod}(\sigma_x)}
      \ar[d]|-{\iota_{\mathrm{cod}(\sigma_x)}} \ar[r]^-{\tau_{x,
          \mathrm{cod}(\sigma_x)}} &
      \mathrm{id} \ar[d]^-{\nu_{a_{x, \mathrm{cod}(\sigma_x)}}} \\
      Q \ar[dr]_-{\delta} &
      Q_x Q \ar[r]^-{\sigma_x Q} \ar[d]|-{\iota_x Q} &
      Q  \ar[d]|-{\nu_{\tau_x}} \ar[r]^-{\tau_x} &
      Q \cdot \mathrm{id} \rlap{ .}\\ &
      QQ \ar[r]^-{\sigma Q} & Q^\circ A \cdot Q \ar[ur]_-{\mathsf{ev}}
    }
  \end{equation*}
  where $\mathsf{ev}$ is unique such that
  $\mathsf{ev} \circ \nu_\tau = \tau$ for all $\tau \in Q^\circ A$. To
  calculate the image of
  $\mu^{\mathsf{Q^\circ}}_A(\sigma)$ under $\theta_A$,
  we observe that, in the displayed diagram, the far left region is
  the definition of $\delta_x$, the two upper squares are instances
  of~\eqref{eq:34} and~\eqref{eq:29}, the lower square is $(\thg)Q$ of
  another instance of~\eqref{eq:29}, and the triangle is definition of
  $\mathsf{ev}$. So by unicity in~\eqref{eq:29}, 
  $\theta_A(\mu^{\mathsf{Q^\circ}}_A(\sigma)) = \lambda x.\, (\tau_{x,
    \mathrm{cod}(\sigma_x)} \circ \sigma_x^\sharp, a_{x,
    \mathrm{cod}(\sigma_x)})$ as required.
\end{proof}

\subsection{Functoriality}
\label{sec:functoriality}

We now describe the manner in which the passage from an accessible comonad 
to its behaviour category is functorial.

\begin{Not}
  Let $f \colon \mathsf{P} \rightarrow \mathsf{Q}$ be a morphism of
  accessible comonads on $\cat{Set}$ and let $x \in Q1$. We write $f_x
  \colon P_x \rightarrow Q_{fx}$ for the unique natural transformation
  (whose unique existence follows from Lemma~\ref{lem:8}) such that $f
  \circ \iota_x = \iota_{fx} \circ f_x \colon P_x \rightarrow Q$.
\end{Not}

\begin{Prop}
  \label{def:19}
  Each morphism $f \colon \mathsf{P} \rightarrow \mathsf{Q}$ of
  accessible comonads on $\cat{Set}$ induces a cofunctor $\mathbb{B}_f
  \colon \mathbb{B}_\mathsf{P} \rightsquigarrow \mathbb{B}_\mathsf{Q}$
  on behaviour categories with action on objects $f_1 \colon P1
  \rightarrow Q1$, and with, for each $x \in P1$, the action on homs
  $(\mathbb{B}_\mathsf{Q})_{fx} \rightarrow
  (\mathbb{B}_\mathsf{P})_{x}$ given by $\tau \mapsto \tau \circ f_x$.
\end{Prop}
\begin{proof}
  We first dispatch axiom (ii) for a cofunctor, which follows by the
  calculation that
  $\varepsilon^\mathsf{Q}_{fx} \circ f_x = \varepsilon^\mathsf{Q} \circ
  \iota_{fx} \circ f_x = \varepsilon^\mathsf{Q} \circ f \circ \iota_x
  = \varepsilon^\mathsf{P} \circ \iota_x = \varepsilon_x^\mathsf{P}
  \colon P_x \rightarrow \mathrm{id}$.
  We next deal with axiom (i). Let $\tau
  \colon Q_{fx} \rightarrow \mathrm{id}$ be an element of
  $(\mathbb{B}_\mathsf{Q})_{fx}$ with image $\tau \circ f_x$ in
  $(\mathbb{B}_{\mathsf{P}})_x$. We must show that $y \defeq
  \mathrm{cod}(\tau \circ f_x)$ is sent by $f$ to $z \defeq
  \mathrm{cod}(\tau)$. To this end, consider the diagram to the left in:
  \begin{equation*}
    \cd[@!@-2.2em@C+0.4em]{
      P_x \ar[rr]^-{(\tau \circ f_x)^\sharp} \ar[dd]_-{\delta_x}
      \ar[dr]^-{f_x} & &
      P_{y} \ar@{ >->}'[d][dd]^-{\iota_{y}} \\ &
      Q_{fx} \ar[rr]^(0.3){\tau^\sharp} \ar[dd]^(0.7){\delta_{fx}} & &
      Q_z \ar@{ >->}[dd]^-{\iota_z} \\
      P_x P \ar'[r][rr]^-{(\tau \circ f_x) P} \ar[dr]_-{f_x f} & & 
      P \ar[dr]^-{f} \\ &
      Q_{fx} Q \ar[rr]^-{\tau Q} & & Q
    } \qquad \qquad 
    \cd[@-0.3em]{
      P_x \ar[r]^-{(\tau \circ f_x)^\sharp} \ar[d]_-{f_x} & P_{y}
      \ar@{-->}[d]^-{f_{y}}
      \ar@{ >->}[r]^-{\iota_{y}} & P \ar[d]^-{f}\\
      Q_{fx} \ar[r]^-{\tau^\sharp} & Q_z \ar@{ >->}[r]^-{\iota_z} & Q\rlap{ .}
    }
  \end{equation*}
  The front and back faces are instances of~\eqref{eq:34}, the left
  face commutes since $f$ is a comonad morphism, and the bottom face
  commutes by naturality. We can thus read off that the outside of the diagram to the right
  commutes; as such, its upper composite (clearly) factors
  through $\iota_z$, but also through $\iota_{fy}$, since
  $f \circ \iota_{y} = \iota_{fy} \circ f_{y}$: whence by
  Lemma~\ref{lem:8}(ii) we have $z = f(y)$, giving the first cofunctor
  axiom.

  Finally, we address cofunctor axiom (iii). Note that we can now
  insert $f_{y}$ into the diagram right above; whereupon the right
  square commutes by definition of $f_{y}$, and the left square
  since it does so on postcomposition by the monic $\iota_z$.
  Postcomposing this left-hand square with some $\sigma \colon Q_z \rightarrow
  \mathrm{id}$ yields the final cofunctor axiom.
\end{proof}

\begin{Prop}
  \label{prop:25}
  For each morphism $f \colon \mathsf{P} \rightarrow \mathsf{Q}$ of
  accessible comonads, we have a commuting square of monad morphisms:
  \begin{equation}\label{eq:33}
    \cd[@-0.3em]{
      \mathsf{Q}^\circ \ar[r]^-{\theta} \ar[d]_-{f^\circ} &
      \mathsf{T}_{\mathbb{B}_\mathsf{Q}}
      \ar[d]^-{\mathsf{T}_{\mathbb{B}_f}} \\
      \mathsf{P}^\circ \ar[r]^-{\theta} &
      \mathsf{T}_{\mathbb{B}_\mathsf{P}}\rlap{ .}
    }
  \end{equation}
\end{Prop}
\begin{proof}
  For each $\tau \colon Q \rightarrow A \cdot \mathrm{id}$ in
  $Q^\circ A$, and each $x \in P1$, we have a diagram
  \begin{equation*}
    \cd{
      {P_x} \ar[r]^-{f_x} \ar@{ >->}[d]_{\iota_x} &
      {Q_{fx}} \ar@{ >->}[d]^{\iota_{fx}} \ar[r]^-{\tau_{fx}} & \mathrm{id}
      \ar@{ >->}[d]^-{\nu_{a_{fx}}}\\
      {P} \ar[r]^-{f} &
      {Q} \ar[r]^-{\tau} & A \cdot \mathrm{id}
    }
  \end{equation*}
  whose right square is as in Proposition~\ref{prop:23}, and whose
  left square is the definition of $f_x$. It thus follows that the image of $\tau$ under $(T_{\mathbb{B}_f})_A \circ
  \theta_A^\mathsf{Q}$ is $\lambda x.\, (a_{fx}, \tau_{fx} \circ
  f_x)$. On the other hand, by unicity in
  Proposition~\ref{prop:23}, the image of
  $f^\circ(\tau) = \tau \circ f \in P^\circ A$ under
  $\theta^\mathsf{P}_A \colon P^\circ A \rightarrow
  T_{\mathbb{B}_\mathsf{P}}A$ is also $\lambda x.\, (a_{fx}, \tau_{fx}
  \circ f_x)$, as desired.
\end{proof}
Combining this result with Proposition~\ref{prop:20}, we obtain:

\begin{Thm}
  \label{thm:4}
    The functor
  $\mathbb{B}_{(\thg)} \colon \cat{Cmd}_a(\cat{Set})
  \rightarrow \cat{Cof}$ taking an accessible comonad $\mathsf{Q}$ to its
  behaviour category $\mathbb{B}_\mathsf{Q}$, and a map of accessible comonads
  $f \colon \mathsf{P} \rightarrow \mathsf{Q}$ to the cofunctor
  $\mathbb{B}_f \colon \mathbb{B}_\mathsf{P} \rightsquigarrow
  \mathbb{B}_\mathsf{Q}$ of Proposition~\ref{prop:25}, yields a within-isomorphism~factorisation 
\begin{equation*}
  \cd[@C+0.2em]{
    \twocong[0.66]{dr}{} & \cat{Cof} \ar[d]^-{\mathsf{Q}_{(\thg)}} \\
    \cat{Cmd}_a(\cat{Set}) \ar[r]_-{\mathrm{Costr}}
    \ar@{-->}[ur]^-{\mathbb{B}_{(\thg)}} & \cat{Mnd}_a(\cat{Set})^\mathrm{op}\rlap{ .}
  }
\end{equation*}
\end{Thm}

\section{Idempotency of costructure--cosemantics}
\label{sec:idemp-monad-comon}

So far, we have seen that cosemantics takes values in presheaf
comonads, and costructure takes values in presheaf monads; to complete
our understanding of the costructure--cosemantics adjunction, we now
show that further application of either adjoint simply interchanges a
presheaf monad with its corresponding presheaf comonad. More
precisely, we will show that the costructure--cosemantics adjunction
is \emph{idempotent}, with the presheaf monads and 
comonads as fixpoints to either side.

\subsection{Idempotent adjunctions}
\label{sec:idemp-adjunct}

We begin by recalling standard category-theoretic background on
fixpoints and idempotency for adjunctions. To motivate this, recall
that any adjunction between posets induces an isomorphism between the
sub-posets of fixpoints to each side. Similarly, any adjunction of categories restricts to
an adjoint equivalence between
the full subcategories of \emph{fixpoints} in the following sense:

\begin{Defn}[Fixpoints]
  \label{def:11}
  Let $L \dashv R \colon \D \rightarrow \C$ be an adjunction. A
  \emph{fixpoint to the left} is an object $X \in \D$ at which the
  counit map $\varepsilon_X \colon LRX \rightarrow X$ is invertible; a
  \emph{fixpoint to the right} is $Y \in \C$ for which $\eta_Y \colon
  Y \rightarrow RLY$ is invertible. We write $\cat{Fix}(LR)$ and
  $\cat{Fix(RL)}$ for the full subcategories of fixpoints to the left
  and right.
\end{Defn}

In the posetal case, the fixpoints to the left and the right are
respectively coreflective and reflective in the whole poset. This is
not true in general for adjunctions between categories, but it is true in
the following situation:

\begin{Defn}[Idempotent adjunction]
  \label{def:12}
  An adjunction $L \dashv R \colon \D \rightarrow \C$ is called
  \emph{idempotent} if it satisfies any one of the following
  equivalent conditions:\vskip-1.5\baselineskip\leavevmode
  \begin{multicols}{2}\raggedright
  \begin{enumerate}[(i)]
  \item Each $RX$ is a fixpoint;
  \item $R$ inverts each counit component;
  \item The monad $RL$  is idempotent;
  \item Each $LY$ is a fixpoint;
  \item $L$ inverts each unit component;
  \item The comonad $LR$  is idempotent.
  \end{enumerate}
  \end{multicols}
\end{Defn}

The equivalence of these conditions is straightforward and well-known;
for the reader who has not seen it, we leave the proof as an
instructive exercise. Equally straightforward are the following
consequences of the definition:

\begin{Prop}
  \label{prop:3}
  If the adjunction $L \dashv R \colon \D \rightarrow \C$ is
  idempotent, then:
  \begin{enumerate}[(i)]
  \item $X \in \D$ is a fixpoint if and only if it is in the
    essential image of $L$;
  \item $Y \in \C$ is a fixpoint if and only if it is in the
    essential image of $R$;
  \item The fixpoints to the left are coreflective in $\D$ via $X \mapsto LRX$.
  \item The fixpoints to the right are reflective in $\C$ via $Y \mapsto RLY$;
  \end{enumerate}
\end{Prop}

\subsection{Presheaf monads and presheaf comonads are fixpoints}
\label{sec:from-presheaf-monads}

We aim to show that the costructure--cosemantics adjunction~\eqref{eq:4} is
idempotent, with the presheaf monads and comonads as the
fixpoints. We first show that
costructure and cosemantics interchange a presheaf monad
with the corresponding presheaf comonad.

\begin{Prop}
  \label{prop:15}
  We have isomorphisms of comonads, natural in $\mathbb{B}$, of the
  form
  \begin{equation}
    \label{eq:14}
    \alpha_\mathbb{B} \colon \mathsf{Q}_\mathbb{B} \rightarrow
    \mathrm{Cosem}(\mathsf{T}_\mathbb{B})
  \end{equation}
  characterised by the fact that they induce on categories of
  Eilenberg--Moore coalgebras the functor
  $\mathbb{B}\text-\cat{Set} \rightarrow {{}^{\mathsf{T}_\mathbb{B}}
    \cat{Set}}$ sending the left $\mathbb{B}$-set $(X,p,\ast)$ to the
  $\mathsf{T}_\mathbb{B}$-comodel $\alg$ with
  \begin{equation}
    \label{eq:12}
    \begin{aligned}
      \dbr{\lambda b.\, (f_b, a_b)}_{\alg} \colon X &\rightarrow A
      \times
      X\\
      x & \mapsto (a_{p(x)}, f_{p(x)} \ast x)\rlap{ .}
    \end{aligned}
  \end{equation}
\end{Prop}

In the statement of this result, we identify
$\cat{Coalg}(\mathsf{Q_\mathbb{B}})$ with $\mathbb{B}\text-\cat{Set}$
by Proposition~\ref{prop:19}, and
$\cat{Coalg}(\mathrm{Cosem}(\mathsf{T}_\mathbb{B}))$ with
${}^{\mathsf{T}_\mathbb{B}}\cat{Set}$ by Proposition~\ref{prop:2}.

\begin{proof}
  By Proposition~\ref{prop:21}, the associated monad of the theory of
  $\mathbb{B}$-valued dependently typed update is the presheaf monad
  $\mathsf{T}_\mathbb{B}$; so by Proposition~\ref{prop:37}, we have an
  isomorphism over $\cat{Set}$ of categories of comodels
  ${}^{\mathbb{T}_\mathbb{B}} \cat{Set} \cong
  {}^{\mathsf{T}_\mathbb{B}}\cat{Set}$, sending the
  $\mathbb{T}_\mathbb{B}$-comodel $\alg$ to the
  $\mathsf{T}_\mathbb{B}$-comodel structure on $X$ with
  $\dbr{\lambda b.\, (f_b, a_b)} = \dbr{\mathsf{get}(\lambda b.\,
    \mathsf{upd}_{f_b}(a_b))}$. Composing this isomorphism with the
  invertible~\eqref{eq:10} yields an invertible functor
  $\mathbb{B}\text-\cat{Set} \rightarrow {}^{\mathsf{T}_\mathbb{B}}
  \cat{Set}$ over $\cat{Set}$, which by inspection has the
  formula~\eqref{eq:12}. We conclude by the full fidelity of the
  Eilenberg--Moore semantics functor (Lemma~\ref{lem:3}).
\end{proof}

\begin{Prop}
  \label{prop:14}
  For any small category $\mathbb{B}$, the monad morphism 
  \begin{equation}\label{eq:9}
    \bar
    \alpha_{\mathbb{B}} \colon \mathsf{T}_\mathbb{B} \rightarrow
    \mathrm{Costr}(\mathsf{Q}_\mathbb{B})
  \end{equation}
  found as the adjoint transpose of the isomorphism~\eqref{eq:14}, is
  itself an isomorphism.
\end{Prop}

\begin{proof}
  By Remark~\ref{rk:1} and~\eqref{eq:12}, we see that $\bar \alpha$
  sends the element $(f,a) = \lambda b.\, (f_b, a_b)$ of
  $T_\mathbb{B}(A) = \prod_b(\mathbb{B}_b \times A)$ to the
  transformation
  $\bar \alpha(f,a) \colon U^\mathbb{B} \Rightarrow A \cdot
  U^\mathbb{B} \colon \mathbb{B}\text-\cat{Set} \rightarrow \cat{Set}$
  whose component at a $\mathbb{B}$-set $(X, p, \cdot)$ is given by
  the function
  \begin{equation}
    \label{eq:11}
    \begin{aligned}
      \bar \alpha(f,a)_{(X,p, \cdot)} \colon X & \rightarrow A \times X
      \\ x & \mapsto (a_{p(x)}, f_{p(x)} \cdot x)\rlap{ .}
    \end{aligned}
  \end{equation}

  We must show every
  $\gamma \colon U^{\mathbb{B}} \Rightarrow A \cdot U^{\mathbb{B}}$
  takes this form for a unique $(f,a) \in T_\mathbb{B}(A)$. For each
  $b \in \mathbb{B}$, we have the representable left $\mathbb{B}$-set
  $y(b)$ with underlying set $\mathbb{B}_b$, projection to
  $\mathrm{ob}(\mathbb{B})$ given by codomain, and action given by
  composition in $\mathbb{B}$. The component of $\gamma$ at $y(b)$ is a
  function $\mathbb{B}_b \rightarrow A \times \mathbb{B}_b$, which, if
  we are to have $\gamma = \bar \alpha(f,a)$, must by~\eqref{eq:11}
  have its value at $1_b \in \mathbb{B}_b$ given by $(a_b, f_b)$. Thus,
  if we define $(a_b, f_b)$ to be $\gamma_{y(b)}(1_b)$ for each
  $b \in \mathbb{B}$, then it remains only to verify that indeed
  $\gamma = \bar \alpha(f,a)$. But for any $\mathbb{B}$-set
  $(X,p,\ast)$ and any $x \in X$, there is by the Yoneda lemma a
  unique map of $\mathbb{B}$-sets
  $\tilde{x} \colon y(p(x)) \rightarrow X$ sending $1_{px}$ to $x$; and
  now naturality of $\gamma$ ensures that
  \begin{equation*}
    \gamma_{(X,p,\ast)}(x) = \gamma_{(X,p,\ast)}(\tilde x(1_{px})) =
    (A \times \tilde x)(\gamma_{y(px)}(1_{px})) = (a_b, \tilde
    x(f_{px})) = (a_{px}, f_{px} \ast x)
  \end{equation*}
  so that $\gamma = \bar \alpha(f,a)$ as desired.
\end{proof}

Given the tight relationship between \eqref{eq:14}
and~\eqref{eq:9}, it is now easy to conclude that presheaf monads and
comonads are fixpoints.
\begin{Prop}
  \label{prop:12}
  Each presheaf monad is a fixpoint on the
  left of the costructure--cosemantics adjunction, while each presheaf comonad
  is a fixpoint on the right.
\end{Prop}
\begin{proof}
  For each small category $\mathbb{B}$ we have a commuting triangle
  \begin{equation*}
    \cd[@!C@C-6em]{
      \mathsf{Q}_\mathbb{B} \ar[dr]_-{\alpha}
      \ar[rr]^-{\eta_{\mathsf{Q}_\mathbb{B}}} & &
      \mathrm{Cosem}(\mathrm{Costr}(\mathsf{Q}_\mathbb{B}))
      \ar[dl]^-{\mathrm{Cosem}(\bar \alpha)} \\
       & \mathrm{Cosem}(\mathsf{T}_\mathbb{B})
    }
  \end{equation*}
  where $\eta_{\mathsf{Q}_\mathbb{B}}$ is the unit of~\eqref{eq:4} and
  $\alpha$ and $\bar \alpha$ are as in~\eqref{eq:14} and~\eqref{eq:9}.
  Since both $\alpha$ and $\bar \alpha$ are invertible, it follows
  that $\eta_{\mathsf{Q}_\mathbb{B}}$ is too; and since every presheaf
  comonad is isomorphic to some $\mathsf{Q}_\mathbb{B}$, it follows
  that every presheaf comonad is a fixpoint on the right. The dual
  argument shows each presheaf monad is a fixpoint on the left.
\end{proof}

\subsection{Idempotency of the costructure--cosemantics adjunction}
\label{sec:idemp-monad-comon-1}

As an immediate consequence of the preceding result, we have:

\begin{Thm}
  \label{thm:5}
  The costructure--cosemantics adjunction~\eqref{eq:4} is idempotent.
  Its fixpoints to the left are the presheaf monads, while those to
  the right are the presheaf comonads.
\end{Thm}

\begin{proof}
  Each $\mathrm{Cosem}(\mathsf{T})$ is a presheaf comonad by
  Proposition~\ref{prop:40}, and each presheaf comonad is a fixpoint
  to the right by Proposition~\ref{prop:12}; thus
  Definition~\ref{def:12}(i) is satisfied and the adjunction is
  idempotent. For the remaining claims, one direction is
  Proposition~\ref{prop:12}; while the other follows on noting that,
  by the preceding result and Proposition~\ref{prop:14}, the composite
  $\mathrm{Cosem} \circ \mathrm{Costr}$ sends each comonad to a
  presheaf comonad, while $\mathrm{Costr} \circ \mathrm{Cosem}$ sends
  each monad to a presheaf monad.
\end{proof}

We may use this result to resolve some unfinished business:

\begin{Prop}
  \label{prop:26}
  The presheaf monad functor $\mathsf{T}_{(\thg)} \colon \cat{Cof} \rightarrow
  \cat{Mnd}_a(\cat{Set})^\mathrm{op}$ of Proposition~\ref{prop:22} is
  full and faithful.
\end{Prop}

\begin{proof}
  Since the costructure--cosemantics adjunction is idempotent, the functor
  $\mathrm{Costr} \colon \cat{Cmd}_a(\cat{Set}) \rightarrow
  \cat{Mnd}_a(\cat{Set})^\mathrm{op}$ is fully faithful when
  restricted to the subcategory of presheaf comonads; and since
  $\mathsf{Q}_{(\thg)}$ takes its image in this subcategory, we see
  that
  $\mathord{\mathrm{Costr}} \circ \mathord{\mathsf{Q}_{(\thg)}} \colon
  \cat{Cof} \rightarrow \cat{Mnd}_a(\cat{Set})^\mathrm{op}$ is fully
  faithful. Now transporting the values of this composite functor
  along the isomorphisms
  $\bar \alpha_\mathbb{B} \colon \mathsf{T}_\mathbb{B} \cong
  \mathrm{Costr}(\mathsf{Q}_{\mathbb{B}})$ of
  Proposition~\ref{prop:14} yields a fully faithful functor
  $\cat{Cof} \rightarrow \mathrm{Mnd}_a(\cat{Set})^\mathrm{op}$ which
  acts on objects by $\mathbb{B} \mapsto \mathsf{T}_\mathbb{B}$, and
  on morphisms by
  $F \mapsto (\bar \alpha_{\mathbb{B}})^{-1} \circ
  \mathrm{Costr}(\mathsf{Q}_F) \circ \bar \alpha_{\mathbb{C}}$. Now
  direct calculation shows this action on morphisms to be precisely
  that of~\eqref{eq:18}.
\end{proof}

It follows from this and Proposition~\ref{prop:16} that the full
embeddings $\mathsf{Q}_{(\thg)} \colon \cat{Cof} \rightarrow
\cat{Cmd}_a(\cat{Set})$ and $\mathsf{T}_{(\thg)} \colon \cat{Cof}
\rightarrow \cat{Mnd}_a(\cat{Set})^\mathrm{op}$ exhibit $\cat{Cof}$ as
equivalent to the full subcategories of fixpoints to the left and to
the right; from which it follows that:

\begin{Prop}
  \label{prop:27}
  The presheaf monad and presheaf comonad functors of
  Propositions~\ref{prop:22} and~\ref{prop:16}, together with the
  behaviour functors of Theorems~\ref{thm:1} and~\ref{thm:4},
  participate in adjunctions
  \begin{equation*}
    \cd{
      {\cat{Cof}} \ar@<-4.5pt>[r]_-{\mathsf{T}_{(\thg)}}
      \ar@{<-}@<4.5pt>[r]^-{\mathbb{B}_{(\thg)}} \ar@{}[r]|-{\top} &
      {\cat{Mnd}_a(\cat{Set})^\mathrm{op}} & &
      {\cat{Cof}} \ar@<-4.5pt>[r]_-{\mathsf{Q}_{(\thg)}}
      \ar@{<-}@<4.5pt>[r]^-{\mathbb{B}_{(\thg)}} \ar@{}[r]|-{\bot} &
      {\cat{Cmd}_a(\cat{Set})} 
    } 
  \end{equation*}
  exhibiting the full subcategories of presheaf monads, respectively
  presheaf comonads, as reflective in $\cat{Mnd}_a(\cat{Set})$,
  respectively $\cat{Cmd}_a(\cat{Set})$.
\end{Prop}

We can describe the units of these reflections explicitly. On the one
hand, if $\mathsf{Q}$ is an accessible comonad on $\cat{Set}$, then
its reflection in the full subcategory of presheaf comonads is the
presheaf comonad of the behaviour category $\mathbb{B}_\mathsf{Q}$, and
the reflection map
$\mathsf{Q} \rightarrow \mathsf{Q}_{\mathbb{B}_\mathsf{Q}}$ has
components
\begin{align*}
  \eta_A \colon Q(A) &\rightarrow \textstyle \sum_{x \in Q1}
  A^{[\cat{Set}, \cat{Set}](Q_x, \mathrm{id})}\\
  a & \mapsto (Q!(a), \lambda \tau.\, \tau_A(a))\rlap{ .}
\end{align*}

On the other hand, if $\mathsf{T}$ is an accessible monad on
$\cat{Set}$, then its reflection into the full subcategory of presheaf
monads is the presheaf monad of the behaviour category
$\mathbb{B}_\mathsf{T}$, and the reflection map
$\eta \colon \mathsf{T} \rightarrow
\mathsf{T}_{\mathbb{B}_\mathsf{T}}$ has components
\begin{align*}
  \eta_A \colon T(A) &\rightarrow \textstyle \prod_{\beta \in
    B_\mathsf{T}} (A \times T(1) \quot \sim_\beta)\\
  t & \mapsto \lambda \beta.\, (\beta(t), \tilde t)\rlap{ .}
\end{align*}

In fact, this reflection map exhibits
$\mathsf{T}_{\mathbb{B}_\mathsf{T}}$ as the result of adjoining to
$\mathsf{T}$ a new $B_\mathsf{T}$-ary operation $\mathsf{beh}$
satisfying the axioms of read-only state and the  axioms
\begin{equation*}
  t(u) \equiv_{\mathsf{beh}, \beta} t \gg u_{\beta(t)}
\end{equation*}
for all $t \in T(A)$ and $u \in T(B)^A$. From a computational
perspective, we understand the new operation $\mathsf{beh}$ as an
``oracle'' which allows the user to request complete information about
the future behaviour of the external system with which we are
interacting. Of course, since this future behaviour is typically
wildly non-computable, we immediately leave the realm of
computationally meaningful theories. In future work, we will see how
to rectify this, to some degree, by considering an adjunction between
accessible monads on $\cat{Set}$ and suitably accessible comonads on
the category of \emph{topological spaces}. In this refined setting, we
will see that the passage from monad to comonad and back adjoins new
operations which observe only \emph{finite} amounts of information
about the future behaviour of the system.

\section{Examples and applications: cosemantics}
\label{sec:cosem-exampl-appl}
In the final two sections of this paper, we give a range of examples
illustrating our main results. In this section, we calculate the
behaviour category, and the comodels classifying admissible
behaviours, for a range of examples of algebraic theories for
computational effects, and calculate some examples of cofunctors
between behaviour categories induced by computationally interesting
interpretations of algebraic theories.

\subsection{Reversible input}
\label{sec:reversible-input}
Given a set $V$, the theory of \emph{$V$-valued reversible input}
(first considered for $\abs V = 2$
in~\cite{Jonsson1961On-two-properties}) is generated by a $V$-ary
operation $\mathsf{read}$, and a $V$-indexed family of unary
operations $\mathsf{unread}_v$, satisfying the equations
\begin{equation}\label{eq:36}
  \mathsf{unread}_v(\mathsf{read}(x)) \equiv x_v \qquad
  \text{and} \qquad \mathsf{read}(\lambda u.\, \mathsf{unread}_u(x))
  \equiv x\rlap{ .}
\end{equation}
If $\mathsf{read}$ is thought of as reading the next value from an
input stream, then $\mathsf{unread}_v$ returns the value $v$ to the
front of that stream. A comodel of this theory comprises the data of a
set $S$, a function $\dbr{\mathsf{read}} = (g,n) \colon S \rightarrow V \times S$
and functions
$\dbr{\mathsf{unread}_v} \colon S \rightarrow S$, or equally a
single function $p \colon V \times S \rightarrow S$; while the
equations force $(g,n)$ and $p$ to be inverse to each other. Thus comodels of $V$-valued
reversible input are equally well
comodels of $V$-valued input whose structure map $\dbr{\mathsf{read}}
\colon S \rightarrow V
\times S$ is invertible. Since, in particular, this is true for the
final comodel $\alg[V^\mathbb{N}]$ of $V$-valued input by the well-known Lambek lemma, we
conclude that this is also the final
comodel of $V$-valued reversible input.

We now calculate the comodel associated to an admissible
behaviour $W \in V^\mathbb{N}$. We begin with some calculations
relating to $\sim_W$-equivalence. First, by Remark~\ref{rk:4}, any unary term is
$\sim_W$-equivalent to one of the form
$\sigma_1 \gg \dots \gg \sigma_n \gg\mathrm{id}$ where each $\sigma_i$
is either $\mathsf{read}$ or some $\mathsf{unread}_v$. Now the first
equation in~\eqref{eq:36} implies that
$\mathsf{unread}_v \gg \mathsf{read} \gg (\thg)$ is the identity
operator, and so any unary term is $\sim_W$-equivalent to one of the
form
\begin{equation*}
  [n, v_m, \dots, v_1] \defeq \smash{\overbrace{\mathsf{read} \gg \dots \gg \mathsf{read}}^{n}} \gg
  \mathsf{unread}_{v_m} \gg \dots \gg \mathsf{unread}_{v_1} \gg \mathrm{id}
\end{equation*}
for some $n \in \mathbb{N}$ and $v_m, \dots, v_1 \in V$. Since the
behaviour $W$ satisfies $W(\mathsf{read}) = W_0$, we have
$\mathsf{read} \gg \mathsf{unread}_{W_0} = \mathsf{read}(\lambda u.\,
\mathsf{unread}_{W_0}) \sim_W \mathsf{read}(\lambda u.\,
\mathsf{unread}_u) = \mathrm{id}$; whence by Lemma~\ref{lem:13}, also
$[n+1, W_n, v_m, \dots, v_1] \sim_W [n, v_m, \dots, v_1]$ for any
$n \in \mathbb{N}$ and $v_m, \dots, v_1 \in V$. Consequently, each
unary term is $\sim_W$-equivalent to an element of the set $S_W$ given
by
\begin{equation}\label{eq:40}
  \{[n, v_m, \dots, v_1] : n,m \in \mathbb{N}, v_i \in V \text{ and $W_{n-1} \neq v_m$ if $n,m >0$}\}\rlap{ .}
\end{equation}

We claim that $S_W$ is in fact a set of $\sim_W$-equivalence class
representatives. For this, it suffices by
Remark~\ref{rk:3} to endow $S_W$
with a comodel structure $\alg[S]_W$ satisfying~\eqref{eq:37}---which
will then make it the comodel classifying states of behaviour $W$. We do so
by taking $\dbr{\mathsf{read}}_{\alg[S]_W}$ to be given by
\begin{equation}\label{eq:44}
  [n, v_m, \dots, v_1] \mapsto
  \begin{cases}
    (W_n, [n+1]) & \text{ if $m=0$;}\\
    (v_1, [n, v_m, \dots, v_{2}]) & \text{ if $m>0$,}\\
  \end{cases}
\end{equation}
and taking $\dbr{\mathsf{unread}_v}_{\alg[S]_W}$ to be given by
\begin{equation}\label{eq:43}
  [n, v_1, \dots, v_m] \mapsto
  \begin{cases}
    [n-1] & \text{ if $m=0$, $n>0$, $v = W_{n-1}$;}\\
    [n,v_m, \dots, v_1, v] & \text{ otherwise.}
  \end{cases}
\end{equation}

We may now use the above calculations to identify maps $W \rightarrow W'$
in the behaviour category. These correspond to comodel homomorphisms $\alg[S]_{W'}
\rightarrow \alg[S]_W$ and so to states of $\alg[S]_W$ of
behaviour $W'$. Since the state $[n, v_m, \dots, v_1] \in \alg[S]_W$
has behaviour given by the stream of values $v_1 \dots v_m
W_{n} W_{n+1} W_{n+2} \dots$, we conclude that morphisms $W
\rightarrow W'$ in the behaviour category are states of the form $[n,
W'_{m-1}, \dots, W'_0]$ where $W'_k = W_{k+n-m}$ for all $k \geqslant
m$ but $W'_{m-1} \neq W_{n-1}$. Such a state is clearly uniquely
determined by the integer $i = n-m$, and so we arrive at:
\begin{Prop}
  \label{prop:5}
  The behaviour category of the theory of $V$-valued reversible input
  has object-set $V^\mathbb{N}$; morphisms $W \rightarrow W'$ are
  integers $i$ such that, for some $N \in \mathbb{N}$, we have
  $W'_k = W_{k+i}$ for all $k > N$; and composition is addition of
  integers. The comodel classifying states of behaviour
  $W \in V^\mathbb{N}$ has underlying set~\eqref{eq:40}, and
  co-operations as in~\eqref{eq:44} and~\eqref{eq:43}.
\end{Prop}

Note that this behaviour category is a groupoid; in fact, it is not
hard to show that it is the
\emph{free} groupoid on the behaviour category for $V$-valued input.
This groupoid is well-known in the study of Cuntz $C^\ast$-algebras:
for example, for finite $V$ it appear already
in~\cite[Definition~III.2.1]{Renault1980A-groupoid}. In this context,
it is important that the groupoid is not just as a groupoid of sets,
but a \emph{topological} groupoid; in a sequel to this paper, we will
explain how this topology arises very naturally via comodels.

\subsection{Stack}
\label{sec:stack}
Given a set $V$, the theory of a \emph{$V$-valued stack}---introduced
for a finite $V$ in~\cite{Goncharov2013Trace}---is generated
by a $V+\{\bot\}$-ary operation $\mathsf{pop}$, whose arguments we
group into an $V$-ary part and a unary part; and a $V$-indexed family
of unary operations $\mathsf{push}_v$ for $v \in V$, satisfying the
equations
\begin{equation*}
  \mathsf{push}_v(\mathsf{pop}(x,y)) \equiv x_v \quad
  \mathsf{pop}(\lambda v.\, \mathsf{push}_v(x), x) \equiv x \quad
  \mathsf{pop}(x, \mathsf{pop}(y,z)) \equiv
  \mathsf{pop}(x,z)\rlap{ .}
\end{equation*}
This theory captures the semantics of a stack of elements from $V$: we
read $\mathsf{push}_v(x)$ as ``push $v$ on the stack and continue as
$x$'', and $\mathsf{pop}(x,y)$ as ``if the stack is non-empty, pop its
top element $v$ and continue as $x_v$; else continue as $y$''.

Note the similarities with the theory of $V$-valued reversible input;
indeed, this latter theory could equally well be seen as the theory of
a $V$-valued \emph{infinite} stack. We can formalise this via an
interpretation of the theory
of $V$-valued stack into $V$-valued reversible input which maps
$\mathsf{push}_u$ to $\mathsf{unread}_u$ and $\mathsf{pop}(x,y)$ to
$\mathsf{read}(x)$.

A comodel of the theory of a $V$-valued stack comprises a set $S$ with
functions 
$(g,n) \colon S \rightarrow (V + \{\bot\}) \times S$ (modelling
$\mathsf{pop}$) and $p \colon V \times S \rightarrow S$ (modelling
the $\mathsf{push}_v$'s) subject to 
conditions corresponding to the three equations above:
\begin{enumerate}
\item $g(p(v,s)) = v$ and $n(p(v,s)) = s$;
\item If $g(s) = \{\bot\}$ then $n(s) = s$, while if $g(s) = v$ then
  $p(v,n(s)) = s$;
\item If $g(s) = \{\bot\}$ then $g(n(s)) = \{\bot\}$ and
  $n(n(s)) = n(s)$ (this is implied by (2)).
\end{enumerate}

Writing $E = \{s \in S : g(s) = \bot\}$ for the set of ``states in
which the stack is empty'', and $j \colon E \rightarrow S$
for the inclusion, (1) implies that
$p \colon V \times S \rightarrow S$ is an injection whose image is
disjoint from that of $j$, and (2) that every $s \in S$ lies either in
$E$ or in the image of $p$. So we have a coproduct diagram
\begin{equation*}
  \cd[@-1em@!C@C-1em]{
    V \times S \ar[dr]_-{p} & & E\rlap{ .} \ar[dl]^-{j} \\ & S
  }
\end{equation*}
In fact, any such coproduct diagram comes from a comodel: we may
recover $(g,n) \colon S \rightarrow (V + \{\bot\}) \times S$ as the
unique map whose composites with $p$ and $j$ are
$\lambda (v,s).\, (v, s)$ and $\lambda e.\, (\bot, j(e))$
respectively. Thus a comodel structure on a set $S$ is
equivalently given by a set $E$ and a coproduct diagram
$V \times S \rightarrow S \leftarrow E$.

The \emph{final} comodel of this theory is the set $V^{\leqslant
  \omega}$ of partial functions $\mathbb{N} \rightharpoonup V$ which
are defined on
some initial segment of $\mathbb{N}$,
under the comodel structure corresponding to the coproduct diagram
\begin{equation*}
  \cd[@-1em@!C@C-1em]{
    V \times V^{\leqslant \omega} \ar[dr]_-{(v,W) \mapsto v.W\,\,\,\,} & &
    \{\ast\}\rlap{ .} \ar[dl]^-{\mathord\ast\, \mapsto \varepsilon} \\ & S
  }
\end{equation*}
Here, we write $\varepsilon$ for the everywhere-undefined element
of $V^{\leqslant \omega}$, and write $v.W \in V^{\leqslant \omega}$ for the element
with $(v.W)_0 = v$ and $(v.W)_{i+1} \simeq W_i$\footnote{
We use \emph{Kleene equality} $a \simeq b$, meaning that $a$ is defined
just when $b$ is defined, and they are then equal.}.
In terms of the generating co-operations, this final comodel is given
by:
\begin{equation*}
  \dbr{\mathsf{push}_v}(W) = v.W \qquad \dbr{\mathsf{pop}}(v.W) = (v,W)
  \qquad \text{and} \qquad \dbr{\mathsf{pop}}(\varepsilon) = (\bot, \varepsilon)\rlap{ .}
\end{equation*}

We now calculate the comodel associated to an admissible
behaviour $W \in V^{\leqslant \omega}$. Given the similarity with the
theory of $V$-valued reversible input, we may argue as in the
previous section to see that any unary term is
$\sim_W$-equivalent to one
\begin{equation*}
  [n, v_m, \dots, v_1] \defeq \underbrace{\mathsf{pop} \gg \dots \gg \mathsf{pop}}_{n} \gg
  \mathsf{push}_{v_m} \gg \dots \gg \mathsf{push}_{v_1} \gg \mathrm{id}
\end{equation*}
for some $n \in \mathbb{N}$ and $v_m, \dots, v_1 \in V$. Now, if
$W_0$ is undefined, then $W(\mathsf{pop}) = \bot$, and so
$\mathsf{pop} \gg m = \mathsf{pop}(\lambda v.\, m, m) \sim_W
\mathsf{pop}(\lambda v.\, \mathsf{push}_v(m),m) = m$ for any $m \in
T(1)$. By Lemma~\ref{lem:13}, it follows that $[n+1, v_m, \dots, v_1]
\sim_W [n, v_m, \dots, v_1]$ whenever $W_n$ is undefined, and so we conclude
that each unary term is $\sim_W$-equivalent to some $[n, v_m, \dots,
v_1]$ for which $W$ is defined at all $k < n$.
At this point, by repeating the arguments of the preceding section,
\emph{mutatis mutandis}, we may show that any unary term is
$\sim_W$-equivalent to an element of the set
\begin{multline}\label{eq:45}
  \{[n, v_m, \dots, v_1] : n, m \in \mathbb{N}, v_i \in V, 
  W \text{ defined at
    all $k < n$,} \\
   \text{ and $W_{n-1} \neq v_m$ if $n,m >0$}\}\rlap{ .}
\end{multline}

We now show, like before, that this is a set of $\sim_W$-equivalence
class representatives, by making it into a comodel
satisfying~\eqref{eq:37}; again, this comodel will then classifying
states of behaviour $W$. This time, we take $\dbr{\mathsf{pop}}$ to be
given by
\begin{equation}\label{eq:39}
  [n, v_m, \dots, v_1] \mapsto
  \begin{cases}
    (W_n, [n+1]) & \text{ if $m=0$ and $W_n$ defined;}\\
    (\bot, [n]) & \text{ if $m=0$ and $W_n$ undefined;}\\
    (v_1, [n, v_m, \dots, v_{2}]) & \text{ if $m>0$,}\\
  \end{cases}
\end{equation}
and take $\dbr{\mathsf{push}_v}$ to be given exactly as in~\eqref{eq:43}.
Transcribing the calculations of the preceding section, we arrive at:
\begin{Prop}
  \label{prop:8}
  The behaviour category of the theory of a $V$-valued stack
  has object-set $V^{\leqslant \omega}$; morphisms $W \rightarrow W'$ are
  integers $i$ such that, for some $N \in \mathbb{N}$, we have
  $W'_k \simeq W_{k+i}$ for all $k > N$; and composition is addition of
  integers. The comodel classifying states of behaviour
  $W \in V^{\leqslant \omega}$ has underlying set~\eqref{eq:45}, and
  co-operations as in~\eqref{eq:39} and~\eqref{eq:43}.
\end{Prop}

In fact, it is easy to see that the behaviour category of a $V$-valued
stack is the disjoint union of the behaviour category for
$V^\ast$-valued state (modelling a finite stack) and for $V$-valued
reversible input (modelling an infinite stack). The cofunctor on
behaviour categories induced by the interpretation of the theory of a
$V$-valued stack into that of $V$-valued reversible input is simply
the connected component inclusion.

\subsection{Dyck words}
\label{sec:dyck-words}

A \emph{Dyck word} is a finite
list $W \in \{\mathsf{U},\mathsf{D}\}^\ast$ with the same number of
\textsf{U}'s as \textsf{D}'s, and with the property that the $i$th \textsf{U} in the list always
precedes the $i$th \textsf{D}. Here, \textsf{U} and \textsf{D} stand for ``up'' and
``down'', and the idea is that a Dyck word records a walk on the
natural numbers $\mathbb{N}$ with steps $\pm 1$ which starts and ends
at $0$. More generally, we can encode walks from
$n \in \mathbb{N}$ to $m \in \mathbb{N}$ by ``affine Dyck
words'':
\begin{Defn}[Affine Dyck words]
  \label{def:16}
  Given $n,m \in \mathbb{N}$, an \emph{affine Dyck word} from $n$ to
  $m$ is a word $W \in \{\mathsf{U}, \mathsf{D}\}^\ast$ such that
  $\#\{\text{\textsf{D}'s in $W$}\} - \#\{\text{\textsf{U}'s in $W$}\}
  = n-m$, and such that the $i$th \textsf{U} precedes the $(i+n)$th
  \textsf{D} for all suitable $i$. We may extend this notation by
  declaring \emph{any} word $W \in \{\mathsf{U}, \mathsf{D}\}^\ast$ to
  be an affine Dyck word from $\infty$ to $\infty$. If
  $n,m \in \mathbb{N} \cup \{\infty\}$, then we write
  $W \colon n \rightsquigarrow m$ to indicate that $W$ is an affine
  Dyck word from $n$ to $m$.
\end{Defn}

For example:
\begin{itemize}
\item UUDUDD is a Dyck word, but also an affine Dyck word $n
  \rightsquigarrow n$ for any $n$;
\item UDDUUU is an affine Dyck word $1 \rightsquigarrow 3$ and $2
  \rightsquigarrow 4$, but
  not $0 \rightsquigarrow 2$.
\end{itemize}

We now describe an algebraic theory which encodes the dynamics of the
walks encoded by affine Dyck words. It has two unary operations
$\mathsf{U}$ and $\mathsf{D}$; and an $\mathbb{N}$-indexed family of
binary operations $\mathsf{ht}_{> n}$ each satisfying the
axioms of read-only state; all subject to the following
axioms:
\begin{gather*}
  \mathsf{ht}_{> n}(x, \mathsf{ht}_{>
    m}(y,z)) \equiv \mathsf{ht}_{> m}(\mathsf{ht}_{> n}(x,y),z) \\
  \mathsf{ht}_{> 0}(x, \mathsf{D}(x)) \equiv x \qquad
  \mathsf{U}(\mathsf{ht}_{> 0}(x,y)) \equiv \mathsf{U}(x)\\
  \mathsf{U}(\mathsf{ht}_{>
    n+1}(x,y)) \equiv
  \mathsf{ht}_{> n}(\mathsf{U}(x), \mathsf{U}(y)) \qquad 
  \mathsf{D}(\mathsf{ht}_{>
    n}(x,y)) \equiv
  \mathsf{ht}_{> n+1}(\mathsf{D}(x), \mathsf{D}(y))
\end{gather*}
for all $m\leqslant n \in \mathbb{N}$. The theory of affine Dyck words
provides an interface for accessing a state machine with an internal
``height'' variable $h \in \mathbb{N}$, which responds to two commands
\textsf{U} and \textsf{D} which respectively increase and decrease $h$
by one, with the proviso that \textsf{D} should do nothing when
applied in a state with $h = 0$. With this understanding, we read the
primitive $\mathsf{ht}_{> n}(x,y)$ as ``if $h > n$
then continue as $x$, else continue as $y$''; read $\mathsf{U}(x)$ as
``perform \textsf{U} and continue as $x$''; and read $\mathsf{D}(x)$
as ``perform \textsf{D} (so long as $h>0$) and continue as $x$''.

Rather than compute the comodels by hand, we pass directly to a
calculation of the behaviour category. We begin by finding the
admissible behaviours. By Remark~\ref{rk:2} and the fact that
$\mathsf{ht}_{> n} \gg (\thg)$ is the identity operator, an admissible
behaviour $\beta$ is uniquely determined by the values
$\beta(\sigma_1 \gg \dots \gg \sigma_k \gg \mathsf{ht}_{> n})$ where
each $\sigma_i$ is either $\mathsf{U}$ or $\mathsf{D}$. Now, the last
three axioms imply that these values are determined in turn by the
values $\beta(\mathsf{ht}_{> n}) \in \{\mathsf{tt}, \mathsf{ff}\}$ for
each $n$. Finally, by the first axiom,
$\beta(\mathsf{ht}_{> m}) = \mathsf{tt}$ implies
$\beta(\mathsf{ht}_{> n}) = \mathsf{tt}$ whenever $m \leqslant n$. So
the possibilities are either that there is a \emph{least} $n$ with
$\beta(\mathsf{ht}_{> n}) = \mathsf{ff}$, or that
$\beta(\mathsf{ht}_{> n}) = \mathsf{tt}$ for all $n \in \mathbb{N}$.

In fact, each of these possibilities for $\beta$ does yield an
admissible behaviour. Indeed, identifying these possibilities with
elements of the set $\mathbb{N} \cup \{\infty\}$, we can try
to make this set into a comodel via the formulae of
Proposition~\ref{prop:33}, by taking:
\begin{equation*}
  \dbr{\mathsf{ht}_{> n}}(k) =
  \begin{cases}
    (\mathsf{tt}, k) & \text{if $k > n$;}\\
    (\mathsf{ff}, k) & \text{otherwise,}
  \end{cases} \quad \dbr{\mathsf{U}}(k) = k+1\text{ ,} \quad
  \dbr{\mathsf{D}}(k) =
  \begin{cases}
    0 & \text{if $k = 0$;}\\
    n-1 & \text{otherwise,}
  \end{cases}
\end{equation*}
where we take $\infty > n$ for any $n \in \mathbb{N}$,
and $\infty + 1 = \infty = \infty -1$. It is not hard to check that
these co-operations do in fact yield a comodel, which is then of
necessity the final comodel of the theory of affine Dyck words.

We now compute the comodel $\alg[k]$ associated to a behaviour
$k \in \mathbb{N} \cup \{\infty\}$. Because each operator
$\mathsf{ht}_{> n} \gg (\thg)$ is the identity, each unary operation
is $\sim_k$-equivalent to one in the submonoid generated by
$\mathsf{U}, \mathsf{D} \in T(1)$. Further, by the second axiom we
have $\mathsf{D} \sim_0 \mathrm{id}$, and so by Lemma~\ref{lem:13},
also $W \mathsf{D} W' \sim_k WW'$ for any $k \in \mathbb{N}$, any
$W' \in \{\mathsf{U}, \mathsf{D}\}^\ast$ and any affine Dyck word
$W \colon k \rightsquigarrow 0$ from $k$ to $0$. Applying this rewrite
rule repeatedly, we find that any unary term is $\sim_k$ equivalent to
an element of the set
\begin{equation}\label{eq:48}
  \{\,W \in \{\mathsf{U}, \mathsf{D}\}^\ast : W \colon k
  \rightsquigarrow \ell \text{ for some } \ell \in \mathbb{N} \cup \{\infty\}\,\}\rlap{ ,}
\end{equation}
and we may apply Remark~\ref{rk:3} to see that there are in fact no
further relations. Indeed, we may make~\eqref{eq:48} into a
classifying comodel $\alg[k]$ satisfying~\eqref{eq:37} by taking
\begin{equation}
\begin{gathered}
  \dbr{\mathsf{ht}_{> n}}(W) =
  \begin{cases}
    (\mathsf{tt}, W) & \text{if $W \colon k \rightsquigarrow \ell$
      with $\ell > n$;}\\
    (\mathsf{ff}, W) & \text{otherwise,}
  \end{cases} \\ \dbr{\mathsf{U}}(W) = W \mathsf{U}\text{ ,} \qquad
  \dbr{\mathsf{D}}(\mathsf{W}) =
  \begin{cases}
    W & \text{if $W \colon k \rightsquigarrow 0$;}\\
    W\mathsf{D} & \text{otherwise.}
  \end{cases}
\end{gathered}\label{eq:38}
\end{equation}
From the preceding calculations, we can now read off:
\begin{Prop}
  \label{prop:18}
  The behaviour category of the theory of affine Dyck words has
  object-set $\mathbb{N} \cup \{\infty\}$, and morphisms from $n$ to
  $m$ given by affine Dyck words $W \colon n \rightsquigarrow m$.
  Composition is given by concatenation of words. The comodel
  classifying the behaviour $k \in \mathbb{N} \cup \{\infty\}$ has
  underlying set~\eqref{eq:48}, and co-operations as in~\eqref{eq:38}.
\end{Prop}

The set of Dyck words of length $2n$ is well-known to have the
cardinality of the $n$th Catalan number
$C_n = \tfrac 1 {n+1}{2n \choose n}$. On the other hand, $C_n$ also
enumerates the set of well-bracketed expressions, such as
$((aa)a)(aa)$, composed of
$n+1$ $a$'s. In fact, there is a bijection between Dyck words of
length $2n$ and well-bracketed expressions of $(n+1)$ $a$'s which can
be obtained by interpreting a Dyck word $W$ as a set of instructions
for a stack machine, as follows:
\begin{enumerate}
\item Begin with a stack containing the single element $a$;
\item Read the next element of the Dyck word $W$:
  \begin{itemize}
  \item If it is $\mathsf{U}$, push an $a$ onto the stack;
  \item If it is $\mathsf{D}$, pop the top two elements $x,y$ of the stack
    and push $(xy)$ onto the stack.
  \end{itemize}
\item When $W$ is consumed, return the single element remaining on the stack.
\end{enumerate}

While the terms on which this stack machine operates are the
well-bracketed expressions of $a$'s, we can do something similar for
stacks of elements of
any set $V$ endowed with a constant $a$ and a binary operation $\ast$,
obtaining for each Dyck word $W$ an element of $V$ built from $a$'s
and $\ast$'s. We can understand this in terms of an
interpretation
$f \colon \mathbb{T}_{\mathsf{Dyck}} \rightarrow
\mathbb{T}_{\mathrm{Stack}}$ of the theory of affine Dyck words into
the theory of a $V$-valued stack, given as follows:
\begin{gather*}
  (\mathsf{ht}_{> 0})^f(x,y) = \mathsf{pop}(\lambda v.\, 
  \mathsf{push}_v(x),y) \\
  (\mathsf{ht}_{>n+1})^f(x,y) = \mathsf{pop}(\lambda v.\,
  (\mathsf{ht}_{>n})^f(\mathsf{push}_v(x), \mathsf{push}_v(y)),y)\\
  \mathsf{U}^f(x) = \mathsf{push}_a(x) \qquad 
  \mathsf{D}^f(x,y) = \mathsf{pop}(\lambda v.\,
  \mathsf{pop}(\lambda w.\, \mathsf{push}_{v \ast w}(x),
  x), y)
\end{gather*}

Here, for the (recursively defined) interpretation of the predicates
$\mathsf{ht}_{> n}$, we attempt to pop $n+1$ elements from the top of
our stack of $V$'s; if this succeeds, then we undo our pushes and
return $\mathsf{tt}$, while if it at any point fails, then we undo our
pushes and return $\mathsf{ff}$. For the interpretation of
$\mathsf{U}$ we simply push our constant $a \in V$ onto the stack;
while for the interpretation of $\mathsf{D}(x,y)$, we attempt to pop
the top two elements $v,w$ from the stack and push $v \ast w$ back on.
If this succeeds, we continue as $x$, but some care is needed if it
fails. By the fifth affine Dyck word equation, if our stack contains
exactly one element, then $\mathsf{D}^f$ should yield a stack with no
elements and continue as $x$; while by the second equation, if our
stack is empty, then $\mathsf{D}^f$ should do nothing and continue as
$y$. The forces the definition given above.

The \emph{dynamics} of the interpretation of (affine) Dyck words as
stack operations is captured by the induced cofunctor
$\mathbb{B}_f \colon \mathbb{B}_{\mathrm{Stack}} \rightarrow
\mathbb{B}_{\mathrm{Dyck}}$. It is an easy calculation to see that
this is given as follows:
\begin{itemize}
\item On objects, we map $S \in V^{\leqslant \omega}$ to the cardinality
  $\abs S \in \mathbb{N} \cup \{\infty\}$ of the initial segment of
  $\mathbb{N}$ on
  which $S$ is defined.
\item On morphisms, given $S \in V^{\leqslant \omega}$ and an affine
  Dyck word $W \colon \abs S \rightsquigarrow k$, we return the morphism $S
  \rightarrow S'$ which updates the stack $S$ via the sequence of
  $\mathsf{U}$'s and $\mathsf{D}$'s which specifies $W$.
\end{itemize}
In particular, we may consider the case where $(V, \ast, a)$ is the
set of well-bracketed expressions of $a$'s under concatenation.
Now give a Dyck word $W$, we may regard it as an affine Dyck word $W
\colon 1 \rightsquigarrow 1$; and now updating the singleton stack $a$
via $W \colon 1 = \abs{a} \rightsquigarrow 1$ yields precisely the
well-bracketed expression of $a$'s which corresponds to the given Dyck
word $W$.

\subsection{Store}
\label{sec:store}

Given a set $L$ of \emph{locations} and a family
$\vec V = (V_\ell : \ell \in L)$ of \emph{value} sets, the theory of
\emph{$\vec V$-valued store} comprises a copy of the theory of
$V_\ell$-valued state for each $\ell \in L$---with operations
$(\mathsf{put}_v^{(\ell)} : v \in V_\ell)$ and
$\mathsf{get}^{(\ell)}$, say---together with, for all
$\ell \neq k \in L$, all $v \in V_\ell$ and all $w \in V_k$, the
commutativity axiom:
\begin{equation}\label{eq:17}
  \mathsf{put}^{(\ell)}_v(\mathsf{put}^{(k)}_w(x)) \equiv
  \mathsf{put}^{(k)}_w(\mathsf{put}^{(\ell)}_v(x))\rlap{ .}
\end{equation}
By the argument of~\cite[Lemma~3.21 and~3.22]{Goncharov2020Toward},
these equations also imply the commutativity conditions
$\mathsf{put}^{(\ell)}_v(\mathsf{get}^{(k)}(\lambda w.\, x_w)) \equiv
\mathsf{get}^{(k)}(\lambda w.\, \mathsf{put}^{(\ell)}_v(x_w))$ and
$\mathsf{get}^{(\ell)}(\lambda v.\, \mathsf{get}^{(k)}(\lambda w.\,
x_{vw})) \equiv \mathsf{get}^{(k)}(\lambda w.\,
\mathsf{get}^{(\ell)}(\lambda v.\, x_{vw}))$.

A set-based comodel of this theory is a set $S$ endowed with an
$L$-indexed family of lens structures
$(g^{(\ell)} \colon S \rightarrow V_\ell,\, p^{(\ell)} \colon V_\ell
\times S \rightarrow S)$ which \emph{commute} in the sense that
\begin{equation*}
  p^{(k)}(v, p^{(\ell)}(u, s)) = p^{(\ell)}(v, p^{(k)}(u, s))
\end{equation*}
for all $\ell,k \in L$, $u \in V_\ell$ and $v \in V_k$. When $L$ is
finite and each $V_\ell$ is the same set $V$, this is the notion of
\emph{array} given in~\cite[\sec4]{Power2004From}. The final comodel
of this theory is the set $\prod_{\ell \in L}V_\ell$, under the
operations
\begin{equation*}
  \dbr{\mathsf{get}^{(\ell)}}(\vec v) = (v_\ell, \vec v) \qquad 
  \dbr{\mathsf{put}_v^{(\ell)}}(\vec v) = \vec v[v / v_\ell]\rlap{ .}
\end{equation*}

By similar arguments to those of the preceding
sections, we may now show that:

\begin{Prop}
  \label{prop:32}
  The behaviour category $\mathbb{B}_{\vec V}$ of $\vec V$-valued
  store has set of objects $\prod_{\ell \in L}V_\ell$, while a
  morphism $\vec v \rightarrow \vec w$ is unique when it exists, and
  exists just when $\vec v$ and $\vec w$ differ in only finitely many
  positions. The comodel classifying the behaviour $\vec v \in
  \prod_{\ell \in L}V_\ell$ is the sub-comodel of the final comodel on
  the set
  \begin{equation*}
    \{\vec w \in \textstyle\prod_{\ell \in L}V_\ell : \vec v \text{ and } \vec
    w \text{ differ in only finitely many positions}\}\rlap{ .}
  \end{equation*}
\end{Prop}

For each $\ell \in L$, there is an obvious interpretation of the
theory of $V_\ell$-valued state into the theory of $\vec V$-valued
store, and this induces a cofunctor on behaviour categories
$\mathbb{B}_{\vec V} \rightarrow \nabla V_\ell$ which:
\begin{itemize}
\item On objects, maps $\vec v \in \mathbb{B}_{\vec V}$ to its
  component $v_\ell \in \nabla V_\ell$;
\item On morphisms, for each $\vec v \in \mathbb{B}_{\vec V}$, sends
  $v_\ell \rightarrow v'_\ell$ in $\nabla V_\ell$ to the unique map
  $\vec v \rightarrow \vec v[v'_\ell / v_\ell]$ in
  $\mathbb{B}_{\vec V}$.
\end{itemize}
This captures exactly the ``view update''
paradigm in database theory: on the one hand, projecting from the
state $\vec v$ to its component $v_\ell$ provides a \emph{view} on the
data encoded by $\vec v$; while lifting the morphism
$v_\ell \rightarrow v'_\ell$ to one
$\vec v \rightarrow \vec v[v'_\ell / v_\ell]$ encodes \emph{updating}
the state in light of the given update of the view. The pleasant
feature here is that all of this is completely automatic once we
specify the way in which $V_\ell$-valued state is to be interpreted
into $\vec V$-valued store.

\subsection{Tape}
\label{sec:tape-1}

Our final example is a variant on a particular case of the previous
one; it was introduced \emph{qua} monad
in~\cite{Goncharov2014Towards}, with the presentation given here due
to~\cite{Goncharov2020Toward} and,
independently,~\cite{Pattinson2016Program}. Given a set $V$, we
consider the constant $\mathbb{Z}$-indexed family of sets
$V^{(\mathbb{Z})} = (V : \ell \in \mathbb{Z})$. The theory of a
\emph{$V$-valued tape} is obtained by augmenting the theory of
$V^{(\mathbb{Z})}$-valued store with two new unary operations
$\mathsf{right}$ and $\mathsf{right}^{-1}$ satisfying the axioms
$\mathsf{right}^{-1}(\mathsf{right}(x))$,
$\mathsf{right}(\mathsf{right}^{-1}(x)) \equiv x$, and
$\mathsf{right}(\mathsf{put}_u^{(\ell)}(x)) \equiv
\mathsf{put}_u^{(\ell+1)}(\mathsf{right}(x))$ for all
$\ell \in \mathbb{Z}$. By arguing much as before, we see that this
last axiom implies also that
$\mathsf{right}(\mathsf{get}^{(\ell)}(x)) \equiv
\mathsf{get}^{(\ell+1)}(\lambda v.\, \mathsf{right}(x_v))$.

This theory encapsulates interaction with an doubly-infinite tape,
each of whose locations $\ell \in \mathbb{Z}$ contains a value in $V$
which can be read or updated via $\mathsf{get}$ and $\mathsf{put}$,
and whose head position may be moved right or left via
$\mathsf{right}$ and $\mathsf{right}^{-1}$. A comodel structure on a
set $S$ comprises a $\mathbb{Z}$-indexed family of commuting lens
structures
$(g^{(\ell)} \colon S \rightarrow V, p^{(\ell)} \colon V \times S
\rightarrow S)$ together with a bijection $r \colon S \rightarrow S$
such that $r(p^{(\ell+1)}(u,\vec v)) = p^{(\ell)}(u,r(\vec v))$ for
each $\ell \in \mathbb{Z}$. It is easy to see that the final comodel
of this theory is the final comodel $\alg[V^\mathbb{Z}]$ of
$\vec V$-valued store, augmented with the co-operations
$\dbr{\mathsf{right}}(\vec v)_k = v_{k+1}$ and
$\dbr{\mathsf{right}^{-1}}(\vec v)_{k} = v_{k-1}$. By similar
calculations to before, we now find that:
  
\begin{Prop}
  \label{prop:34}
  The behaviour category of the theory of $V$-valued tape has
  object-set $V^\mathbb{Z}$, while a map $\vec v \rightarrow \vec w$
  is an integer $i$ such that $\vec v_{(\thg) + i}$ and $\vec w$
  differ in only finitely many places. The comodel classifying the
  behaviour $\vec v \in V^\mathbb{Z}$ has underlying set
  \begin{equation*}
    \{(i \in \mathbb{Z}, \vec w \in V^\mathbb{Z}) : \vec v \text{ and } \vec
    w \text{ differ in only finitely many positions}\}
  \end{equation*}
  with operations
  $\dbr{\mathsf{right}^{\pm 1}}(i, \vec w) = (i \pm 1, \vec w)$,
  $ \dbr{\mathsf{get}^{(\ell)}}(i, \vec w) = (w_{i+\ell}, (i, \vec
  w))$, and
  $ \dbr{\mathsf{put}^{(\ell)}_v}(i, \vec w) = (i, \vec
  w[v/w_{i+\ell}])$.
\end{Prop}

If the set $V$ comes endowed with a bijective pairing function
$\spn{\thg, \thg} \colon V \times V \rightarrow V$, say with inverse
$(p, q) \colon V \rightarrow V \times V$, then there is an
interpretation $f$ of the theory of $V$-valued reversible input into the
theory of $V$-valued tape, given by
\begin{align*}
  \mathsf{read}^f(\lambda v.\, x_v) \defeq \mathsf{get}^{(0)}(\lambda
  w.\, \mathsf{put}^{(0)}_{p(w)}(\mathsf{right}(x_{q(w)}))) \\
  (\mathsf{unread}_u)^f(x) \defeq \mathsf{left}(\mathsf{get}^{(0)}(\lambda
  w.\, \mathsf{put}^{(0)}_{\spn{w,u}}(x)))\rlap{ .}
\end{align*}
This induces a cofunctor on behaviour categories which acts as
follows. 
\begin{itemize}
\item On objects, for each $\vec v \in V^\mathbb{Z}$ in the behaviour
  category for a $V$-valued tape, we produce the object
  $f^\ast(\vec v) \in V^\mathbb{N}$ given by
  $f^\ast(\vec v)_n = q(v_n)$;
\item On morphisms, if we are given $\vec v \in V^\mathbb{Z}$ and a
  map $i \colon f^\ast(\vec v) \rightarrow W$ in the behaviour
  category for $V$-valued reversible input then our cofunctor lifts
  this to the morphism $i \colon \vec v \rightarrow \vec w$ in the
  behaviour category of $V$-valued tape where
  \begin{equation*}
    w_k =
    \begin{cases}
      v_{k+i} & \text{ if $k < -i$ or $k > N$;}\\
      p(v_{k+i}) & \text{ if $-i \leqslant k < 0$;}\\
      \spn{p(v_{k+i}), W_k} & \text{ if $0 \leqslant k$.}
    \end{cases}
  \end{equation*}
\end{itemize}

\section{Examples and applications: costructure}
\label{sec:costr-exampl}

In this final section, we illustrate our understanding of the
costructure functor by providing some sample calculations of the
behaviour categories associated to comonads on $\cat{Set}$.

\subsection{Coalgebras for polynomial endofunctors}
\label{sec:coalg-endof}

For any endofunctor $F \colon \cat{Set} \rightarrow \cat{Set}$, we can
consider the category $F\text-\mathrm{coalg}$ of \emph{$F$-coalgebras}, i.e., sets $X$ endowed
with a map $\xi \colon X \rightarrow FX$. As is well-known, for
suitable choices of $F$, such coalgebras can model diverse kinds of
automaton and transition system; see~\cite{Rutten2000Universal} for an
overview.

If $F$ is accessible, then the forgetful functor
$F\text-\mathrm{coalg} \rightarrow \cat{Set}$ will have a right
adjoint and be strictly comonadic, meaning that we can identify
$F\text-\mathrm{coalg}$ with the category of Eilenberg--Moore
coalgebras of the induced comonad $\mathsf{Q}_F$ on $\cat{Set}$; in
light of this, we call $\mathsf{Q}_F$ the \emph{cofree comonad} on the
endofunctor $F$. Explicitly, the values of the cofree comonad can be
described via the greatest fixpoint formula
\begin{equation}
  Q_F(V) = \nu X.\, V \times F(X)\rlap{ .}\label{eq:35}
\end{equation}

The objective of this section is to calculate the behaviour categories
of cofree comonads $\mathsf{Q}_F$ for some natural choices of $F$. To
begin with, let us assume that $F$ is \emph{polynomial} in the sense
of Section~\ref{sec:presheaf-comonads}, meaning that it is can be
written as a coproduct of representables
\begin{equation}\label{eq:41}
  F(X) = \textstyle\sum_{\sigma \in \Sigma} X^{\abs{\sigma}}
\end{equation}
for some set $\Sigma$ and family of sets
$(\abs{\sigma} : \sigma \in \Sigma)$. In this case, as is well-known,
the fixpoint~\eqref{eq:35} can be described as a set of \emph{trees}.

\begin{Defn}[$F$-trees]
  \label{def:21}
  Let $F$ be a polynomial endofunctor~\eqref{eq:41}.
  \begin{itemize}
  \item An \emph{$F$-path of length $k$} is a sequence
    $P = \sigma_0 e_1 \sigma_1 \cdots e_k \sigma_k$ where each
    $\sigma_i \in \Sigma$ and each $e_i \in \abs{\sigma_{i-1}}$.
  \item An \emph{$F$-tree} is a subset $T$ of the set of $F$-paths
    such that:
    \begin{enumerate}[(i)]
    \item $T$ contains a unique path of length $0$, written
      $\ast \in T$;
    \item If $T$ contains
      $\sigma_0 e_1 \cdots e_k \sigma_k e_{k+1}\sigma_{k+1}$, then it
      contains $\sigma_0 e_1 \cdots e_{k}\sigma_{k}$;
    \item If $T$ contains $\sigma_0 e_1 \cdots e_k\sigma_k$, then for
      each $e_{k+1} \in \abs{\sigma_k}$, it contains a unique path of
      the form
      $\sigma_0 e_1 \cdots e_k \sigma_k e_{k+1} \sigma_{k+1}$.
    \end{enumerate}
  \item If $V$ is a set, then a \emph{$V$-labelling} for an $F$-tree
    $T$ is a function $\ell \colon T \rightarrow V$.
  \item If $T$ is an $F$-tree and
    $P = \sigma_0 e_1 \cdots e_k \sigma_k\in T$, then $T_{P}$ is the
    $F$-tree
    \begin{equation*}
      T_{P} = \{\sigma_k e_{k+1} \cdots e_{m} \sigma_m : \sigma_0 e_1
      \cdots e_m \sigma_m \in T\}\rlap{ .}
    \end{equation*}
    If $\ell \colon T \rightarrow V$ is a labelling for $T$, then
    $\ell_P \colon T_P \rightarrow V$ is the labelling with
    $\ell_P(\sigma_k e_{k+1} \cdots e_{m} \sigma_m) = \ell(\sigma_0 e_1
    \cdots e_m \sigma_m)$.
  \end{itemize}
\end{Defn}

\begin{Lemma}
  \label{lem:6}
  The cofree comonad $\mathsf{Q}_F$ on a polynomial $F$ as
  in~\eqref{eq:41} is given as follows:
  \begin{itemize}
  \item $Q_F(V)$ is the set of $V$-labelled $F$-trees;
  \item The counit $\varepsilon_V \colon Q_F(V) \rightarrow V$ sends
    $(T, \ell)$ to $\ell(\ast) \in V$;
  \item The comultiplication
    $\delta_V \colon Q_F(V) \rightarrow Q_FQ_F(V)$ sends $(T, \ell)$
    to $(T, \ell^\sharp)$, where
    $\ell^\sharp \colon T \rightarrow Q_F(V)$ sends $P$ to
    $(T_P, \ell_P)$.\qed
  \end{itemize}
\end{Lemma}

We may use this result to calculate the behaviour category
$\mathbb{B}_{F}$ of the comonad $\mathsf{Q}_F$. Clearly, objects of
$\mathbb{B}_F$ are elements of $Q_F(1)$, i.e., (unlabelled) $F$-trees.
Morphisms of $\mathbb{B}_F$ with domain $T$ are, by definition,
natural transformations $Q_T \Rightarrow \mathrm{id}$; but the functor
$Q_T$ is visibly isomorphic to the representable functor $(\thg)^T$,
so that by the Yoneda lemma, maps in $\mathbb{B}_F$ with domain $T$
correspond bijectively with elements $P \in T$. Given this, we may
easily read off the remainder of the structure in
Definition~\ref{def:22} to obtain:

\begin{Prop}
  \label{prop:28}
  Let $F$ be a polynomial endofunctor of $\cat{Set}$. The behaviour
  category $\mathbb{B}_F$ of the cofree comonad $\mathsf{Q}_F$ has:
  \begin{itemize}
  \item Objects given by $F$-trees $T$;
  \item Morphisms $P \colon T \rightarrow T'$ are elements $P \in T$
    such that $T_P = T'$;
  \item Identities are given by $1_T = \ast \colon T \rightarrow T$;
  \item Composition is given by
    $(\sigma_k e_{k+1} \cdots e_m \sigma_m) \circ (\sigma_0 e_1 \cdots
    e_{k} \sigma_k) = \sigma_0 e_1 \cdots e_m \sigma_m$.
  \end{itemize}
\end{Prop}

It is not hard to see that $\mathbb{B}_F$ is, in fact, the free
category on a graph: the generating morphisms are those of the form
$\sigma_0 e_1 \sigma_1$.

\begin{Rk}
  \label{rk:5}
  When $F$ is polynomial, the cofree comonad $Q_F$ is again
  polynomial: indeed, we have
  $Q_F(V) \cong \sum_{T \in F\text-\mathrm{tree}} V^T$. Thus
  $\mathsf{Q}_F$ is a presheaf comonad, and it will follow from
  Proposition~\ref{prop:27} below that it is in fact the presheaf
  comonad of the behaviour category $\mathbb{B}_F$. Thus, we arrive at
  the (not entirely obvious) conclusion that, for $F$ polynomial, the
  category of $F$-coalgebras is equivalent to the presheaf category
  $[\mathbb{B}_F, \cat{Set}]$.
\end{Rk}

\begin{Ex}
  \label{ex:2}
  Let $E$ be an alphabet. A \emph{deterministic automaton over $E$} is
  a set $S$ of states together with a function
  $(t,h) \colon S \rightarrow S^E \times \{\top, \bot\}$. For a state
  $s \in S$, the value $h(s)$ indicates whether or not $h$ is an
  accepting state; while $t(s)(e) \in S$ gives the state reached from
  $s$ by transition along $e \in E$.

  Deterministic automata are $F$-coalgebras for the polynomial functor
  $F(X) = \sum_{a \in \{\bot, \top\}} X^E$. It is easy to see that, in
  this case, the set of $F$-trees can be identified with the power-set
  $\P(E^\ast)$ via the assignment
  \begin{equation*}
    T \mapsto \{ e_1 \cdots e_n \in E^\ast : \sigma_0 e_1
    \dots
    \sigma_{n-1} e_n \top \in T\}\rlap{ .}
  \end{equation*}
  In these terms, the behaviour category $\mathbb{B}_F$ can be
  identified with the free category on the graph whose vertices are
  subsets of $E^\ast$, and whose edges are
  $e \colon L \rightarrow \partial_e L$ for each $L \subseteq E^\ast$
  $e \in E$, where
  $\partial_e L = \{ e_1 \cdots e_n \in E^\ast : ee_1 \cdots e_n \in L
  \}$. Note that this is precisely the transition graph of the final
  deterministic automaton over $E$.
\end{Ex}

\subsection{Coalgebras for non-polynomial endofunctors}
\label{sec:coalg-non-polyn}

When we consider cofree comonads over \emph{non-}polynomial
endofunctors $F$, things become more delicate.  To illustrate this,
let us consider the case of the \emph{finite multiset} endofunctor
\begin{equation*}
  M(X) = \textstyle\sum_{n \in \mathbb{N}} X^n / \mathfrak{S}_n\rlap{ .}
\end{equation*}

An $F$-coalgebra can be seen as a kind of non-deterministic transition
system. As in the preceding section, we have a description of the
associated cofree comonad in terms of trees:

\begin{Defn}[Symmetric trees]
  \label{def:25}
  A \emph{symmetric tree} $T$ is a diagram of finite sets and
  functions
  \begin{equation*}
    \cdots \xrightarrow{\partial} T_n
    \xrightarrow{\partial} \cdots \xrightarrow{\partial} T_1
    \xrightarrow{\partial} T_0
  \end{equation*}
  where $T_0 = \{\ast\}$. We may write $\abs{T}$ for the set
  $\sum_{k} T_k$. A \emph{$V$-labelling} for a symmetric tree $T$ is a
  function $\ell \colon \abs T \rightarrow V$. Given a $V$-labelled tree
  $(T, \ell)$ and $t \in T_k$, we write $(T_t, \ell_t)$ for the labelled
  tree with $(T_t)_n = \{ u \in T_{n+k} : \partial^k(u) = t\}$, and with
  $\partial$'s and labelling inherited from $T$. An \emph{isomorphism}
  $\theta \colon (T, \ell) \rightarrow (T', \ell')$ of $V$-labelled
  trees is a family of functions $\theta_n \colon T_n \rightarrow T'_n$
  commuting with the $\partial$'s and the maps to $V$.
\end{Defn}

\begin{Lemma}
  \label{lem:7}
  The cofree comonad on the finite multiset endofunctor $M$ has:
  \begin{itemize}
  \item $Q_M(V)$ given by the set of isomorphism-classes of
    $V$-labelled symmetric trees;
  \item The counit $\varepsilon_V \colon Q_M(V) \rightarrow V$ given
    by $(T, \ell) \mapsto \ell(\ast)$;
  \item The comultiplication
    $\delta_V \colon Q_M(V) \rightarrow Q_MQ_M(V)$ given by
    $(T, \ell) \mapsto (T, \ell^\sharp)$, where
    $\ell^\sharp \colon \abs T \rightarrow Q_M(V)$ sends $t \in T_k$
    to $(T_t, \ell_t)$.\qed
  \end{itemize}
\end{Lemma}

Given a symmetric tree $T$, we call $t \in T_k$ \emph{rigid} if any
automorphism of $T$ fixes $t$.

\begin{Prop}
  \label{prop:9}
    The behaviour category of the cofree comonad
  $\mathsf{Q}_M$ has: 
  \begin{itemize}
  \item Objects given by isomorphism-class representatives of
    symmetric trees $T$;
  \item Morphisms $t \colon T \rightarrow T'$ are rigid elements $t \in T$
    such that $T_t \cong T'$;
  \item The identity on $T$ is $\ast \colon T \rightarrow T$;
  \item The composite of $t \colon T \rightarrow T'$ and $u \colon
    T' \rightarrow T''$ is $\theta(u) \colon T \rightarrow T''$, where
    $\theta$ is any tree isomorphism $T' \rightarrow T_t$.
  \end{itemize}
\end{Prop}

\begin{proof}
  Let us write $\mathsf{Q} = \mathsf{Q}_M$. Clearly the object-set
  $Q(1)$ of the behaviour category can be identified with a set of
  isomorphism-class representatives of symmetric trees. Now, for any
  such representative $T$, the subfunctor $Q_T \subseteq Q$ sends a
  set $V$ to the set of all isomorphism-classes of $V$-labellings of
  $T$, which is easily seen to be the quotient
  $V^{\abs T} \quot \mathrm{Aut}(T)$ of $V^{\abs T}$ by the evident
  action of the group of tree automorphisms of $T$. Thus, by the
  Yoneda lemma, the set of natural transformations
  $Q_T \Rightarrow \mathrm{id}$ can be identified with the set of
  elements $t \in \abs{T}$ which are fixed by the $\mathrm{Aut}(T)$
  action, i.e., the rigid elements of $T$. For a given rigid element
  $t \in \abs{T}$, the corresponding $Q_T \Rightarrow \mathrm{id}$
  sends a $V$-labelling $\ell \colon \abs{T} \rightarrow V$ in
  $Q_T(V)$ to $\ell(t) \in V$; and it follows that the unique
  factorisation in~\eqref{eq:34} is of the form
  $Q_T \Rightarrow Q_{T'}$ where $T' \cong T_t$. Tracing through the
  remaining aspects of the definition of behaviour category yields the
  result.
\end{proof}

For a similar example in this vein, we may calculate the behaviour
category of the cofree comonad generated by the finite powerset
functor $P_f$. In this case, things are even more degenerate: the
behaviour category turns out to be the \emph{discrete} category on the
final $P_f$-coalgebra.

\subsection{Local homeomorphisms}
\label{sec:local-homeomorphisms}

For our final example, we compute the behaviour category of the
comonad classifying local homeomorphisms over a topological space.

\begin{Defn}[Reduced power]
  \label{def:28} If $A,X$ are sets and $\F$ is a filter of subsets of
  $X$, then two maps $\varphi, \psi \colon X \rightarrow A$ are
  \emph{$\F$-equivalent} when
  $\{ x \in X : \varphi(x) = \psi(x)\} \in \F$. The \emph{reduced
    power} $A^\F$ is the quotient of $A^X$ by $\F$-equivalence.
\end{Defn}

\begin{Defn}[Sheaf comonad]
  \label{def:26}
  Let $X$ be a topological space. The \emph{sheaf comonad}
  $\mathsf{Q}_{X}$ is the accessible comonad on $\cat{Set}$ induced by
  the adjunction
  \begin{equation}\label{eq:49}
    \cd{
      {\cat{Lh} / X}
      \ar@<-4.5pt>[r]_-{U}
      \ar@{<-}@<4.5pt>[r]^-{C} \ar@{}[r]|-{\top} &
      {\cat{Set} / X} \ar@<-4.5pt>[r]_-{\Sigma} \ar@{<-}@<4.5pt>[r]^-{\Delta} \ar@{}[r]|-{\top} &
      {\cat{Set}} \rlap{ ,}
    }
  \end{equation}
  where $\cat{Lh}/X$ is the category of local homeomorphisms over $X$,
  where $U$ is the evident forgetful functor, and where $C$ sends
  $p \colon A \rightarrow X$ to the space of germs of partial sections
  of $p$. If we write $\N_x$ for the filter of open neighbourhoods of
  $x \in X$, then then this comonad has
  $Q_{X}(A) = \sum_{x \in X} A^{\N_x}$, and counit and comultiplication
  \begin{equation}\label{eq:50}
    \begin{aligned}
      \varepsilon_A \colon \textstyle\sum_{x} A^{\N_x} & \rightarrow A & \quad
      \delta_A \colon \textstyle\sum_{x} A^{\N_x} &
      \rightarrow \textstyle\sum_{x} \bigl(\sum_{x'}
      A^{\N_{x'}}\bigr)^{\N_x}\\
      (x,\varphi) & \mapsto \varphi(x) & (x, \varphi) & \mapsto \bigl(x,
      \lambda y.\, (y, \varphi)\bigr)\rlap{ .}
    \end{aligned}
  \end{equation}
\end{Defn}

The adjunction in~\eqref{eq:49} is in fact \emph{strictly} comonadic,
so that we can identify the category of $\mathsf{Q}_X$-coalgebras with
the category of local homeomorphisms (=sheaves) over $X$.

\begin{Prop}
  \label{prop:29}
  Let $X$ be a topological space. The behaviour category of the sheaf
  comonad $\mathsf{Q}_X$ is the poset $(X, \leqslant)$ of points of
  $X$ under the specialisation order: thus $x \leqslant y$ just when
  every open set containing $x$ also contains $y$.
\end{Prop}

\begin{proof}
  Writing $\mathsf{Q}$ for $\mathsf{Q}_X$, we clearly have $Q(1) = X$,
  so that objects of the behaviour category are points of $X$. To
  characterise the morphisms with domain $x \in X$, we observe that
  the subfunctor $Q_x \subseteq Q$ is the reduced power functor
  $(\thg)^{\N_x}$, which can also be written as the directed colimit
  of representable functors $\mathrm{colim}_{U \in \N_x} (\thg)^U$.
  Thus, by the Yoneda lemma, the set of natural transformations
  $Q_x \Rightarrow \mathrm{id}$ can be identified with the filtered
  intersection $\bigcap_{U \in \N_x} U$, i.e., with the upset
  $\{y \in X : x \leqslant y\}$ of $x$ for the specialisation order.
  Given $y \geqslant x$, the corresponding natural transformation
  $Q_x \Rightarrow \mathrm{id}$ has components
  $\varphi \mapsto \varphi(y)$; whence the unique factorisation
  in~\eqref{eq:34} is of the form $Q_x \Rightarrow Q_{y}$. By
  definition of behaviour category, we conclude that
  $\mathbb{B}_\mathsf{Q}$ is the specialisation poset of $X$.
\end{proof}

What we learn from this is that a local homeomorphism
$p \colon S \rightarrow X$ can be seen as a coalgebraic structure with
set of ``states'' $S$, with the ``behaviour'' associated to a state
$s$ given by $p(s)$, and with the possibility of transitioning
uniformly from a state $s$ of behaviour $x$ to a state $s'$ of
behaviour $y$ whenever $x \leqslant y$. This is intuitively easy to
see: given $s \in S$ of behaviour $x$, we pick an open neighbourhood
$U \subseteq S$ mapping homeomorphically onto an open
$V = p(U) \subseteq X$. Since $x = p(s) \in V$ and $x \leqslant y$,
also $y \in V$ and so we may define $s'$ of behaviour $y$ to be
$(\res p U)^{-1}(y) \in U$.

\bibliographystyle{acm}
\bibliography{bibdata}

\begin{thebibliography}{10}

\bibitem{Abbott2005Containers:}
{\sc Abbott, M., Altenkirch, T., and Ghani, N.}
\newblock Containers: constructing strictly positive types.
\newblock {\em Theoretical Computer Science 342}, 1 (2005), 3--27.

\bibitem{Aguiar1997Internal}
{\sc Aguiar, M.}
\newblock {\em Internal categories and quantum groups}.
\newblock PhD thesis, Cornell University, 1997.

\bibitem{Ahman2020Runners}
{\sc Ahman, D., and Bauer, A.}
\newblock Runners in action.
\newblock In {\em Programming Languages and Systems\/} (2020), vol.~12075 of
  {\em Lecture Notes in Computer Science}, Springer, pp.~29--55.

\bibitem{Ahman2012When}
{\sc Ahman, D., Chapman, J., and Uustalu, T.}
\newblock When is a container a comonad?
\newblock In {\em Foundations of software science and computational
  structures}, vol.~7213 of {\em Lecture Notes in Comput. Sci.} Springer,
  Heidelberg, 2012, pp.~74--88.

\bibitem{Ahman2014Coalgebraic}
{\sc Ahman, D., and Uustalu, T.}
\newblock Coalgebraic update lenses.
\newblock In {\em Proceedings of the 30th {C}onference on the {M}athematical
  {F}oundations of {P}rogramming {S}emantics ({MFPS} {XXX})\/} (2014), vol.~308
  of {\em Electron. Notes Theor. Comput. Sci.}, Elsevier Sci. B. V., Amsterdam,
  pp.~25--48.

\bibitem{Ahman2014Update}
{\sc Ahman, D., and Uustalu, T.}
\newblock Update monads: cointerpreting directed containers.
\newblock In {\em 19th {I}nternational {C}onference on {T}ypes for {P}roofs and
  {P}rograms}, vol.~26 of {\em LIPIcs. Leibniz Int. Proc. Inform.} Schloss
  Dagstuhl. Leibniz-Zent. Inform., Wadern, 2014, pp.~1--23.

\bibitem{Ahman2017Taking}
{\sc Ahman, D., and Uustalu, T.}
\newblock Taking updates seriously.
\newblock In {\em Proceedings of the 6th International Workshop on
  Bidirectional Transformations\/} (2017), CEUR Workshop Proceedings,
  pp.~59--73.

\bibitem{Barr1985Toposes}
{\sc Barr, M., and Wells, C.}
\newblock {\em Toposes, triples and theories}, vol.~278 of {\em Grundlehren der
  Mathematischen Wissenschaften}.
\newblock Springer, 1985.

\bibitem{Diers1978Spectres}
{\sc Diers, Y.}
\newblock Spectres et localisations relatifs {\`a} un foncteur.
\newblock {\em Comptes rendus hebdomadaires des s{\'e}ances de l'Acad{\'e}mie
  des sciences 287\/} (1978), 985--988.

\bibitem{Dubuc1970Kan-extensions}
{\sc Dubuc, E.~J.}
\newblock {\em Kan extensions in enriched category theory}, vol.~145 of {\em
  Lecture Notes in Mathematics}.
\newblock Springer, 1970.

\bibitem{Foster2007Combinators}
{\sc Foster, J.~N., Greenwald, M.~B., Moore, J.~T., Pierce, B.~C., and Schmitt,
  A.}
\newblock Combinators for bidirectional tree transformations: A linguistic
  approach to the view-update problem.
\newblock {\em ACM Transactions on Programming Languages and Systems 29}, 3
  (2007), 17--es.

\bibitem{Goncharov2013Trace}
{\sc Goncharov, S.}
\newblock Trace semantics via generic observations.
\newblock In {\em Algebra and coalgebra in computer science}, vol.~8089 of {\em
  Lecture Notes in Comput. Sci.} Springer, Heidelberg, 2013, pp.~158--174.

\bibitem{Goncharov2014Towards}
{\sc Goncharov, S., Milius, S., and Silva, A.}
\newblock Towards a coalgebraic chomsky hierarchy.
\newblock In {\em {TCS} 2014, Rome\/} (2014), vol.~8705 of {\em Lecture Notes
  in Computer Science}, Springer, pp.~265--280.

\bibitem{Goncharov2020Toward}
{\sc Goncharov, S., Milius, S., and Silva, A.}
\newblock Toward a uniform theory of effectful state machines.
\newblock {\em ACM Transactions on Computational Logic 21}, 3 (2020).

\bibitem{Higgins1993Duality}
{\sc Higgins, P.~J., and Mackenzie, K. C.~H.}
\newblock Duality for base-changing morphisms of vector bundles, modules, {L}ie
  algebroids and {P}oisson structures.
\newblock {\em Mathematical Proceedings of the Cambridge Philosophical Society
  114}, 3 (1993), 471--488.

\bibitem{Johnstone1990Collapsed}
{\sc Johnstone, P.~T.}
\newblock Collapsed toposes and {C}artesian closed varieties.
\newblock {\em Journal of Algebra 129}, 2 (1990), 446--480.

\bibitem{Jonsson1961On-two-properties}
{\sc J\'{o}nsson, B., and Tarski, A.}
\newblock On two properties of free algebras.
\newblock {\em Mathematica Scandinavica 9\/} (1961), 95--101.

\bibitem{Katsumata2019Interaction}
{\sc Katsumata, S., Rivas, E., and Uustalu, T.}
\newblock Interaction laws of monads and comonads.
\newblock Preprint, available as
  \href{https://arxiv.org/abs/1912.13477}{arXiv:1912.13477}, 2019.

\bibitem{Kelly1993Adjunctions}
{\sc Kelly, G.~M., and Power, A.~J.}
\newblock Adjunctions whose counits are coequalizers, and presentations of
  finitary enriched monads.
\newblock {\em Journal of Pure and Applied Algebra 89\/} (1993), 163--179.

\bibitem{Kupke2009Characterising}
{\sc Kupke, C., and Leal, R.~A.}
\newblock Characterising behavioural equivalence: three sides of one coin.
\newblock In {\em Algebra and coalgebra in computer science}, vol.~5728 of {\em
  Lecture Notes in Computer Science}. Springer, 2009, pp.~97--112.

\bibitem{Lawvere1963Functorial}
{\sc Lawvere, F.~W.}
\newblock {\em Functorial semantics of algebraic theories}.
\newblock PhD thesis, Columbia University, 1963.
\newblock Also \emph{Proc. Nat. Acad. Sci. U.S.A. 50} (1963), 869--872.
  Republished as: \textit{Reprints in Theory and Applications of Categories 5}
  (2004).

\bibitem{Manes1976Algebraic}
{\sc Manes, E.}
\newblock {\em Algebraic theories}, vol.~26 of {\em Graduate Texts in
  Mathematics}.
\newblock Springer, 1976.

\bibitem{Mogelberg2014Linear}
{\sc M{\o}gelberg, R.~E., and Staton, S.}
\newblock Linear usage of state.
\newblock {\em Logical Methods in Computer Science 10}, 1 (2014), 1:17, 52.

\bibitem{Moggi1991Notions}
{\sc Moggi, E.}
\newblock Notions of computation and monads.
\newblock {\em Information and Computation 93\/} (1991), 55--92.

\bibitem{Pattinson2015Sound}
{\sc Pattinson, D., and Schr\"{o}der, L.}
\newblock Sound and complete equational reasoning over comodels.
\newblock In {\em The 31st {C}onference on the {M}athematical {F}oundations of
  {P}rogramming {S}emantics ({MFPS} {XXXI})}, vol.~319 of {\em Electronic Notes
  in Theoretical Computer Science}. Elsevier, 2015, pp.~315--331.

\bibitem{Pattinson2016Program}
{\sc Pattinson, D., and Schr\"{o}der, L.}
\newblock Program equivalence is coinductive.
\newblock In {\em Proceedings of the 31st {A}nnual {ACM}-{IEEE} {S}ymposium on
  {L}ogic in {C}omputer {S}cience ({LICS} 2016)\/} (2016), ACM, p.~10.

\bibitem{Plotkin2001Adequacy}
{\sc Plotkin, G., and Power, J.}
\newblock Adequacy for algebraic effects.
\newblock In {\em Foundations of software science and computation structures
  ({G}enova, 2001)}, vol.~2030 of {\em Lecture Notes in Comput. Sci.} Springer,
  Berlin, 2001, pp.~1--24.

\bibitem{Plotkin2002Notions}
{\sc Plotkin, G., and Power, J.}
\newblock Notions of computation determine monads.
\newblock In {\em Foundations of software science and computation structures
  ({G}renoble, 2002)}, vol.~2303 of {\em Lecture Notes in Computer Science}.
  Springer, Berlin, 2002, pp.~342--356.

\bibitem{Plotkin2003Algebraic}
{\sc Plotkin, G., and Power, J.}
\newblock Algebraic operations and generic effects.
\newblock vol.~11. 2003, pp.~69--94.
\newblock On formal description of computations---refined structures and new
  trends (Cork, 2000).

\bibitem{Plotkin2008Tensors}
{\sc Plotkin, G., and Power, J.}
\newblock Tensors of comodels and models for operational semantics.
\newblock {\em Electronic Notes in Theoretical Computer Science 218\/} (2008),
  295--311.

\bibitem{Power2004From}
{\sc Power, J., and Shkaravska, O.}
\newblock From comodels to coalgebras: state and arrays.
\newblock In {\em Proceedings of the {W}orkshop on {C}oalgebraic {M}ethods in
  {C}omputer {S}cience\/} (2004), vol.~106 of {\em Electron. Notes Theor.
  Comput. Sci.}, Elsevier Sci. B. V., Amsterdam, pp.~297--314.

\bibitem{Renault1980A-groupoid}
{\sc Renault, J.}
\newblock {\em A groupoid approach to {$C^{\ast} $}-algebras}, vol.~793 of {\em
  Lecture Notes in Mathematics}.
\newblock Springer, Berlin, 1980.

\bibitem{Rutten2000Universal}
{\sc Rutten, J. J. M.~M.}
\newblock Universal coalgebra: a theory of systems.
\newblock vol.~249. 2000, pp.~3--80.
\newblock Modern algebra and its applications (Nashville, TN, 1996).

\bibitem{Rutten1993On-the-foundations}
{\sc Rutten, J. J. M.~M., and Turi, D.}
\newblock On the foundations of final semantics: nonstandard sets, metric
  spaces, partial orders.
\newblock In {\em Semantics: foundations and applications ({B}eekbergen,
  1992)}, vol.~666 of {\em Lecture Notes in Comput. Sci.} Springer, Berlin,
  1993, pp.~477--530.

\bibitem{Thielecke1997Categorical}
{\sc Thielecke, H.}
\newblock {\em Categorical structure of continuation passing style}.
\newblock PhD thesis, University of Edinburgh, 1997.

\bibitem{Uustalu2015Stateful}
{\sc Uustalu, T.}
\newblock Stateful runners of effectful computations.
\newblock In {\em The 31st {C}onference on the {M}athematical {F}oundations of
  {P}rogramming {S}emantics ({MFPS} {XXXI})}, vol.~319 of {\em Electron. Notes
  Theor. Comput. Sci.} Elsevier Sci. B. V., Amsterdam, 2015, pp.~403--421.

\end{thebibliography}
\end{document}